\documentclass[hidelinks,11pt]{article}
\usepackage{mymacros}
\usepackage{hyperref}
\usepackage[normalem]{ulem}
\usepackage{enumitem}

\usepackage{etoolbox}

\usepackage{orcidlink}

\newtheorem{innercustomgeneric}{\customgenericname}
\providecommand{\customgenericname}{}
\newcommand{\newcustomtheorem}[2]{%
  \newenvironment{#1}[1]
  {%
   \renewcommand\customgenericname{#2}%
   \renewcommand\theinnercustomgeneric{##1}%
   \innercustomgeneric
  }
  {\endinnercustomgeneric}
}
\newcustomtheorem{customthm}{Theorem}

\usepackage[T1]{fontenc}
\usepackage{lmodern}

\usepackage{graphicx}
\usepackage{subcaption}

\usepackage{bbm} 

\DeclareMathAlphabet{\mathpzc}{OT1}{pzc}{m}{it}

\definecolor{darkmagenta}{rgb}{0.55, 0.0, 0.55}

\definecolor{magenta}{rgb}{0.8, 0.0, 0.8}

\newcommand{\1}{\mathbbm{1}}

\newcommand{\pt}{\operatorname{pt}}
\newcommand{\one}{{\1}}
\newcommand{\eps}{\varepsilon}
\usepackage[titles]{tocloft}
\newcommand{\bulk}{\varnothing}
\newcommand{\0}{0}
\renewcommand{\o}{o}
\newcommand{\x}{x}
\newcommand{\y}{y}
\newcommand{\z}{z}

\newcommand{\ox}{{ox}}
\newcommand{\G}{G_{\nu_c,\infty}}
\newcommand{\vv}{\alpha_\infty}

\makeatletter
\newsavebox\myboxA
\newsavebox\myboxB
\newlength\mylenA

\newcounter{savedenum} 

\newcommand{\resumeenumerate}{%
    \setcounter{enumi}{\value{savedenum}} 
}

\newcommand*\xoverline[2][0.60]{%
    \sbox{\myboxA}{$\m@th#2$}%
    \setbox\myboxB\null
    \ht\myboxB=\ht\myboxA%
    \dp\myboxB=\dp\myboxA%
    \wd\myboxB=#1\wd\myboxA
    \sbox\myboxB{$\m@th\overline{\copy\myboxB}$}
    \setlength\mylenA{\the\wd\myboxA}
    \addtolength\mylenA{-\the\wd\myboxB}%
    \ifdim\wd\myboxB<\wd\myboxA%
       \rlap{\hskip 0.5\mylenA\usebox\myboxB}{\usebox\myboxA}%
    \else
        \hskip -0.5\mylenA\rlap{\usebox\myboxA}{\hskip 0.5\mylenA\usebox\myboxB}%
    \fi}
\makeatother

\newcommand{\Tay}{\operatorname{Tay}}
\newcommand{\cst}{c^{\rm st}}
\renewcommand{\bar}[1]{\xoverline{#1}}

\renewcommand{\geq}{\ge}
\renewcommand{\leq}{\le}

\newcommand{\rg}{\operatorname{RG}}

\newcommand{\Loc}{\operatorname{Loc}}
\newcommand{\Eplus}{\E_+}

\newcommand{\uf}{f}

\newcommand{\bbb}{B}
\newcommand{\hh}{\mathpzc{h}}
\newcommand{\lp}{\mathpzc{l}}

\newcommand{\domRG}{\mathbb{D}}
\newcommand{\Wkappa}{\cW^{\kappa}}

\newcommand{\blocks}{\cB}

\newcommand{\free}{{\rm F}}
\newcommand{\crit}{{\rm crit}}
\newcommand{\per}{{\rm P}}
\newcommand{\II}{\mathbb{I}}

\newcommand{\ka}{a}
\newcommand{\kaa}{\mathfrak{a}}
\newcommand{\kb}{\mathfrak{b}}

\newcommand{\kp}{\mathfrak{p}}
\newcommand{\kh}{\mathfrak{h}}

\newcommand{\qLap}{q}
\newcommand{\scale}{\rho}

\newcommand{\Npp}{N}

\newcommand{\joxm}{j_{{\ox_ -}}}

\newcommand{\HLap}{\mathbb{C}} 

\newcommand{\constF}{c^{\rm F}}
\newcommand{\const}{{\rm const}}

\newcommand{\assumPhi}{(\textnormal{A}_{\Phi})}
\makeatletter
\newcommand{\customlabel}[2]{%
   \protected@write \@auxout {}{\string \newlabel {#1}{{#2}{\thepage}{#2}{#1}{}} }%
   \hypertarget{#1}{#2}
}
\makeatother

\title{Boundary conditions and the two-point function plateau for
\\
the hierarchical $|\varphi|^4$ model in dimensions 4 and higher}

\author{
  Jiwoon Park\,\orcidlink{0000-0002-1159-2676}%
  \thanks{Department of Mathematics,
  Republic of Korea Air Force Academy,
  635, Danjae-ro, Cheongju-si, Chungcheongbuk-do, Republic of Korea.
  {\tt jp711@cantab.ac.uk}
  }
  \and
  Gordon Slade\,\orcidlink{0000-0001-9389-9497}%
  \thanks{Department of Mathematics,
     University of British Columbia,
     Vancouver BC, Canada V6T 1Z2.
       {\tt slade@math.ubc.ca}}
}

\date{}
\vspace{-5ex}

\begin{document}

\maketitle

\begin{abstract}
We obtain precise plateau estimates for the two-point function of the finite-volume weakly-coupled hierarchical
$|\varphi|^4$ model in dimensions $d \ge 4$, for both free and periodic boundary conditions,
and for any number $n \ge 1$ of components of the field $\varphi$. We prove that, within a critical window around their respective effective critical points,
the two-point functions for both free and periodic boundary conditions have a plateau, in the sense that they decay as $|x|^{-(d-2)}$ until reaching a constant plateau value of order $V^{-1/2}$ (with a logarithmic correction
for $d=4$), where $V$ is the size of the finite volume.
The two critical windows for free and periodic boundary conditions do not overlap.
The dependence of the plateau height on the location within the critical window
is governed by an explicit $n$-dependent universal profile which is independent
of the dimension.
The proof is based on a rigorous renormalisation group method
and extends the method used by Michta, Park and Slade (arXiv:2306.00896)
to study  the
finite-volume susceptibility
and related quantities.
Our results lead to precise conjectures concerning Euclidean
(non-hierarchical) models
of spin systems and self-avoiding walk in dimensions $d \ge 4$.
\end{abstract}

%

\section{Introduction and main results}
\label{sec:intro_main_results}

\subsection{Introduction}
\label{sec:intro}

The $|\varphi|^4$ model on the Euclidean lattice $\Z^d$ is
one of the most fundamental
spin models in statistical mechanics and Euclidean quantum field theory (see, e.g.,\cite{GJ87,FFS92}).
For $n \in \N$, $g>0$, $\nu \in \R$, a finite subset $\Lambda \subset \Z^d$,
and a spin field $\varphi : \Lambda \to \R^n$, the model's \emph{Hamiltonian}
is defined by
\begin{align}
\label{eq:Hamiltonian}
    H (\varphi) =
    \sum_{x\in \Lambda}\Big(
    \frac 12 \varphi_x \cdot (-\Delta \varphi)_x
    +
    \frac{1}{4} g |\varphi_x|^4
    + \frac{1}{2} \nu  |\varphi_x|^2   \Big),
\end{align}
with $\Delta$ the discrete Laplace operator.
The theory's critical finite-size scaling for $d\ge 4$ and its dependence on the boundary condition have
been extensively studied in the physics literature.
The two common choices of boundary condition are periodic (PBC) and free (FBC).
PBC corresponds to taking $\Lambda$ to be a discrete torus,
while FBC means that spin-spin interactions occur
only inside $\Lambda$.

In particular, it has been observed that
at and above the upper critical dimension
$4$, with PBC the finite-volume critical
two-point function $\langle \varphi_0\cdot \varphi_x  \rangle$
has a \emph{plateau}: it has the Gaussian $|x|^{-(d-2)}$ decay of the infinite-volume
two-point function for small $x$, but for larger $x$ it levels off to a constant plateau
value.
Finite-size scaling and the plateau phenomenon are discussed, e.g., in
\cite{BEHK22,DGGZ24,F-SBKW16,Kenna04,LB97,LM16,WY14,ZGFDG18,FDZ21},
and there has been some debate about the plateau.  Early numerical evidence for
the existence of a plateau in spin systems can be seen in \cite[Figure~4]{LB97}.

Recently, the plateau phenomenon has also been studied in the
mathematical literature.
For the Ising model on $\Z^d$ in dimensions $d>4$ at the infinite-volume critical
point, with FBC the absence of the plateau has been proved
in \cite{CJN21} (with \cite{DP24,Saka07} for verification of
a hypothesis), whereas the presence of a plateau with PBC is
proved in \cite{LPS25-Ising}.
For (spread-out) percolation in dimensions $d>6$ with PBC, a plateau
is proven to exist
throughout the critical window in \cite{HMS23}, whereas with
FBC
the absence of the plateau at the infinite-volume critical point is proved in \cite{CH20}.
For PBC, a plateau is proven to exist for self-avoiding walk in dimensions $d>4$ in
\cite{Slad23_wsaw,Liu24}, for simple random walk in dimensions $d>2$ in \cite{Slad23_wsaw,DGGZ24},
and for (spread-out) lattice trees and lattice animals in dimensions $d>8$
in \cite{LS25a}.  An overview is given in \cite{LPS25-universal}.

\emph{Hierarchical} models provide a simplified context for studying critical
phenomena, and there is an extensive literature on statistical mechanical
models on hierarchical lattices,
e.g., \cite{ACG13,BM87,CE78,GK82,HHW01,BBS-brief,Hutc23,BI03d}.
A detailed  account of the finite-size scaling
of the weakly-coupled $|\varphi|^4$ model  on the hierarchical lattice
in dimensions $d \ge 4$
was presented in \cite{MPS23}, via an extension and improvement of
the rigorous renormalisation group method of  \cite{BBS-brief}.
The results of \cite{MPS23} focussed on the universal scaling of the total field and its
moments, in the vicinity of effective finite-volume critical points which depend on the
choice of boundary condition.

The results of \cite{MPS23} do not include the critical finite-size scaling of the
hierarchical two-point function $\langle \varphi_0^{(1)}\varphi_x^{(1)} \rangle$,
where $\varphi_x^{(1)}$ denotes the first component of $\varphi_x \in \R^n$.
Our purpose in this paper is to elucidate the plateau behaviour
of the two-point function in dimensions $d \ge 4$.
The main technical innovation is to augment the renormalisation group
flow used in \cite{MPS23} with \emph{observable fields} in order to study
the two-point function.  Observable fields have been used in our context
multiple times in the past, and we are inspired by the approach used in \cite{BBS-saw4}.
(Different observable fields were used to obtain all multi-point correlation
functions for a certain hierarchical model in \cite{ACG13}.)
Our results for the hierarchical model (defined in Section~\ref{sec:model})
have the following features:
\begin{itemize}
\item
We include $n$-component models for all $n \ge 1$, and include the upper critical
dimension $d=4$ (with its logarithmic corrections) as well as dimensions $d>4$.
\item
We prove that the amplitude of the plateau has a universal profile in
non-overlapping critical windows
around distinct effective critical points for FBC and PBC, with the
\emph{same} profile for both FBC and PBC.  The profile depends on $n$ but not on $d$.
\item
We explain why it is the case that, at the infinite-volume critical point, there
is a plateau with PBC but not for FBC:
the free boundary condition generates an effective
mass in the Hamiltonian.
\end{itemize}

We now summarise our results in more detail.
Given
$n \ge 1$ and $s \in \R$, we define the \emph{universal profile}
$\uf_n:\R \to (0,\infty)$ by
\begin{equation}
    \uf_n(s)
    =
    \frac{\int_{\R^n} |x|^2 e^{-\frac 14 |x|^4 - \frac s2 |x|^2} dx}
    {n\int_{\R^n} e^{-\frac 14 |x|^4 - \frac s2 |x|^2} dx}
    ,
\end{equation}
where $|x|$ denotes the $\ell_2$ norm of $x\in \R^n$.
We consider the weakly-coupled hierarchical model
in a finite volume of size $L^{dN}$ in dimensions $d \ge 4$, with $L$ fixed
and in the limit of large $N$.
We denote expectation with fixed $g$ and variable $\nu$ as $\langle \cdot\rangle_{\nu,N}$ and the infinite-volume
critical point as $\nu_c$ (it depends on $d$, $g$ and $L$).
In our formulation, the infinite hierarchical lattice is an orthant of $\Z^d$.
We denote points in the orthant as $x$, and $|x|$ (again) denotes the
$\ell_2$ norm.

\medskip\noindent
We prove that for $d>4$ and for all $n \ge 1$:
\begin{itemize}
\item
For PBC,
with a window scale $w_N$ of order $L^{-Nd/2}$,
\begin{equation}
\label{eq:PBC-intro}
    \langle \varphi_\o^{(1)}\varphi_\x^{(1)} \rangle_{\nu_c+sw_N,N}
    = \Big( \HLap_{0,\infty}(\x) + c_d\uf_n(s) L^{-Nd/2} \Big)[1+o(1)]
\end{equation}
as $N\rightarrow \infty$,
where the massless hierarchical Green function
$\HLap_{0,\infty}(\x)$
is comparable to  $|x|^{-(d-2)}$,
and $o(1)$ denotes an error term which is small when $N$ and $x$ are large (depending on $s$).
This is valid for all $s \in (-\infty,+\infty)$, so both above and below
the infinite-volume critical point within a critical window whose width
is the reciprocal of the square root of the volume.  The Green function
$\HLap_{0,\infty}(\x)$ dominates the right-hand
side of \eqref{eq:PBC-intro} when $|x|^{d-2}$ is smaller than the square root
of the volume, but for larger $x$ it is the plateau term $c_d\uf_n(s) L^{-Nd/2}$
that is dominant.
\item
For FBC, \eqref{eq:PBC-intro} also holds for FBC after $\nu_c$ is replaced by
an effective critical point $\nu_{c,N}^\free$ with $\nu_{c,N}^\free = \nu_c - O(L^{-2N})$.
In particular, the \emph{same} profile $\uf_n$ occurs.
The windows of width $L^{-Nd/2}$ around the PBC and FBC effective
critical points $\nu_c$ and $\nu_{c,N}^\free$ are separated by $L^{-2N}$ and do not overlap.
\item
For FBC there is no plateau for the two-point function at the infinite-volume critical
point $\nu_c$.  Instead, only the $\HLap_{0,\infty}(\x)$ term occurs.
\end{itemize}

\medskip\noindent
For $d=4$, a similar picture holds with logarithmic corrections (polynomial in $N$),
with the same profile $\uf_n$.
The logarithmic corrections involve exponents defined by
\begin{equation}
\label{eq:hatexponents}
    \hat\gamma     = \frac{n+2}{n+8},
    \qquad
    \hat\theta  = \frac 12 -\hat\gamma  = \frac{4-n}{2(n+8)}.
\end{equation}
We prove that for $d=4$ and for all $n \ge 1$:
\begin{itemize}
\item
For PBC,
with a window
scale $w_N$ of order $N^{-\hat\theta}L^{-2N}$,
and for all $s \in (-\infty,+\infty)$,
\begin{equation}
\label{eq:4PBC-intro}
    \langle \varphi_\o^{(1)}\varphi_\x^{(1)} \rangle_{\nu_c+sw_N,N}
    = \Big( \HLap_{0,\infty}(\x) + c_4\uf_n(s) N^{1/2}L^{-2N} \Big)[1+o(1)].
\end{equation}
\item
For FBC, \eqref{eq:4PBC-intro} also holds for FBC after $\nu_c$ is replaced by
an effective critical point $\nu_{c,N}^\free$ with $\nu_{c,N}^\free = \nu_c - O(N^{\hat\gamma}L^{-2N})$.
Since $\hat{\gamma} > - \hat{\theta}$,
the windows of width $N^{-\hat\theta}L^{-2N}$ around the PBC and FBC effective
critical points $\nu_c$ and $\nu_{c,N}^\free$ do not overlap.
\item
For FBC there is again no plateau for the two-point function at the infinite-volume critical
point $\nu_c$.
\end{itemize}

Although our results are proved only for the hierarchical model,
we will explain below that
they lead to precise conjectures for the behaviour on $\Z^d$
(the Euclidean, non-hierarchical model) for all $d \ge 4$.
In particular, we believe that versions of the above statements hold also for
$\Z^d$ with the same scaling for $w_N$, with the same scaling for the
shift in the effective critical value
$\nu_{c,N}^\free$, and with the same universal profile $\uf_n(s)$.
Moreover, we
believe that all these statements hold for weakly or strictly self-avoiding walk,
when we set $n=0$.
See also \cite[Section~1.6]{MPS23} for further conjectures.

There are extensive results for $\Z^4$ obtained by a renormalisation group method
closely related to the one we employ for the hierarchical lattice, e.g.,
\cite{ST-phi4,BBS-phi4-log,BSTW-clp,BBS-saw4,BBS-saw4-log}.
These methods have the potential to be extended to obtain
our results also for finite volumes in $\Z^d$ in dimensions $d \ge 4$,
with sufficient effort.

\medskip\noindent
Notation:  We write $f_n \sim g_n$ to denote $\lim_n f_n/g_n =1$,
and $f_n \asymp g_n$ to denote the existence of $C>0$ such that $C^{-1}g_n \le f_n \le Cg_n$.

\subsection{The model}
\label{sec:model}

To state our results precisely, we now introduce the hierarchical model,
its Laplacian and Green function, following the presentation in \cite{MPS23}.
We also recall from \cite{BBS-brief,MPS23} a theorem
proving existence of the hierarchical model's critical point in dimensions
$d \ge 4$.

\subsubsection{The hierarchial group}

Given integers $d\ge 1$, $L>1$, and $N \ge 1$, we work in a finite volume defined by
\begin{equation}
    \Lambda_N =\{x \in \Z^d: 0 \le x_i < L^N \; \text{for}\; i=1,\ldots,d\}.
\end{equation}
The \emph{volume} of $\Lambda_N$ is its cardinality, $L^{dN}$.
The infinite-volume case is
\begin{equation}
    \Lambda_\infty = \{x \in \Z^d: 0 \le x_i \; \text{for}\; i=1,\ldots,d\}.
\end{equation}
For any integer $0 \le j\le N$, $\Lambda_N$ can be partitioned into $L^{d(N-j)}$
disjoint $j$-\emph{blocks}, which are each translates of $\Lambda_j$, and which each contain
$L^{dj}$ vertices.  We denote the set of these
$j$-blocks by $\cB_j$.

There is an equivalent representation of $\Lambda_\infty$, as follows.
Let $\Z_L = \Z/L\Z$ denote the cyclic group, and let
\begin{equation}
    \mathbb{H}_\infty
    =
    \{\tilde x = (\tilde x_1,\tilde x_2,\ldots):
    \text{$\tilde x_i \in \Z_L^d$ with only finitely many nonzero $\tilde x_i$}\}.
\end{equation}
This is an abelian group with coordinatewise addition mod-$L$, and we denote
the group addition as $+$.
Let $\mathbb{H}_N$ denote the subgroup of $\mathbb{H}_\infty$
with $\tilde x_i=0$ for $i>N$.
Let $x_j$ denote the representative of $\tilde x_j$ in $\Lambda_1 \subset \Z^d$.
The map $\sigma : \mathbb{H}_\infty \to \Lambda_\infty$ defined
by
\begin{equation}
\label{eq:hier-bijection}
    \sigma(\tilde x) = \sum_{j=1}^\infty L^{j-1} x_j
\end{equation}
is a bijection, where the
addition on the right-hand side of \eqref{eq:hier-bijection} is in $\Z^d$.
By definition, $\sigma$ restricts to a bijection $\sigma_N:\mathbb{H}_N \to \Lambda_N$.
The bijection induces an addition and a group structure on $\Lambda_N$ via
$x \oplus y =
\sigma_N (\sigma_N^{-1} x + \sigma_N^{-1} y)$ (with $+$ the addition on $\mathbb{H}_N$).
This makes $\sigma_N$ and $\sigma$ into group isomorphisms.

For $x,y\in \Lambda_N$,
we define the \emph{coalescence scale} $j_{xy}$ to be the
smallest $j$ such that $x$ and $y$ lie in the same $j$-block.
In terms of $\mathbb{H}_N$, $j_{xy}$ is the largest coordinate $i$ such
that $(\sigma_N^{-1}x)_i$ differs from $(\sigma_N^{-1}y)_i$.
From this last observation, we see that the coalescence scale is translation invariant
in the sense that
\begin{equation}
\label{eq:joplus}
    j_{x\oplus z,y\oplus z} = j_{xy} \qquad (x,y,z \in \Lambda_N).
\end{equation}
In particular, with $\o\in\Lambda_N$ the point with all coordinates zero, and
with subtraction in $\Lambda_N$ denoted by $\ominus$,
\begin{equation}
\label{eq:jx-y}
    j_{xy} = j_{\o,y\ominus x} \qquad (x,y \in \Lambda_N).
\end{equation}

\subsubsection{The hierarchial Laplacian}
\label{sec:Lapdef}

We define a probability matrix
$J:\Lambda_\infty \times \Lambda_\infty \to [0,1]$ by $J(\x,\x)=0$ and
\begin{equation}
\label{eq:Jxy}
    J(\x,\y)
    = \frac 1z \frac{1}{L^{(d+2)j_{\x\y}}}
    \qquad
    (\x,\y\in\Lambda_\infty, \; \x \neq \y),
\end{equation}
where $z=\frac{1-L^{-d}}{L^2-1}$ is chosen so that $\sum_{\x \neq \o}J(\o,\x)=1$.  By definition,
$J(\x,\y)=J(\o,\x\ominus \y)$.
For finite volume, we define FBC and PBC versions of $J$ by
\begin{alignat}{3}
    J^\free(\x,\y)
    &=
    \frac{1}{z} L^{-(d+2)j_{xy}}
    & \qquad &
    (\x,\y\in\Lambda_N,\; \x \neq \y),
    \\
    J^\per(\x,\y)
	&= \frac{1}{z} L^{-(d+2)j_{\x\y}}
    + {L^{-(d+2)N}}
    & \qquad &
    (\x,\y\in\Lambda_N,\;\x \neq \y).
\end{alignat}
As shown in \cite[Section~1.2.2]{MPS23}, $J^\per(x,y)$ is equal to the sum of $J(x,y')$ over periodic copies $y'$ of $y$, in $N$-blocks other than $\Lambda_N$.
For FBC, the quadratic form
$\sum_{x,y\in \Lambda_\infty} \varphi_x \cdot (J^\free (x,y)\varphi_y)$
equals $\sum_{x,y\in \Lambda_N} \varphi_x \cdot (J (x,y)\varphi_y)$,
so has the effect of setting the field outside $\Lambda_N$ to equal zero.
Thus, FBC here is the same as zero Dirichlet BC.

We set
\begin{equation}
\label{eq:qLapdef}
    \qLap = \frac{1-L^{-d}}{1-L^{-(d+2)}}
\end{equation}
and define the Laplacians $\Delta^*$
on $\Lambda_N$, with PBC and FBC, by
\begin{alignat}{2}
    -\Delta^\per(\x,\y) & = q(\delta_{\ox}-J^\per(\x,\y))
    \qquad
    & (\x,\y\in\Lambda_N),
    \\
    -\Delta^\free(\x,\y) & = q(\delta_{\ox}-J^\free(\x,\y))
    =  -\Delta^\per(\x,\y) + qL^{-(d+2)N}
    \qquad
    & (\x,\y\in\Lambda_N).
\label{eq:PBCLapdef}
\end{alignat}
The factor $\qLap$ has been introduced so that
the PBC Laplacian $\Delta_N^\per$ is identical to the hierarchical Laplacian defined in \cite[(4.1.8)]{BBS-brief}.  It allows for importation of results from \cite{BBS-brief}, but otherwise its precise value is not significant.
When  $L$ is large, $q$ is close to $1$.

Given $x\in \Lambda_N$ and $0\le j \le N$,
let $B_j (x)$ denote the unique $j$-block that contains $x$.
We then define the matrices of symmetric operators $Q_j$ and $P_j$ on $\ell^{2} (\Lambda_N)$ by
\begin{align}
    \label{eq:Qj-def}
    Q_{j}(\x,\y)
    &=
    \begin{cases}
    L^{-dj} & B_j(x)=B_j(y) \\
    0 & B_j(x) \neq B_j(y)
    \end{cases}
    \qquad
    (j=0,1,\ldots,N),
\end{align}
and
\begin{align}
    \label{eq:Pj-def}
    P_{j}
    &=
    Q_{j-1} - Q_{j}
    \qquad
    (j = 1,\dots ,N).
\end{align}
By definition,
\begin{equation}
\label{eq:Psum0}
    \sum_{x\in\Lambda_N}P_{j}(\o,\x)=0 \qquad (j = 1,\dots ,N).
\end{equation}
It is shown in \cite[Lemma~4.1.5]{BBS-brief} that the operators $P_{1},\dots ,P_{N}, Q_{N}$ are
orthogonal projections whose ranges are disjoint and provide a direct
sum decomposition of $\ell^2(\Lambda_N)$.
It follows by definition that $P_{j}(\x,\y)=0$ if $x$ and $y$ are in different $j$-blocks.
By \cite[(4.1.7)]{BBS-brief},  we obtain
\begin{equation}
    -\Delta_{N}^{\per} = \sum_{j=1}^N L^{-2(j-1)}P_j
\end{equation}
and
\begin{equation}
    -\Delta_{N}^{\rm F} =
    \sum_{j=1}^N L^{-2(j-1)}P_j+ qL^{-2 N}Q_N
	.
\end{equation}

\subsubsection{The hierarchical Green function}
\label{sec:C}

For $a>-L^{-2(N-1)}$ and $j \le N$, we define
\begin{align}
\label{eq:gamma_j}
	\gamma_j(\ka)
	&=\frac{L^{2 (j-1)}}{1+\ka L^{2(j-1)}}
\end{align}
and the symmetric, real, and positive semi-definite matrices
\begin{align}
\label{eq:Cjdef}
	C_{\ka,j} &= \gamma_j(\ka)P_j, \qquad
	C_{\ka,\leq N}(\ka) = \sum_{j=1}^N C_{\ka,j} .
\end{align}
For $a\neq 0$ and $a \neq -\qLap L^{-2N}$,
respectively,
we also define
\begin{align}
\label{eq:ChatN}
	C^*_{\ka,\hat N}  &=
	Q_N \times
	\begin{cases}
	\ka^{-1}  \quad & (* = \per) \\
	(\ka + qL^{-2 N})^{-1} &(* = \free).
	\end{cases}
\end{align}
With $\HLap_{\ka,N}^*$ denoting the resolvent of $-\Delta_N^*$,
the decompositions
\begin{equation}
\label{eq:Csum}
    \HLap_{\ka,N}^*  =
 (-\Delta_N^{*}+\ka)^{-1}
    = C_{\ka,\le N} + C^*_{\ka,\hat N} \quad (* \in \{\per,\free\})
\end{equation}
follow as in \cite[Proposition~4.1.9]{BBS-brief}.
We refer to the variable $\ka$ as the mass (squared), even though we permit it to be negative.
The shifted mass in the $Q_N$ term for FBC, compared to PBC,
is ultimately the source of the shift away from $\nu_c$ for
the finite-volume effective critical point for FBC.
Since $C_{\ka,j}(\x,\y)=C_{\ka,j}(o,y\ominus x)$, we write simply
$C_{\ka,j}(\x)$ instead of $C_{\ka,j}(\o,\x)$,
and similarly for $\HLap_{\ka,N}^*$.

For infinite volume,
the massless Laplacian is $-\Delta_\infty =\sum_{j=1}^\infty L^{-2(j-1)} P_j$.
The hierarchical Green function $\HLap_{\ka,\infty}$, defined to be the inverse of
$-\Delta_\infty +\ka$, is  given by
\begin{equation}
    \HLap_{\ka,\infty} =
    \sum_{j=1}^\infty \gamma_j(\ka) P_j
    =
    \sum_{j=1}^\infty C_{\ka,j}
    \qquad ( a \ge 0).
\end{equation}
We restrict now to dimensions $d>2$.
According to \cite[(4.1.29)]{BBS-brief}, the hierarchical and Euclidean
massless Green functions
are comparable in the sense that
\begin{equation}
\label{eq:HLap0-asy}
    \HLap_{0,\infty}(\x) \asymp |x|^{-(d-2)},
\end{equation}
where the norm on the right-hand side is the Euclidean norm of $x\in\Lambda_\infty
\subset \Z^d$.

By \eqref{eq:Psum0},
\begin{equation}
\label{eq:Csum0}
    \sum_{\x\in\Lambda_N}C_{\ka,j}(\x) =0 \quad (j = 1,\dots ,N).
\end{equation}
Therefore, for $\ka >0$, the finite-volume susceptibility is given by
\begin{equation}
    \chi_N^{\per,0}(\ka) = \sum_{\x \in \Lambda_N} \HLap^\per_{\ka,N}(\x) = \frac 1a.
\end{equation}
Thus, by definition,
\begin{equation}
\label{eq:CPa}
    \HLap^\per_{\ka,N}(\x) =  C_{\ka,\le N}(\x) + \frac{1}{\ka L^{dN}}
    = C_{\ka,\le N}(\x) + \frac{\chi^{\per,0} (\ka)}{L^{dN}}.
\end{equation}
If $\ka$ is small enough, e.g., $\ka=tL^{-(2+\delta)N}$
with $t,\delta >0$, then
(as we show in Lemma~\ref{lem:Capp}),
\begin{equation}
\label{eq:Cpa-plateau}
    \HLap^\per_{\ka,N}(\x) \asymp \frac{1}{|x|^{d-2}} + \frac{L^{\delta N}}{tL^{(d-2)N}}.
\end{equation}
As soon as $|x/L^N|^{d-2} > tL^{-\delta N}$, the constant term dominates.
This \emph{plateau} is a harbinger
of the plateau that we exhibit for the $|\varphi|^4$ model.
The plateau term $\chi^{\per,0} (\ka)/L^{dN}$ for the hierarchical Green function is
explicit and straightforward.
Its counterpart for the (Euclidean) lattice
Green function on $\Z^d$ is shown in \cite[Proposition~3.3(ii)]{DGGZ24}
to be asymptotically equivalent to the susceptibility divided by the volume,
rather than equality as in the simpler hierarchical setting.

\subsubsection{The hierarchical $|\varphi|^4$ model}

Given a choice of boundary condition $* \in \{\per,\free\}$,
$n \in \N$, $g>0$, $\nu \in \R$, and a spin field
$\varphi : \Lambda_N \to \R^n$, we define the hierarchical \emph{Hamiltonian} by
\begin{align}
\label{eq:Hamiltonian_def}
    H_{\nu,N}^* (\varphi)
    = \sum_{x\in \Lambda_N}
    \Big(
    \frac 12 \varphi_x \cdot (-\Delta_N^* \varphi)_x
    +
     \frac{1}{4} g  |\varphi_x|^4
     +
     \frac{1}{2} \nu  |\varphi_x|^2  \Big).
\end{align}
The hierarchical Laplacian acts coordinatewise for functions (such as $\varphi$)
taking values in $\R^n$,
and $|\varphi_x|=(\varphi_x \cdot \varphi_x)^{1/2}$ is the Euclidean norm of $\varphi_x$.
The expectation of a function $F$ of the spin field is defined by
\begin{equation}
\label{eq:expectation}
    \langle F \rangle_{\nu,N}^*
    =
    \frac{1}{Z_{\nu,N}^*}
    \int_{(\R^{n})^{\Lambda_N}}F(\varphi) e^{-H_{\nu,N}^* (\varphi)} d\varphi
    ,
\end{equation}
where the \emph{partition function} $Z_{\nu,N}^*$ is the normalisation constant
\begin{equation}
\label{eq:partition_function}
    Z_{\nu,N}^* = \int_{(\R^{n})^{\Lambda_N}}  e^{-H_{\nu,N}^* (\varphi)} d\varphi.
\end{equation}
The definitions \eqref{eq:Hamiltonian_def}--\eqref{eq:partition_function}
depend on $g$ but since we regard $g$ as fixed
(and small) in the following, we do not make this dependence explicit.
On the other hand, $\nu$ is variable, and the dependence on $\nu$ is our primary interest.

Let $\varphi_x^{(j)}$ denote the $j^{\rm th}$ component of $\varphi_x \in \R^n$.
For $\x,\y\in\Lambda_N$, the finite-volume \emph{two-point function} is
defined to be $\langle \varphi_{\x}^{(1)}\varphi_{\y}^{(1)} \rangle_{\nu,N}^*$.
The invariance of the Hamiltonian under the translation
$\varphi_{\boldsymbol{\cdot}} \mapsto \varphi_{\boldsymbol{\cdot}\oplus\z}$ shows that it suffices to consider the special case defined by
\begin{equation}
    G_{\nu, N}^*(\x) = \langle \varphi_{\o}^{(1)}\varphi_{\x}^{(1)} \rangle_{\nu,N}^*.
\end{equation}
Since the distribution of the field $\varphi$ under \eqref{eq:expectation}
is invariant under the map
$(\varphi_x)_{x\in\Lambda_N} \mapsto (M\varphi_x )_{x\in\Lambda_N}$ for every
orthogonal transformation $M \in O(n)$, the two-point function can
equivalently be written as
\begin{equation}
    G_{\nu, N}^*(\x)
    =
    \frac{1}{n} \langle \varphi_{\o}\cdot\varphi_{\x} \rangle_{\nu,N}^*.
\end{equation}
The finite-volume \emph{susceptibility} is defined by
\begin{align}
\label{eq:chiPhi}
    \chi_N^* (\nu )
    =
    \sum_{\x\in \Lambda_N}G_{\nu,N}^*(\x)
    .
\end{align}

\subsubsection{The critical point}

The following theorem is the $p=1$ special case of the more general
\cite[Theorem~1.1(ii)]{MPS23}.
It identifies the \emph{critical point} $\nu_c$: the point
at which the infinite-volume susceptibility is singular.
(The PBC case $d=4$ of Theorem~\ref{thm:book_main_theorem_bis}
was proved earlier in \cite{BBS-brief}.)
For $d=4$, recall from \eqref{eq:hatexponents} the definition
\begin{equation}
\label{eq:hatexponents-bis}
    \hat\gamma     = \frac{n+2}{n+8}.
\end{equation}

\begin{theorem}
\label{thm:book_main_theorem_bis}
Let $d\ge 4$,
let $n \in \N$, let $L$ be sufficiently large, and
let $g>0$ be sufficiently small (depending on $L$).
There exist a critical value $\nu_c \in \R$ and
a constant $A_d>0$ (both depending on $d,g,n,L$)
such that the infinite-volume limit of the susceptibility exists
for $\nu=\nu_c+\eps$ for all $\eps >0$, and the limit is independent of the
boundary condition $*=\free$ or $* =\per$.  Moreover, as $\eps \to 0$,
the limit diverges according to
\begin{equation}
\label{eq:subcrit_gaussian_moment}
	\chi_\infty(\nu_c+\eps)
    =
    \lim_{N \to \infty} \chi_N^*(\nu_c+\eps)
    \sim
	A_d \times
    \begin{cases}
    (\log \eps^{-1})^{\hat\gamma} \eps^{-1} & (d=4)
    \\
    \eps^{-1} & (d>4).
    \end{cases}
\end{equation}
\end{theorem}

It is also proved in \cite[Theorem~1.1(ii)]{MPS23} that,
as $g \downarrow 0$, the amplitudes and critical value obey
\begin{align}
    A_{4} &\sim
    \Big( \frac{B g}{\log L^{2}} \Big)^{\hat\gamma},
    \quad A_{d} =1+O(g) \;\;\;(d>4),
    \quad
    \nu_c \sim - (n+2) g \HLap_{0,\infty}  (\o) ,
    \label{eq:A,B_definition_bis}
\end{align}
with
\begin{equation}
\label{eq:B-def}
    B=(n+8) (1- L^{-d}) .
\end{equation}
For $d=4$, the constant $B$ will play an important role.
With PBC,
Theorem~\ref{thm:book_main_theorem_bis} is proved for the Euclidean model on $\Z^4$
for all $n \ge 1$ in \cite{BBS-phi4-log}.  A version of the theorem is also
proved for $n=1$ on $\Z^d$ for all $d \ge 4$ in \cite{HT87}.

\subsubsection{The effective critical points and critical windows}

We now recall some definitions from \cite[Section~1.4]{MPS23}.
Let $\hat\theta =  \frac 12 -\hat\gamma$ as in \eqref{eq:hatexponents}.
The \emph{window scale} (or \emph{rounding} scale) is defined by
\begin{equation}
\label{eq:window_choice}
	w_N =
    \begin{cases}
        A_{4} (\log L^2)^{\hat\gamma}\bbb^{-1/2}N^{-\hat\theta} L^{-2N} & (d=4)
        \\
        A_{d}\, g_\infty^{1/2} L^{-Nd/2}
        & (d>4),
    \end{cases}
\end{equation}
where $g_\infty = g+O(g^2)$ is determined in \cite{MPS23}
and $B$ is given by \eqref{eq:B-def}.

For PBC, we define the \emph{effective critical point} simply
to equal the infinite-volume critical point:
\begin{equation}
    \nu^{\rm P}_{c,N} = \nu_c .
\end{equation}
For FBC, we first define
\begin{align}
\label{eq:vNdef}
    v_N=
    \begin{cases}
        A_{4} (\log L^2)^{\hat{\gamma}}N^{\hat\gamma}L^{-2N}
                    & (d=4)
        \\
        A_{d}L^{-2N} & (d>4),
    \end{cases}
\end{align}
and then specify a volume-dependent \emph{effective critical point} $\nu_{c,N}^{\rm F}$ by
\begin{align}
\label{eq:nu_c,N_def}
    \nu_{c,N}^{\free}
    &=
    \begin{cases}
    \nu_{c}-\qLap v_N (1+\constF N^{-\hat \gamma}) & (d=4)
    \\
    \nu_{c}-\qLap v_N & (d=5)\\
	\nu_{c}-\qLap v_N(1+ O(L^{-N}) ) &  (d>5) ,
    \end{cases}
\end{align}
with $\qLap$ the constant from \eqref{eq:qLapdef},
and $\constF  = O(g)$.
Thus $\nu_{c,N}^{\free}$ is
specified exactly when $d=4,5$ and approximately when $d>5$.
This is discussed further
in \cite[Section~1.4.2]{MPS23}.
In particular, although we do not compute $\nu_{c,N}^{\free}$ explicitly for $d>5$, it has a definite specific value.

\subsection{Main results}

It is proved in \cite{MPS23} that the average field has a non-Gaussian limit
in the critical window $\nu=\nu_{c,N}^* + sw_N$ for $s \in \R$,
and has a Gaussian limit when $\nu=\nu_{c,N}^* + sv_N$ for $s>0$,
 for either boundary condition $* \in \{ \free, \per \}$.
We prove that in the non-Gaussian case the two-point function has a plateau,
whereas the plateau is absent in the Gaussian case.

\subsubsection{The plateau in the non-Gaussian regime}

For $k \in ( -1,\infty)$ let
\begin{align}
    I_k (s) = \int_{0}^{\infty} x^{k} e^{-\frac 14 x^4  -   \frac 12 s x^2}dx
    \qquad
    ( s \in \R )
    .
\label{eq:Ik_definition}
\end{align}
For $n \ge 1$,
we define the \emph{universal profile} $\uf_n:\R \to (0,\infty)$ by
\begin{equation}
\label{eq:profile-def}
    \uf_n(s)
    =
    \frac{\int_{\R^n} |x|^2 e^{-\frac 14 |x|^4 - \frac 12 s|x|^2} dx}
    {n\int_{\R^n} e^{-\frac 14 |x|^4 - \frac 12 s|x|^2} dx}
    =   \frac{ I_{n+1}  (s)}{n I_{n-1}  (s)}
    \qquad
    ( s \in \R )
    .
\end{equation}
A plot of $\uf_n$ for various values of $n$ can be found in \cite[Appendix~A]{MPS23}
(including its extension to $n=0$),
as well as a proof that $\uf_n(s)$ is strictly decreasing both as a function of $n$ and
as a function of $s$.
The \emph{large-field scale} is defined by
\begin{equation}
\label{eq:hNdef}
    \hh_N  =
    \begin{cases}
        (BN)^{1/4} L^{-N} & (d=4)
        \\
        g_\infty^{-1/4} L^{-dN/4} & (d>4),
    \end{cases}
\end{equation}
where $B$ and $g_{\infty}$ are as in \eqref{eq:window_choice}.
For error terms, we use the following convention throughout the paper.
The maximum of two real numbers is written as
$a\vee b = \max\{a,b\}$.

\begin{definition} \label{def:errorconv}
Given functions $F_1 (N,\x,s)$ and $F_2 (N,s)$, for $|\x| \le L^N$ and $s\in\R$,  we write:
\begin{enumerate}
\item $F_1 (N,\x,s) =  \cO_{N,\x}$ if $\lim_{N \vee |\x| \rightarrow \infty} F_1 (N,\x,s) = 0$,
\item $F_2 (N,s) = \cO_{N}$ if $\lim_{N\rightarrow \infty} F_2 (N,s) = 0$.
\end{enumerate}
The limits defining $\cO_{N,\x}$ and $\cO_{N}$ do \emph{not} claim uniformity in $s$.
However, $F_1 = \cO_{N,\x}$ does imply, by definition, that
$\lim_{N\rightarrow \infty} F_1 (N, \x, s) = 0$ uniformly in $\x$.
\end{definition}

Recall from \eqref{eq:HLap0-asy} that the massless hierarchical Green function obeys $\HLap_{0,\infty}(\x)\asymp |x|^{-(d-2)}$.

\begin{theorem}
\label{thm:mr-plateau}
Let $d\ge 4$,
let $n \in \N$, let $L$ be sufficiently large,
let $g>0$ be sufficiently small (depending on $L$), and consider
the $n$-component hierarchical
model with boundary conditions $*={\rm F}$ or $*={\rm P}$.
There is a function $\vv : \Lambda_\infty \rightarrow \R$ with
$\vv (\x) = O(g)$ (uniformly in $x$) and
$\lim_{|x| \rightarrow \infty} |\x|^{d-2} \vv (\x) = 0$, such that,
with
\begin{equation}
    \G (\x) = \HLap_{0,\infty} (\x) + \vv (\x)
\end{equation}
and for $s \in \R$ and $\x \in \Lambda_N$,
\begin{align}
\label{eq:mr-plateau}
    G^*_{\nu^*_{c,N}+sw_N,N}(\x)
    & =
    \G (\x) \Big( 1+ \cO_{N,\x} \Big)
    +
    f_n(s) \hh_N^{2} \Big( 1+ \cO_{N} \Big)
    .
\end{align}
The error terms in \eqref{eq:mr-plateau} are not uniform in $s$.
\end{theorem}

The function $\G(\x)$ is the critical two-point function of the infinite volume hierarchical $|\varphi|^4$ model.  Indeed, since $\lim_{N\to\infty}\hh_N=0$
by definition,
it follows from \eqref{eq:mr-plateau} that for fixed $\x\in\Lambda_\infty$
and fixed $s\in \R$
the infinite-volume limit exists independently of the boundary condition and
	\begin{equation}
        \label{eq:G-infvol}
		\G (\x) = \lim_{N\to \infty}G^*_{\nu^*_{c,N}+sw_N,N}(\x) .
	\end{equation}	
The bounds on $\vv(\x)$,  together with the asymptotic
behaviour of $\HLap_{0,\infty} (\x)$ in \eqref{eq:HLap0-asy},  imply that
\begin{align}
	\G (\x) = \HLap_{0,\infty} (\x) + \vv (\x) \asymp (|\x| \vee 1)^{-(d-2)}    \label{eq:Gnucx}
\end{align}
and that
\begin{align}
	\G (\x) \sim \HLap_{0,\infty} (\x) \quad \text{as} \quad |\x| \rightarrow \infty .
\end{align}

Theorem~\ref{thm:mr-plateau} shows that the
finite-volume two-point function has a plateau in
the critical windows centred around both effective critical points, with a universal
profile $\uf_n(s)$ for the $s$-dependence of the plateau term.
The two critical windows around $\nu^\free_{c,N}$ for FBC and $\nu^\per_{c,N}$ for PBC do not overlap, as their width
$w_N$ is smaller than their separation of order $v_N$.
The plateau
is made more explicit in the following corollary.
In its statement, the notation $a_N \ll b_N$ means $\lim_{N\to \infty} a_N/b_N = 0$.
The corollary is proved simply by substituting
for $\hh_N$ in \eqref{eq:mr-plateau} and comparing the size of the two terms
on the  right-hand side of \eqref{eq:mr-plateau}.

\begin{corollary}
	\label{cor:two_point_function_plateau}
	Under the
hypotheses of Theorem~\ref{thm:mr-plateau},
the following statements hold
(without uniformity in $s$).
\begin{enumerate}
\item
	Let $\x\in\Lambda_N$ possibly
    depend on $N$ and suppose that $N\rightarrow \infty$ in a manner that
    $|x| \ll L^{\frac{d}{2(d-2)} N}$
    when $d>4$,
     or
     $|\x| \ll N^{-1/4}L^{N}$
     when $d=4$.  In this limit, with $\G (\x)$ as in \eqref{eq:Gnucx}, we have
	\begin{equation}
		\label{eq:G-smallx}
		G^*_{\nu^*_{c,N}+sw_N,N}(\x) \sim
		\G (\x)
    \qquad (d\ge 4) .
	\end{equation}
	
\item
	Let $\x\in\Lambda_N$ depend on $N$ and suppose that
    $|x|,N\rightarrow \infty$ in a manner that
    $|x| \gg L^{\frac{d}{2(d-2)} N}$
    when $d>4$, or
     $|\x|   \gg N^{-1/4}L^{N}$
     when $d=4$.  Then
	\begin{equation}
        \label{eq:G-largex}
		G^*_{\nu^*_{c,N}+sw_N,N}(\x)
    \sim
    \begin{cases}
        (BN)^{1/2}f_n(s)L^{-2N} & (d=4)
        \\
        g_\infty^{-1/2} f_n(s)L^{-dN/2} & (d>4).
    \end{cases}
	\end{equation}

\end{enumerate}
\end{corollary}

The following
corollary for the finite-volume susceptibility
recovers \cite[Corollary~1.4]{MPS23},
with a different proof.  The proof shows that the susceptibility's leading volume dependence arises solely from the plateau.

\begin{corollary}
\label{cor:chi}
Under the hypotheses of Theorem~\ref{thm:mr-plateau},
as $N \to \infty$,
\begin{equation}
\label{eq:chi-cor}
    \chi_N^*(\nu^*_{c,N}+sw_N) \sim L^{dN} \hh_N^{2}  f_n(s)
    =
    f_n(s) \times
    \begin{cases}
    (BN)^{1/2}L^{2N} & (d=4)
    \\
    g_\infty^{-1/2} L^{dN/2} & (d>4).
    \end{cases}
\end{equation}
\end{corollary}

\begin{proof}
We sum \eqref{eq:mr-plateau} over $\Lambda_N$.  The leading part of
the plateau term gives exactly $f_n(s)\hh_N^2 L^{dN}$, which is the right-hand
side of \eqref{eq:chi-cor}.  The remainder term from the plateau is relatively small.
By \eqref{eq:Gnucx},  the decaying term satisfies
\begin{align}
	\sum_{x \in \Lambda_N} \G (\x)
    \left( 1 + \cO_{N,x}
    \right)
    \asymp \sum_{|x| \le L^{N}} \frac{1}{|x|^{d-2}} \asymp L^{2N} \ll  L^{dN} f_n (s) \hh_N^2
	.
\end{align}
This proves \eqref{eq:chi-cor} and shows that the plateau term is
dominant for the susceptibility.
\end{proof}

Corollary~\ref{cor:chi} is a special case of the following more general theorem
proved in \cite[Theorem~1.2(ii)]{MPS23}, which states that
for all integers $p \ge 1$
the moments of the average field
$\Phi_N = |\Lambda_N|^{-1}\sum_{x\in\Lambda_N}\varphi_x$
have the asymptotic behaviour
\begin{align}
    \big\langle
    |\Phi_{N} |^{2p}
    \big\rangle_{g, \nu_{c,N}^* +sw_N,N}^{*}
    & =
    \hh_N^{2p} \,
    \Sigma_{n,2p}(s)
    \big( 1+ o(1)  \big),
\label{eq:mrPBC-nongaussian_moments}
\end{align}
where
\begin{equation}
\label{eq:Signks}
    \Sigma_{n,k}(s)
    =
    \frac{\int_{\R^n}|x|^k e^{-\frac 14 |x|^4 - \frac s2 |x|^2} dx}{\int_{\R^n} e^{-\frac 14 |x|^4 - \frac s2 |x|^2} dx}.
\end{equation}

\subsubsection{Absence of plateau in the Gaussian regime}

The Gaussian regime concerns values of $\nu$ of the form
$\nu=\nu_{c,N}^* + sv_N$ with $s>0$, where as usual the asterisk
corresponds to the choice
of boundary condition, PBC or FBC.
The two-point function in the Gaussian regime is given by the following theorem.

\begin{theorem}
\label{thm:mr-plateau-Gaussian}
Let $d\ge 4$,
let $n \in \N$, let $L$ be sufficiently large,
let $g>0$ be sufficiently small (depending on $L$), and consider
the $n$-component hierarchical
model with boundary conditions $*={\rm F}$ or $*={\rm P}$.
For $s >0$, for any strictly positive sequence $s_N$ with $s_N \to s$, and
for $\x \in \Lambda_N$,
we have
\begin{align}
\label{eq:mr-Gaussian-PBC}
    G^\per_{\nu^\per_{c,N}+s_Nv_N,N}(\x)
    & =
    \left( \HLap_{sL^{-2N},N}^\per(\x) + \vv (\x) \right)
	\Big( 1+ \cO_{N,\x} \Big)
    ,
\\
\label{eq:mr-Gaussian-FBC}
    G^{\free}_{\nu^\free_{c,N}+s_Nv_N,N}(\x)
    & =
    \left( \HLap_{(s-q) L^{-2N},N}^\free (\x) + \vv (\x) \right)
	\Big( 1+ \cO_{N,\x} \Big),
\end{align}
with the error $\cO_{N,\x}$ as defined in Definition~\ref{def:errorconv} and $\vv (\x)$ the same as in Theorem~\ref{thm:mr-plateau}.
\end{theorem}

The covariances appearing on the right-hand sides of
\eqref{eq:mr-Gaussian-PBC}--\eqref{eq:mr-Gaussian-FBC}
are defined in \eqref{eq:Csum},
and complete information about them can be extracted.  In particular,
we prove in Lemma~\ref{lemma:corGaussianLemma}
that both covariances are bounded
above and below by multiples of $|x|^{-(d-2)}$, and both have limit
$\HLap_{0,\infty}(x)$ as $N \to \infty$.  These two facts immediately lead
to the following corollary.
It shows that the critical infinite-volume two-point function
is the limiting two-point function in both the Gaussian and
(with \eqref{eq:G-infvol}) the non-Gaussian regime,
independently of the boundary condition, and that the plateau
is absent in the Gaussian regime.

\begin{corollary}
\label{cor:Gaussian}
Under the assumptions of Theorem~\ref{thm:mr-plateau-Gaussian},
for fixed $\x\in\Lambda_\infty$, the infinite-volume limit exists independently of the boundary condition and, with $\G (\x)$ as in \eqref{eq:Gnucx},
	\begin{equation}
        \label{eq:G-infvol-Gaussian}
		\G (\x) = \lim_{N\to \infty}G^*_{\nu^*_{c,N}+sv_N,N}(\x) .
	\end{equation}
Also, for sufficiently large $N$,
with $s$-dependent (but $N,x$-independent) constants, for all $x \in \Lambda_N$
we have
\begin{align}
\label{eq:mr-Gaussian-x}
    G^*_{\nu^*_{c,N}+sv_N,N}(\x)
    & \asymp
    \big( |\x | \vee 1 \big)^{-(d-2)}
    .
\end{align}
\end{corollary}

It follows from \eqref{eq:nu_c,N_def} that
the equation $\nu_{c,N}^\free + \tilde s_N v_N = \nu_c$
has a solution $\tilde s_N \sim \qLap$ for all $d \ge 4$.
Therefore \eqref{eq:mr-Gaussian-x} shows that
the two-point function with FBC has no plateau at the infinite-volume
critical point $\nu_c$, unlike for PBC where Theorem~\ref{thm:mr-plateau} (with $s=0$)
guarantees that there is a plateau at $\nu_c$.
The next corollary recovers the result of
\cite[(1.51)]{MPS23} that with FBC the susceptibility at $\nu_c$ grows as $L^{2N}$,
which is smaller than the PBC susceptibility given by Corollary~\ref{cor:chi}.

\begin{corollary}
\label{cor:chiv}
Under the
hypotheses of Theorem~\ref{thm:mr-plateau-Gaussian},  as $N \to \infty$,
\begin{equation}
\label{eq:chiv}
    \chi_N^*(\nu^*_{c,N}+s_Nv_N) \sim  s^{-1} L^{2N}.
\end{equation}
In particular,
\begin{align}
\label{eq:chiNnuc}
	\chi_N^{\free} (\nu_c ) \sim q^{-1} L^{2N}	
	.
\end{align}
\end{corollary}
\begin{proof}
It follows from \eqref{eq:Csum} and \eqref{eq:Csum0} that
\begin{align}
    \sum_{x\in \Lambda_N} \HLap_{sL^{-2N},N}^\per(\x)
    & =
    \sum_{x\in \Lambda_N}  \HLap_{(s-q) L^{-2N},N}^\free (\x)
    =
    \frac{1}{s}L^{2N}.
\end{align}
This implies \eqref{eq:chiv}, since the terms $\alpha_\infty(x)$ and $\cO_{N,x}$
are relatively small and do not contribute to the leading behaviour of the sum
over $x$.
Then \eqref{eq:chiNnuc} follows from the existence of a sequence
$\tilde{s}_N \rightarrow q$ satisfying $\nu_{c,N}^\free + \tilde s_N v_N = \nu_c$,
as indicated above the statement of the corollary.
\end{proof}

Corollary~\ref{cor:chiv} is a special case of the following more general theorem
proved in \cite[Theorem~1.3(ii)]{MPS23}, which states that
under the hypotheses of Theorem~\ref{thm:mr-plateau-Gaussian},
for all integers $p \ge 1$ the moments
of the average field
$\Phi_N$
have the asymptotic behaviour
\begin{align}
    \big\langle
    |\Phi_{N} |^{2p}
    \big\rangle_{g, \nu_{c,N}^* +sv_N,N}^{*}
    & =
    L^{-p(d-2)N}
    M_{n,2p}(s)
    \big( 1+ o(1)  \big),
\label{eq:mrPBC-gaussian_moments}
\end{align}
where
\begin{equation}
\label{eq:Mnks}
    M_{n,2p}(s)
    =
    \frac{\int_{\R^n} |x|^{2p} e^{-\frac s2 |x|^2} dx}{\int_{\R^n} e^{- \frac s2 |x|^2} dx}
    =
    \left(\frac{2}{s}\right)^p
    \frac{\Gamma(\frac{n+2p}{2})}{\Gamma(\frac{n}{2})}
    .
\end{equation}

The asymptotic formula
\eqref{eq:mr-Gaussian-PBC}
is related to
\cite[Theorem~1.3]{BBS-phi4-log}, which proves that
when $\nu-\nu_c$ is of order $L^{-2N}$ the scaling limit of the Euclidean
(usual, non-hierarchical) spin field on the
discrete $4$-dimensional torus is a massive Gaussian Free Field on the unit continuum torus.
The result of \cite[Theorem~1.3]{BBS-phi4-log} is for the scaling limit of the spin field
smeared against a smooth test function.  Theorem~\ref{thm:mr-plateau-Gaussian} handles the
more singular situation without smearing, in the hierarchical setting
and in all dimensions $d \ge 4$.
It is complementary  to
\cite[Theorem~1.3]{BBS-phi4-log} since it shows
that the two-point function is asymptotically that of a
Gaussian field.
For the Euclidean model on $\Z^4$, the fact that the critical two-point function has
the Gaussian $|x|^{-(d-2)}$ decay is proved in \cite{ST-phi4,GK85}.
Gaussian limits at the critical point are also studied in
\cite{Aize82,Froh82,AD21}.

By \eqref{eq:Csum}, the covariances on the right-hand sides of \eqref{eq:mr-Gaussian-PBC}--\eqref{eq:mr-Gaussian-FBC} each contain a constant term $s^{-1}L^{-(d-2)N}$.
The domination of this constant term
over the $|x|^{-(d-2)}$ decay
as $s \downarrow 0$ reflects the onset of non-Gaussian behaviour at the
effective critical point $\nu_{c,N}^*$.
In this regard, if we formally
replace $s v_N$ by $s'w_N$,
corresponding to the replacement of the Gaussian by
the non-Gaussian regime, then by the definitions of $w_N$, $v_N$, and $\hh_N$ in
\eqref{eq:window_choice}, \eqref{eq:vNdef}, and \eqref{eq:hNdef},
the constant term $s^{-1}L^{-(d-2)N}$
becomes
\begin{equation}
\label{eq:crossover}
    \frac{v_N L^{-(d-2)N}}{s' w_N}
    =
    \frac{1}{s'} \hh_N^2.
\end{equation}
Since $f_n(s) \sim s^{-1}$ as $s \to \infty$ (see \cite[(1.72)]{MPS23}),
this is comparable to \eqref{eq:mr-plateau},
illustrating a crossover to the non-Gaussian plateau.

\subsection{Open problems}
\label{sec:open-problems}

Although our results are proved for the weakly-coupled
hierarchical $|\varphi|^4$ model, our proof uses properties of the renormalisation
group flow which are predicted to be universal.  This points to
 the following open problems.

\subsubsection{Euclidean (non-hierarchial) models}

\smallskip\noindent{\bf Problem 1.}
Prove the results of Theorems~\ref{thm:mr-plateau} and \ref{thm:mr-plateau-Gaussian}
for the Euclidean (defined with the usual nearest-neighbour Laplacian)
$|\varphi|^4$ model on $\Z^d$ for dimensions $d \ge 4$ and for $n \ge 1$.
Do the same for the Ising, XY, and Heisenberg models,
and more generally for all $N$-vector models.

\smallskip
We expect that the conclusions of Theorems~\ref{thm:mr-plateau} and \ref{thm:mr-plateau-Gaussian}
hold in the Euclidean setting for all these models, with the following modifications:
\begin{enumerate}
\item
Due to wave-function renormalisation (not present in the hierarchical model),
the free covariance $\HLap$
will occur with a non-universal
constant factor different from $1$.
\item
Non-universal constant prefactors will also occur in the large-field scale $\hh_N$, the window
width $w_N$, and the FBC shift $v_N$ of the effective critical point.
\item
For FBC, it may be necessary to restrict $\x$ to lie away from the boundary of the box.
\end{enumerate}

The challenges for extending our methods and results
to the weakly-coupled Euclidean $|\varphi|^4$ model are serious.
The easier case is PBC, for which an extension
would require new ideas even though a large part of our analysis
has already been extended \cite{BS-rg-IE,BS-rg-step}.  One
challenge in the Euclidean setting would be to
improve the large-field regulator used in \cite{BS-rg-IE,BS-rg-step}, which bounds
the non-perturbative RG coordinate by an exponentially \emph{growing}
factor (see  \cite[(1.38)]{BS-rg-IE}) rather than the exponentially \emph{decaying} factor
that we use
(see \eqref{eq:Gj}).  Also,  the hierarchical covariance decomposition we use
has a constant covariance at the final scale, whereas the decomposition used
in \cite{BS-rg-IE,BS-rg-step} does not have this useful feature.
This issue could likely be addressed using the decomposition in \cite[Section~3]{BPR24},
which does have a constant covariance at scale $N$, and which obeys similar estimates
to the decomposition used in \cite{BS-rg-IE,BS-rg-step}
(see \cite[Corollary~4.1]{BPR24} and \cite[Proposition~3.4]{BPR24}).

For FBC, an additional issue is that a Euclidean finite volume lacks translation
invariance, unlike in our hierarchical setting,
and its genuine boundary adds to the difficulty.
Nevertheless, a lesson from the hierarchical model is that
for the plateau and the susceptibility (which are bulk quantities) the
behaviour for FBC at the effective critical point $\nu_{c,N}^\free$
is due to a shift of the PBC behaviour at $\nu_{c,N}^\per =\nu_c$ caused by
what is essentially a mass shift of order $L^{-2N}$ in the FBC
Laplacian \eqref{eq:ChatN}.
The Euclidean Laplacian
with FBC also has a mass of the same order $L^{-2N}$, compared to the PBC Laplacian
which has a zero mode.  It remains a challenging open problem to turn this analogy
into a proof for Euclidean FBC.

\subsubsection{Self-avoiding walk}

A lower bound proving the existence of a plateau for the critical
weakly self-avoiding walk in dimensions $d>4$ with PBC
is given in \cite{Slad23_wsaw},
and for strictly self-avoiding walk in $d>4$ in \cite{Liu24}.

As was first suggested in \cite{Genn72}, the self-avoiding
walk can be considered as the $n=0$ version of an $n$-component spin model.
A rigorous version of this statement is that self-avoiding walk models can
be exactly represented by certain supersymmetric spin models (see
\cite[Chapter~11]{BBS-brief} for details and history).
This has allowed, e.g., for proof that the susceptibility of
the $4$-dimensional weakly self-avoiding walk obeys the $4$-dimensional
case of \eqref{eq:subcrit_gaussian_moment} (and more)
with $n=0$ \cite{BBS-saw4-log,BBS-saw4}. It is therefore natural to expect
that our results for $n \ge 1$ have counterparts for self-avoiding walk, when
$n$ is set equal zero.

As discussed in \cite[Section~1.5.2]{MPS23}, the universal profile $\uf_n$
defined in \eqref{eq:profile-def} extends
naturally to real values $n \ge -2$,
In particular,
\begin{equation}
	\uf_0(s) =  \int_{0}^{\infty} x  e^{-\frac 14 x^4  -   \frac 12 s x^2}dx .
\end{equation}

\noindent{\bf Problem 2.}
Prove the statements of both Theorems~\ref{thm:mr-plateau} and \ref{thm:mr-plateau-Gaussian}
with the parameter $n$ set equal to $n=0$ in
the definitions of $w_N$, $v_N$, and $f_n$ (in particular, with
$\hat\gamma = \hat\theta = \frac 14$),
for the
weakly self-avoiding walk in dimension $d \ge 4$, both on the hierarchical lattice
and on $\Z^d$.
A much more ambitious problem is
to prove this for the strictly self-avoiding walk on $\Z^d$
in dimensions $d \ge 4$.

\smallskip
For the hierarchical continuous-time weakly self-avoiding walk, the asymptotic
formula \eqref{eq:G-infvol} has been proven already in \cite{BI03d}, and this
was used to prove a logarithmic correction to the end-to-end distance in
\cite{BI03c}.
We believe that our method could be extended in a straightforward way to
prove our results for the hierarchical continuous-time weakly self-avoiding walk
in dimensions $d \ge 4$.  This would involve including fermions in the analysis,
and it is well-understood how to do so \cite{BI03d,BBS-saw4}.
The challenges for an extension to the Euclidean model are similar to those
mentioned below Problem~1.

Evidence of the universality of the profile $\uf_0$ is provided by
the fact that it appears for both weakly and strictly
self-avoiding walk on the complete graph.  The occurrence of $\uf_0$ in the profile for the susceptibility is discussed in \cite[Section~1.5.4]{MPS23}, on the basis of results of
\cite{BS20,Slad20}.  Those results show that the susceptibility on the complete
graph on $V$ vertices,  in a critical window of width $V^{-1/2}$,
is given for both models as  $\lambda_1 \uf_0(\lambda_2s) V^{1/2}$
with non-universal constants $\lambda_1,\lambda_2$.  Since the two-point
function takes the same value at any pair of distinct points, the two-point
function at distinct points
is therefore asymptotic to $\lambda_1 \uf_0(\lambda_2s) V^{-1/2}$
in the window.  This is consistent with the $d>4$ case of \eqref{eq:G-largex}.
A related conjecture for the
universality of the profile for the expected length of self-avoiding walks in dimensions $d>4$
is investigated numerically in \cite{DGGZ22}.

\subsection{Guide to paper}

Our analysis extends the renormalisation group method
which was used in \cite{MPS23} to study
the average field and its moments, with the new feature that we include
\emph{observable fields} which enable the computation of the two-point
function.  The observable fields are introduced in Section~\ref{sec:2}.
In Section~\ref{sec:2}, we reduce the computation of the two-point function
to a problem of Gaussian integration with covariance $\HLap_{N}^*$.
By \eqref{eq:Csum}, $\HLap_{N}^*$ is a sum of two covariances
$C_{\le N} + C_{\hat N}^*$, and the Gaussian integral can be performed
in two steps.  The result of the integration with covariance $C_{\le N}$
is stated in Theorem~\ref{thm:RG-obs}.
The integration over the final covariance $C_{\hat N}^*$ is simply an integral over $\R^n$,
and we analyse that integral in Section~\ref{sec:2} and thereby prove our main results.

Our main work is to prove Theorem~\ref{thm:RG-obs}.
Its proof occupies
Sections~\ref{sec:3}--\ref{sec:4}.
In Section~\ref{sec:3}, we define the renormalisation group map (RG map)
in the presence
of the observable fields.
This RG map is what we use to perform integration
with covariance $C_{\le N}$ scale by scale, in $N$ steps.
The observable fields necessitate modifications to the norms that we use
to control the RG map.

In Section~\ref{sec:RGflow}, we provide
the statement that permits a generic RG step to be carried
out in order to construct the RG flow, in
Theorem~\ref{thm:Phi^K_q0}.
Theorem~\ref{thm:Phi^K_q0} is an extension of \cite[Theorem~5.7]{MPS23}
to include the observables.  It is the most important ingredient
in our analysis as it
provides a rigorous control of all non-perturbative corrections to perturbation
theory.
In Section~\ref{sec:RGflow},
we show that the critical points constructed
in \cite{MPS23} serve as initial conditions for an RG flow also in the presence
of the observables.  Moreover, the \emph{bulk} RG flow constructed in \cite{MPS23}
remains unchanged in the presence of the observables---the observables extend the
bulk flow to include additional coordinates.
Although the observable fields have been employed in previous work exactly at the
upper critical dimension \cite{BBS-saw4,BLS20,BSTW-clp,LSW17,ST-phi4}, this has
not been done previously above the upper critical dimension.  Innovation is required
to include the observables here for $d>4$, due to the fact that
$\varphi^4$ is a so-called \emph{dangerous irrelevant variable}
\cite{Fish83}.
The ``danger'' refers to the fact that although $\varphi^4$ is
irrelevant for $d>4$, setting its coefficient
equal to zero gives a qualitatively different result, and its RG flow has to be carefully controlled.
Indeed, in the RG flow, $g_j\sum_{x\in \Lambda_j}|\varphi_x|^j$ scales as
$g_\infty L^{-(d-4)j}$, and this $g_\infty$ appears in our main results
via \eqref{eq:window_choice} and \eqref{eq:hNdef}.

In Section~\ref{sec:RGflow},
we reduce the proof of the important
Theorem~\ref{thm:Phi^K_q0} to two propositions:
Proposition~\ref{prop:Phi+0} and
Proposition~\ref{prop:crucial-short-3}.
Proposition~\ref{prop:Phi+0} provides bounds on a single RG step.
It is proved in Section~\ref{sec:RGextendednorm}.
Proposition~\ref{prop:crucial-short-3} is a crucial contraction theorem
which extends the $(p,q)=(0,1)$
case of
\cite[Theorem~5.7]{MPS23} to include the observables.  It is proved in
Section~\ref{sec:contractions}.
Several estimates which are needed in the proof of
Propositions~\ref{prop:Phi+0}--\ref{prop:crucial-short-3}
are deferred to Appendix~\ref{sec:pfSProps}.
Appendix~\ref{app:covariance} provides bounds on the hierarchical covariances introduced in Section~\ref{sec:C}.

\section{Reduction of proof of main results}
\label{sec:2}

In this section, we state the important Theorem~\ref{thm:RG-obs}
and show how it implies our main results Theorems~\ref{thm:mr-plateau} and
\ref{thm:mr-plateau-Gaussian}.  In preparation for the statement of
Theorem~\ref{thm:RG-obs}, we introduce the observable fields, and recall
results concerning the bulk RG flow from \cite{MPS23}.

\subsection{Observable fields and Gaussian integration}

We follow the approach of \cite{BBS-saw4,ST-phi4,BLS20,LSW17} to use observable fields
to generate the two-point function.
This is related to but simpler than the
observable field used in \cite{BI03d}.
Let $\sigma_\o$ and $\sigma_\x$ generate a commutative ring with
\begin{equation}
\label{eq:sigmaox}
    1 = \sigma_\bulk, \qquad
    \sigma_\o^2 = \sigma_\x^2 = 0, \qquad
    \sigma_\o \sigma_\x = \sigma_{\o\x}.
\end{equation}
The ring element $\sigma_\x$ should \emph{not} be thought of as a function of $\x$, it is simply a name for the ring element and we will take $\sigma_\x \neq \sigma_\o$ even when $\x = \o$.
We refer to $\sigma_{\o}$ and $\sigma_{\x}$ as the \emph{observable fields},
or, more compactly, as the \emph{observables}.
We consider
functions of $\varphi$ and the observables that are defined by
Taylor polynomials in $\sigma_\o$ and $\sigma_\x$ and have the form
\begin{equation}
    F(\varphi,\sigma) = F_\bulk(\varphi) + \sigma_\o F_\o(\varphi)  + \sigma_\x F_\x(\varphi)
    + \sigma_{\ox} F_{\ox}(\varphi).
\end{equation}
When the observables are absent, we are left with the \emph{bulk},
denoted by the symbol $\varnothing$.

Given a spin field $\varphi :\Lambda_N \to \R^n$, $g>0$, and $\nu\in \R$, let
\begin{align}
\label{eq:V0bulkdef}
    V_{0,\bulk}(\varphi)
    &=
    \sum_{z \in \Lambda_N}
    \Big(\frac{1}{4} g  |\varphi_z|^4
    +
    \frac{1}{2} \nu |\varphi_z|^2
    \Big)
	,
	\\
	\label{eq:V0def}
	V_0 ( \varphi)
	&	=   V_{0, \bulk} (\varphi)
    - \sigma_{\o} \varphi_\o^{(1)}
    - \sigma_{\x} \varphi_\x^{(1)}
    ,
    \\
	Z_0 (\varphi) &= e^{- V_0(\varphi)}
	\label{eq:Z_0_definition}	
	.
\end{align}
Then
\begin{equation}
\label{eq:eVOxy}
    (Z_0 (\varphi))_{\ox}
    =
    e^{-V_{0, \bulk} (\varphi)}
    \varphi_\o^{(1)}  \varphi_\x^{(1)}
    .
\end{equation}

The Gaussian measure with mean zero and covariance $C$ is the probability measure
on $(\R^n)^{\Lambda_N}$ which is proportional to
$\exp[-\sum_{x,y\in\Lambda_N}\varphi_x \cdot (C^{-1}\varphi)_y]\prod_{z\in\Lambda_N}
d\varphi_z$, where $d\varphi_z$ denotes Lebesgue measure on $\R^n$.
We write $\E_C$ for the expectation with respect to this measure.
Then, with the expectation on the left-hand side defined
by \eqref{eq:expectation}, and with the covariance $\HLap_{\ka,N}^* = (-\Delta^* + \ka)^{-1}$, we have
\begin{equation}
\label{eq:EF}
    \langle e^{\sigma_{\o} \varphi_\o^{(1)}+ \sigma_{\x} \varphi_\x^{(1)}} \rangle_{\nu+ \ka,N}^*
    =
    \frac{\E_{\HLap_{\ka,N}^*}  Z_0}{\E_{\HLap_{\ka,N}^*} Z_{0, \bulk}}
    =
    \frac{\E_{\HLap_{\ka,N}^*}   e^{-{V_{0}}}}{\E_{\HLap_{\ka,N}^*} e^{-{V_{0,\bulk}}}}
    .
\end{equation}
Since
$\HLap_{\ka,N}^* =  C_{\ka,\le N}  + C_{\ka,\hat N}^*$ (recall \eqref{eq:Csum}),
it follows from a standard fact (see \cite[Corollary~2.1.11]{BBS-brief}) about Gaussian integration that
the expectations in
\eqref{eq:EF} can be computed in two steps as
\begin{equation}
\label{eq:2int}
    \E_{\HLap_{\ka,N}^*} F = \E_{C_{\ka,\hat N}^*} \E_{C_{\ka,\le N}} F(\varphi + \zeta),
\end{equation}
where $\E_{C_{\ka,\le N}}$ is independent of the boundary condition
and involves integration over $\zeta$, and $\E_{C_{\ka,\hat N}^*}$ does depend on the boundary
condition and involves integration over $\varphi$.
Moreover, $\E_{C_{\hat N}^*}$ is supported on constant fields $\varphi$,
 so it is merely an integral over $\R^n$
 (degenerate Gaussian fields with $C$ not invertible are
discussed in \cite[Section~2.1]{BBS-brief}).
We define
\begin{equation}
\label{eq:Z_N_definition}
    Z_N(\varphi ; \ka) =  Z_N(\varphi) = \E_{C_{\ka,\le N}}Z_0(\varphi + \zeta ),
\end{equation}
which is independent of the choice of boundary condition.  Then \eqref{eq:2int} gives
\begin{equation}
\label{eq:ECh}
    \E_{\HLap_{\ka,N}^*}   Z_0    =
    \E_{C_{\ka,\hat N}^*}
    Z_N .
\end{equation}

Since $Z_N$ depends also on the observables, it can be written as
\begin{equation}
    Z_N = Z_{N,\bulk} + \sigma_\o Z_{N,\o} + \sigma_\x Z_{N,\x} + \sigma_\ox Z_{N,\ox}
    .
\end{equation}
Each of the components of $Z_N$ is a function of a constant field on $\Lambda_N$.
The observables are absent in the analysis of the average field in \cite{MPS23},
and there $Z_N$ is purely given by $Z_{N,\bulk}$, which is well understood in \cite{MPS23}.
It is the observable components that we must study here in order to analyse the two-point
function.

The two-point function is given by the formula in the following lemma.
The denominator on the right-hand side of \eqref{eq:Gobs}
is independent of the observable field and
has already been well understood in \cite{MPS23}.  Our task now will be to understand the numerator of \eqref{eq:Gobs}.

\begin{lemma}
\label{lem:Gox}
For $g>0$, for $\nu \in \R$, for $\ka > -\frac 12 L^{-2(N-1)}$,
and for $\x\in\Lambda_N$,
\begin{equation}
\label{eq:Gobs}
    G_{\nu+\ka,N}^*(\x)
    =
    \frac{\E_{\HLap_{\ka,N}^*} (e^{-{V_{0}}})_{\ox}}{\E_{\HLap_{\ka,N}^*} e^{-{V_{0,\bulk}}}}
    =
    \frac{\E_{C_{\ka,\hat N}^*}Z_{N,\ox}(\ka)}{\E_{C_{\ka,\hat N}^*}Z_{N,\bulk}(\ka)}
    .
\end{equation}
\end{lemma}

\begin{proof}
The two-point function is the $\ox$ component of the left-hand side
of \eqref{eq:EF}.  With this observation,
the first equality follows from \eqref{eq:EF} and \eqref{eq:eVOxy},
and the second follows from \eqref{eq:ECh}.
\end{proof}

\subsection{Integration to the final scale: the bulk}

\subsubsection{Result of the bulk RG}

In \cite{MPS23}, the observables are absent,
and a renormalisation group analysis is used to obtain good control
of the bulk part $Z_{N,\varnothing}$ of $Z_N$.  Since we build on that control in order
to incorporate the observables, we gather some facts about
$Z_{N,\varnothing}$ in Theorem~\ref{thm:rg_main_theorem_new} below.

To state Theorem~\ref{thm:rg_main_theorem_new}, we need the following definitions.
Given $g_N>0$ and $\nu_N\in\R$, we define
\begin{align}
V_{N,\bulk} (y)
=
|\Lambda_N|
\big(  \textstyle{\frac{1}{4}} g_N |y|^4  + \textstyle{\frac12} \nu_N |y|^2 \big)
\qquad  (y \in \R^n).
\label{eq:V_N_form}
\end{align}
We further define the Gaussian scale
\begin{equation}
\label{eq:lpNdef}
    \lp_N = L^{-N(d-2)/2},
\end{equation}
and, for error terms,
\begin{equation}
	\label{eq:eNdef}
	e_N =
	\begin{cases}
		N^{-3/4} & (d=4)\\
		g_N^{3/4} L^{-3(d-4)N/4} & (4 < d \le 12)
		\\
		 g_N^{3/4} L^{-Nd/2} & (d>12).
	\end{cases}
\end{equation}
Finally, we define an $N$-dependent mass domain for $\ka$,
as in \cite[(3.71)]{MPS23}, by
\begin{align}
	\label{eq:I_crit_asymp}	
	\II_\crit
	&= \begin{cases}
		(-\frac12 L^{-2(N-1)}, L^{-2N(1-N^{-1/2})} ) &(d=4)\\
		(-\frac12 L^{-2(N-1)},2 L^{-3N/2} ) &(d > 4).
	\end{cases}
\end{align}
Theorem~\ref{thm:rg_main_theorem_new} combines a restricted version of
\cite[Theorem~3.1]{MPS23} (for positive $\ka\in \II_{\rm crit}$) with
\cite[Theorem~3.3]{MPS23} (for negative $\ka\in \II_{\rm crit}$), as well
as \cite[Lemma~3.10]{MPS23} (for all $\ka \in \II_{\rm crit}$).

Bounds are also known for the \emph{vacuum energy} $\bar u_{N,\varnothing}$
which appears on \eqref{eq:Z_N_form_ter},
but since it cancels when computing the two-point function, we do not
state those bounds here.
As mentioned above \eqref{eq:Z_N_definition}, we only need to evaluate $Z_N$
on constant fields.  Such fields are identified with $y\in\R^n$ in the theorem.
The choice of the (small)
parameter $\kappa=\kappa(n)>0$ in \eqref{eq:K_bound_at_nu_c_ter}
is indicated at \eqref{eq:kappa-def}.

\begin{theorem}
\label{thm:rg_main_theorem_new}
Let $d \ge 4$, let $n \ge 1$, let $L$ be sufficiently large,
and let $g >0$ be sufficiently small.
Let $N \in \N$.
There exists a continuous strictly increasing function
$\nu_{1,N} (\ka)$ of $\ka \in \II_{\rm crit}$,
and there exist
$(\bar u_{N,\varnothing}, g_N, \nu_N, K_{N,\varnothing})$  (all depending on $\ka$), such that
$Z_{N,\varnothing}$ defined by
\eqref{eq:Z_N_definition} with the choice $\ka$ for the mass in the covariance $C_{\ka,\le N}$
and with the choice $\nu = \nu_{1,N} (\ka)$ for the $\nu$ in $V_{0,\varnothing}$, satisfies
\begin{align}
Z_{N,\varnothing} (y)
=
e^{\bar u_{N,\varnothing}} \big( e^{-V_{N,\varnothing} (y) } + K_{N,\varnothing} (y) \big)
\label{eq:Z_N_form_ter}
\end{align}
for all $y\in\R^n$ with, uniformly in $\ka$:
\begin{align}
	g_N &  =
    \begin{cases}
        (BN)^{-1} \big( 1+O(N^{-1}\log N) \big) & (d=4)
        \\
        g_\infty  + O(g^2 L^{-N / 2}) & (d>4),
    \end{cases}
	\label{eq:gNB_ter}
	\\
	|\nu_N| & \le  O(g_N\lp_N^2)  ,
	\label{eq:nu_j_bound_ter}  \\
	|K_{N,\varnothing} (y)| & \le  O( e_N )
    e^{-\kappa g_N |\Lambda_N| |y|^4},
	\label{eq:K_bound_at_nu_c_ter}
\end{align}
with $e_N$ defined in \eqref{eq:eNdef}, and with $g_\infty = g+O(g^2)$
in \eqref{eq:gNB_ter}.
\end{theorem}

In the notation of \cite{MPS23},
the critical value $\nu_{1,N}$ is the continuous function on $\II_{\rm crit}$
given by
\begin{equation}
\label{eq:nu_1_N_def}
	\nu_{1,N}(\ka) =
	\begin{cases}
		\nu_c(\ka) &  ( \ka \in \II_\crit \cap [0,1) )   \\
		\nu_{0,N}(\ka) & ( \ka \in (-\frac12 L^{-2(N-1)} ,  0) ).
	\end{cases}
\end{equation}

\subsubsection{The renormalised mass}

In \cite[Definition~3.6]{MPS23},
the \emph{renormalised masses} are defined to be the functions
$\ka_N^*$
and $\tilde\ka_N^*$
given by the unique solutions to
\begin{alignat}{2}
\label{ren-mass}
	\nu_{c,N}^*+sw_N &= \nu_{1,N}(\ka^*_N(s))+ \ka^*_N(s)
    \qquad & (s\in \R),
    \\
\label{eq:ren-mass-tilde}
    \nu_{c,N}^*+sv_N &= \nu_{1,N}(\tilde\ka^*_N(s))+ \tilde\ka^*_N(s)
    \qquad & (s>0).
\end{alignat}
We use the term ``mass'' for brevity, even though ``mass squared'' would be more accurate, and even though $\ka^*_N(s)$
and $\tilde \ka^*_N(s)$ can be negative.
Recall the definition of the constant $\qLap$ from \eqref{eq:qLapdef},
and of $\lp_N$ in \eqref{eq:lpNdef}.
With an unimportant caveat (see \cite[(3.47)]{MPS23}), by
\cite[Proposition~3.9]{MPS23} the renormalised masses satisfy
\begin{equation}
    \ka_N^\per(s) = s\hh_N^{-2}L^{-dN}(1+o(1)),	
    \qquad
    \ka_N^\free(s) = s\hh_N^{-2}L^{-dN}(1+o(1))  -\qLap L^{-2N},
    	\label{eq:aN1}
\end{equation}
\vspace{-20pt}

\begin{equation}
    \tilde\ka_N^\per(s) = s\lp_N^{-2}L^{-dN}(1+o(1)),
    \qquad
    \tilde\ka_N^\free(s) = s\lp_N^{-2}L^{-dN}(1+o(1))  -\qLap L^{-2N}.
		\label{eq:aN2}
\end{equation}
It follows from \eqref{eq:aN1}
(together with the fact that $\qLap$ is as close as desired to $1$ for large $L$)
that for $s \in \R$ and for $N$ large enough depending on $s$,
the mass $\ka_N^*(s)$ is in the domain $\II_{\rm crit}$
treated by
Theorem~\ref{thm:rg_main_theorem_new}.
The same is true for $s>0$ and $\tilde \ka_N^*(s)$.

\subsubsection{Scaling of integrals}

The following is a special case of \cite[Lemma~4.1]{MPS23},
which has been simplified by imposing stricter hypotheses that are nevertheless adequate for our needs.

\begin{lemma}
\label{lem:integral-scaling}
Let $d\geq 4$.  Let
$\lambda_N$ be a positive real sequence with $\lim_{N\to\infty}\lambda_N = \lambda >0$,
and let $V_{N,\bulk}$ and $K_{N,\bulk}$ obey the estimates of
Theorem~\ref{thm:rg_main_theorem_new}.
Let $h:\R^n \to [0,\infty)$ be bounded by $C |y|^C$ for some $C>0$.
Then
\begin{align}
\label{eq:t_N_limit_l_scaling}
	\lim_{N\to\infty}
    \int_{\R^n}h(y)e^{-V_{N,\bulk}(\lp_Ny)}
    e^{-\frac12 \lambda_N   |y|^2}   dy
	&=
	\int_{\R^n}h(x)e^{- \frac12 \lambda |x|^2}dx
        ,
\\
\label{eq:t_N_limit_h_scaling}
	\lim_{N\to\infty}
    \int_{\R^n}h(y)e^{-V_{N,\bulk}(\hh_Ny)}
    e^{-\frac12 \lambda_N   |y|^2}   dy
	&=
	 \int_{\R^n}h(x)e^{- \frac14 |x|^4- \frac12 \lambda |x|^2}dx
        ,
\\
\label{eq:K_N_limit_l_scaling}
	\lim_{N\to\infty}
    \int_{\R^n}h(y)K_{N,\bulk}(\lp_Ny)
    e^{-\frac12 \lambda_N   |y|^2}  dy
	&=
	 0,
\\
\label{eq:K_N_limit_h_scaling}
	\lim_{N\to\infty}
    \int_{\R^n}h(y)K_{N,\bulk}(\hh_Ny)
    e^{-\frac12 \lambda_N   |y|^2}  dy
	&=
	 0
        .
\end{align}
\end{lemma}

Lemma~\ref{lem:integral-scaling} is an elementary consequence of
Theorem~\ref{thm:rg_main_theorem_new}.  For example, the $d=4$ case of
\eqref{eq:t_N_limit_h_scaling} holds because the formulas for
$V_N$ and $\hh_N$ in \eqref{eq:V_N_form} and \eqref{eq:hNdef} imply that
\begin{equation}
    V_{N,\bulk}(\hh_Ny)
    =
    L^{4N}
    \Big(  \textstyle{\frac{1}{4}} g_N(BN)L^{-4N} |y|^4
    + \textstyle{\frac12} \nu_N (BN)^{1/2}L^{-2N}|y|^2 \Big),
\end{equation}
and the bounds of Theorem~\ref{thm:rg_main_theorem_new} then show that
$V_{N,\bulk}(\hh_Ny) \to \frac 14 |y|^4$.
Similarly, with the scaling by $\lp_N$ in \eqref{eq:t_N_limit_l_scaling},
$V_{N,\bulk}(\lp_Ny) \to 0$.
The zero limits in \eqref{eq:K_N_limit_l_scaling}--\eqref{eq:K_N_limit_h_scaling}
simply follow from the bound \eqref{eq:K_bound_at_nu_c_ter} on $K_{N,\bulk}$.

\subsubsection{A short summary}

The proof of Theorem~\ref{thm:rg_main_theorem_new} is a substantial endeavour
which occupies much of \cite{MPS23}.  In turn, \cite{MPS23} depends in a significant
way on results from \cite{BBS-brief}.  The mass restriction $\ka \in \II_{\rm crit}$
in Theorem~\ref{thm:rg_main_theorem_new} reflects the fact that our
main results Theorems~\ref{thm:mr-plateau} and \ref{thm:mr-plateau-Gaussian}
concern near-critical behaviour (including a peek into the
low-temperature regime $\nu<\nu_c$ in Theorem~\ref{thm:mr-plateau}).
In \eqref{eq:Z_N_definition}, the term $|\varphi|^2$ in the action contains a contribution both
from the covariance (coefficient $\ka$) and from $V_\bulk$ (coefficient $\nu$).  In order
to construct an RG flow to scale $N$, it is essential that the initial value
of $\nu$ be carefully tuned depending on the value of $\ka$.
This tuning is achieved by the choice of $\nu_{1,N}(\ka)$
in Theorem~\ref{thm:rg_main_theorem_new}, which continues to serve its
purpose for our present needs.

The proof of Theorem~\ref{thm:rg_main_theorem_new}
involves a progressive integration in which the field $\varphi$ is decomposed over constant fields on blocks of side $L^j$, which we refer to as scale $j$.  The integration over
fields is then done scale by scale.  A polynomial $V_{j,\bulk}$ captures the leading
part of the integration up to scale $j$, with $K_{j,\bulk}$ comprising error terms.
The coefficients $g_j,\nu_j$
of the effective potential $V_{j,\bulk}$ need to be accurately tracked, and
\eqref{eq:gNB_ter}--\eqref{eq:nu_j_bound_ter} give the result at the final
scale $N$.  A specially constructed norm is used to control the error terms $K_{j,\bulk}$,
and we will extend this norm in Section~\ref{sec:norms}.

An inverse problem must be solved to apply Theorem~\ref{thm:rg_main_theorem_new}.
The theorem tells us how to integrate to scale $N$ if we are given $\ka$
and choose $\nu=\nu_{1,N}(\ka)$.  However, our main results
Theorems~\ref{thm:mr-plateau} and \ref{thm:mr-plateau-Gaussian}
concern $\nu=\nu_c+sw_N$ or $\nu=\nu_c+sv_N$, so in order to apply
\eqref{eq:Gobs} at these values of $\nu$ we must write these values
in the form $\nu_{1,N}(\ka)+\ka$, which entails solving the
equations \eqref{ren-mass}--\eqref{eq:ren-mass-tilde} for the renormalised
masses.  This task is accomplished in \cite{MPS23} and we have no further work to
do here in this regard.

What we do need to do here is to extend Theorem~\ref{thm:rg_main_theorem_new}
to incorporate the observables.  We state the extension below in
Theorem~\ref{thm:RG-obs}, and show how the extended theorem suffices
to prove our main results.  Our main work will then be to prove
Theorem~\ref{thm:RG-obs}.

\subsection{Integration to the final scale: observables}
\label{sec:itgfso}

\subsubsection{Result of the RG flow with observables}

The following theorem extends Theorem~\ref{thm:rg_main_theorem_new} to include the observables and with good bounds.
We defer discussion of its proof to Section~\ref{sec:3}.
Note that there is no renormalisation of the observable terms in $V_N$:
they maintain exactly the same form as in $V_0$.
In \eqref{eq:KNoxbd}, $\kappa$
 is the positive constant of Theorem~\ref{thm:rg_main_theorem_new};
 it depends only on $n$.
The $\cO_{N,\x}$ is from Definition~\ref{def:errorconv} and
is not uniform in $s$.

\begin{theorem}
\label{thm:RG-obs}
Let $x\in\Lambda_N$ and $V_0( \varphi)= V_{0, \bulk} (\varphi)
- \sigma_{\o} \varphi_\o^{(1)} - \sigma_{\x} \varphi_\x^{(1)}$.
Under the hypotheses of Theorem~\ref{thm:rg_main_theorem_new}, $Z_N$ can be written as
\begin{equation}
\label{eq:ZNy}
    Z_{N} (y)
    =
    e^{\bar u_N   }
    (e^{-V_{N}(y)} +K_{N}(y) ),
\end{equation}
with $V_{N,\bulk}$ and $K_{N,\bulk}$ as in Theorem~\ref{thm:rg_main_theorem_new},
with $V_N( y)= V_{N, \bulk} (y)
- \sigma_{\o} y^{(1)} - \sigma_{\x} y^{(1)}$, and
with $\bar u_N  = \bar u_{N,\varnothing} + \sigma_{\ox}u_{N,\ox}$.
For the observable part of $\bar u_N$,
there is a function $\vv : \Lambda_\infty \rightarrow \R$ with
$\vv (\x) = O(g)$ (uniformly in $x$) and
$\lim_{|x| \rightarrow \infty} |\x|^{d-2} \vv (\x) = 0$, such that,
\begin{align}
	\label{eq:qN-thm}
	u_{N,\ox} & = C_{\ka,\le N}(\x)
    	+ \vv (\x ) +  \frac{1}{|x|^{d-2}} \cO_{N,\x} .
\end{align}
The observable part of $K_N$ obeys
\begin{align}
	\label{eq:KNoxbd}
    |K_{N,\o\x}(y)|
    & \le
    \const \; g_N^{1/4} 4^{-(N - j_{\o\x})}\frac{1}{|x|^{d-2}}
    e^{-  \kappa |y/\hh_N|^{4}}
    \times
    	\begin{cases}
    		1 & (d=4) \\
    		|x|^{-1/4} & (d \ge 5) .
    \end{cases}
\end{align}
\end{theorem}

The term $K_{N,\ox}$ in \eqref{eq:KNoxbd} is an error term; it decays in $\x$ more rapidly than $|\x|^{-(d-2)}$.  As we will see in Section~\ref{sec:pf-mr},
the term $C_{\ka,\le N}(\x)$ in \eqref{eq:qN-thm} contributes
to the covariance term implicit in \eqref{eq:mr-plateau}.

The proof of Theorem~\ref{thm:RG-obs} also
provides bounds on $K_{N,\o}$ and $K_{N,\x}$, but we do not need them at this point.
To analyse the $u_{N,\ox}$ term, we use the following lemma concerning
the hierarchical covariances. We defer its proof to Appendix~\ref{app:covariance}.
There are four cases, corresponding to the dependence of $u_{N,\ox}$
on $a_N^\per(s)$, $a_N^\free(s)$ for the non-Gaussian case, and on
$\tilde a_N^\per(s)$, $\tilde a_N^\free(s)$
for the Gaussian case.

\begin{lemma}
\label{lem:uNox}
Let $s\in\R$ and recall $\cO_{N,\x}$ from Definition~\ref{def:errorconv}.
Then
\begin{align}
	\label{eq:uNox-PBC}
    u_{N;\ox}(a_N^\per(s)) & =
	\mathbb{C}_{0,\infty} (\x) (1+\cO_{N,\x} ) + \vv (\x)
    + O( |x|^{-(d-3)} L^{-N})
    ,
    \\
	\label{eq:uNox-FBC}
    u_{N;\ox}(a_N^\free(s)) & =
	\mathbb{C}_{0,\infty} (\x) (1+\cO_{N,\x} ) + \vv (\x)
    + O( |x|^{-(d-3)} L^{-N} )
   	.
\end{align}
For $s_N>0$ such that $s_N \rightarrow s >0$ and for all $x \in \Lambda_N$,
\begin{align}
\label{eq:uNox-PBC-tilde}
    u_{N;\ox}(\tilde a_N^\per(s_N))+ \frac {1}{s_{N}} L^{-(d-2)N}
    & =
    \HLap^{\per}_{sL^{-2N},   N} (\x) \big(1+ {\cO}_{N,\x} \big) + \vv (\x)
    \asymp |x|^{-(d-2)}
    ,
    \\
\label{eq:uNox-FBC-tilde}
    u_{N;\ox}(\tilde a_N^\free(s_N))+ \frac {1}{s_{N}} L^{-(d-2)N}
    & =
	\HLap^{\free}_{(s - q)L^{-2N} , N } (\x) \big(1+ {\cO}_{N,\x} \big) + \vv (\x)
    \asymp |x|^{-(d-2)}
    .
\end{align}
\end{lemma}

\subsubsection{Proof of main results: the two-point function}
\label{sec:pf-mr}

According to \eqref{ren-mass}--\eqref{eq:ren-mass-tilde},
\begin{align}
\label{eq:Gsa}
    G_{\nu_{c,N}^*+sw_N}^*(\x) &= G_{\nu_{1,N}(\ka_N^*(s))+\ka_N^*(s),N}^*(\x),
    \\
    G_{\nu_{c,N}^*+sv_N}^*(\x) &= G_{\nu_{1,N}(\tilde\ka_N^*(s))+\tilde\ka_N^*(s),N}^*(\x)  .
\end{align}
For the right-hand sides, we use
Lemma~\ref{lem:Gox}, which asserts that if
$\ka > \frac 12 L^{-2(N-1)}$, then for $\x\in\Lambda_N$ the two-point function obeys
\begin{equation}
\label{eq:Gobs-bis}
    G_{\nu+\ka,N}^*(\x)
    =
    \frac{\E_{C_{\ka,\hat N}^*}Z_{N,\ox}(\ka)}{\E_{C_{\ka,\hat N}^*}Z_{N,\bulk}(\ka)}
    .
\end{equation}
With Theorem~\ref{thm:RG-obs} and \eqref{eq:Gsa}, this gives
\begin{equation}
\label{Gox-ratio}
    G_{\nu_{c,N}^*+sw_N}^*(\x)
    =
    \frac{\E^{*,s}_{\hat N} Z_{N,\ox}}
    {\E^{*,s}_{\hat N}Z_{N,\bulk}}
    ,
\end{equation}
where,
with a test function $f$, we have made the abbreviation
\begin{equation}
\label{eq:Estar}
    \E^{*,s}_{\hat N} [f] =
    \E_{C_{\ka_N^*(s),\hat N}^*} [f]
    \propto
    \begin{cases}
    	\int_{\R^n} f(y) e^{-\frac 12 |\Lambda_N|\ka_N^\per(s)|y|^2} dy & (*=\per)
    \\
    	\int_{\R^n} f(y) e^{-\frac 12 |\Lambda_N|(\ka_N^\free(s)+\qLap L^{-2N})|y|^2} dy & (*=\free).
    \end{cases}
\end{equation}
We have left implicit the dependence of $Z_{N,\ox}$ and $Z_{N,\bulk}$
on the mass.
Similar formulas hold when $w_N$ is replaced by $v_N$, with the replacement
of the renormalised mass $\ka_N^*(s)$ by its counterpart $\tilde\ka_N^*(s)$.

It follows from \eqref{eq:ZNy} that
\begin{equation}
    Z_{N,\ox}(y)
    =
    u_{N,\ox} Z_{N,\varnothing} (y)
    +
    e^{\bar u_{N,\varnothing} - V_{N,\varnothing} (y) }
    |y^{(1)}|^2
    +
    e^{\bar u_{N,\varnothing}}
     K_{N,\ox}(y)
    .
\end{equation}
After cancellation of $e^{\bar u_{N,\varnothing}}$ in the numerator and denominator of
\eqref{eq:Gobs-bis}, this gives
\begin{align}
\label{eq:Gdecomp}
    G_{\nu_{c,N}^*+sw_N,N}  (\x)
    =
    u_{N,\ox} & +
    \frac{\E^{*,s}_{\hat N} |y^{(1)}|^2 e^{-{V_{N,\bulk}(y)}}}
    {\E^{*,s}_{\hat N}[e^{-V_{N,\bulk}(y)} + K_{N,\bulk}(y)] }
    +
    \frac{\E^{*,s}_{\hat N} K_{N, \ox}}
    {\E^{*,s}_{\hat N}[e^{-V_{N,\bulk}(y)} + K_{N,\bulk}(y)] }
    .
\end{align}
Everything on the right-hand side of \eqref{eq:Gdecomp}
depends on $\ka_N^*(s)$.
The same formulas hold when the two-point function is
in the Gaussian regime, i.e., evaluated at $\nu_{c,N}^*+sv_N$, after replacement
of $\ka_N^*(s)$ by $\tilde\ka_N^*(s)$.

In brief,
the three terms on the right-hand side of \eqref{eq:Gdecomp} play three different roles:
\begin{enumerate}
\item
Via Lemma~\ref{lem:uNox},
the first term gives rise to the to the appropriate hierarchical Green function
$\HLap$,
with the mass $\ka$  equal to $\ka_N^*(s)$ in the non-Gaussian
regime and equal to $\tilde \ka_N^*(s)$ in the Gaussian
regime.  In the Gaussian case, the term $s^{-1}L^{-(d-2)N}$ in
\eqref{eq:uNox-PBC-tilde}--\eqref{eq:uNox-FBC-tilde} arises from
the second term on the right-hand side of \eqref{eq:Gdecomp}.
\item
The second term gives a constant term
(independent of $\x$), which is the plateau in the non-Gaussian regime,
and $s^{-1}L^{-(d-2)N}$ in the Gaussian regime.
This is true for both FBC and PBC.
The difference between the two regimes comes from the fact that in the Gaussian regime
the $|y|^4$ term in $V_{N,\bulk}$
disappears in the limit in \eqref{eq:t_N_limit_l_scaling}, but it
survives in the limit in
\eqref{eq:t_N_limit_h_scaling}.
\item
The third term is an error term.
\end{enumerate}
The existence of a plateau is due to a competition between the first two terms,
with the plateau arising when the second term dominates the first term.
We now prove our two main theorems, subject to Theorem~\ref{thm:RG-obs}.

\begin{proof}[Proof of Theorem~\ref{thm:mr-plateau}]

The claim is that for $s \in \R$ and $\x \in \Lambda_N$, we have
\begin{align}
\label{eq:mr-plateau-pf}
    G^*_{\nu^*_{c,N}+sw_N,N}(\x)
    & =
    \big( \HLap_{0,\infty} (\x) + \vv (\x) \big) \Big( 1+ \cO_{N,\x} \Big)
    +
    f_n(s) \hh_N^{2} \Big( 1+ \cO_N  \Big)
\end{align}
with $\vv (\x)$ as in Theorem~\ref{thm:RG-obs}.
We analyse the three terms in the formula for $G_{\nu_{c,N}^*+sw_N,N} (\x)$
on the right-hand side of \eqref{eq:Gdecomp}.

By Lemma~\ref{lem:uNox},
the first term satisfies
\begin{align}
	u_{N,\ox}(\ka_N^*(s)) &=
	\HLap_{0,\infty} (\x) (1 + \cO_{N,\x} )+ \vv (\x)
	+ O(|x|^{-(d-3)} L^{-N} )   .
\end{align}
By Young's inequality,
\begin{align}
	\frac{1}{ |x|^{d-3} L^N} \le O\Big( \frac{N^{-\alpha p_1}}{|x|^{(d-3) p_1}} + \frac{N^{\alpha p_2}}{L^{N p_2} }\Big)
\end{align}
for $\alpha \in \R$ and $p_1^{-1} + p_2^{-1} =1$.
For $d=4$, we choose
$\alpha = \frac{1}{8}$ and $p_1 = p_2 =2$ to obtain
\begin{align}
	\frac{1}{|x|^{d-3} L^N} \le O\Big( \frac{N^{-1/4}}{|x|^{2}} + \frac{N^{1/4}}{L^{2N}} \Big) \le
    \frac{1}{|x|^2} \cO_{N}
    +
    \hh_N^{2} \cO_{N} 	\label{eq:young1}
    .
\end{align}
For $d >4$, we choose $\alpha = 0$ and $p_2 = \frac{d}{2} (1 + \eta)$ with $\eta >0$ to obtain
\begin{align}
	\frac{1}{|x|^{d-3} L^N } \le O\Big( \frac{1}{|x|^{(d-3) p_1}} + \frac{1}{L^{\frac{d}{2} N (1+ \eta)}} \Big)
	 \le
    \frac{1}{|x|^{d-2}} \cO_{N,\x}
    +
    \hh_N^{2} \cO_{N}, 	\label{eq:young2}
\end{align}
where the final inequality holds whenever $\eta \in (0, \frac{d-4}{d})$
(corresponding to $p_1>\frac{d-2}{d-3}$).

The second term on the right-hand side of \eqref{eq:Gdecomp} is independent of $\x$ and involves
purely bulk quantities. It can be handled exactly as in \cite[Section~4.2]{MPS23}, as
follows.
After the change of variables
$y \mapsto \hh_N y$, the second term becomes
\begin{align}
    \hh_N^2 \frac
    {\int_{\R^n}
	|y^{(1)}|^2 e^{-V_{N,\bulk}(\hh_N y)}
    e^{-\frac12 \lambda_N |y|^2}dy}
    {\int_{\R^n}\big( e^{-V_{N,\bulk}(\hh_N y)}+K_{N,\bulk}(\hh_N y) \big)
    e^{-\frac12 \lambda_N  |y|^2}dy}
    \quad\text{with}\quad
    \lambda_N = |\Lambda_N| \ka_N^\per(s)\hh_N^2 \sim s .
\end{align}
By \eqref{eq:aN1}, $\lim_{N\to\infty}\lambda_N =  s$.
The same $\lambda_N$ occurs for both FBC and PBC, since for FBC there is a cancellation
of terms $\qLap L^{-2N}$ in \eqref{eq:aN1} and \eqref{eq:Estar}.
Recall the definition of the profile $\uf_n$ in \eqref{eq:profile-def}.
By Lemma~\ref{lem:integral-scaling} and by symmetry,
the above ratio is asymptotic to $\hh_N^2 f_n(s)$
as $N \to \infty$ (which is independent of $\x$), as desired.

The last term
on the right-hand side of \eqref{eq:Gdecomp} is handled similarly as the second term.
By \eqref{eq:KNoxbd},
it can be absorbed into the error term in \eqref{eq:mr-plateau-pf}, since it is small relative to $\HLap_{0,\infty} (\x) \asymp |\x|^{-(d-2)}$
(for $d=4$, recall from \eqref{eq:gNB_ter} that $g_N \asymp N^{-1}$).
This completes the proof.
\end{proof}

\begin{proof}[Proof of Theorem~\ref{thm:mr-plateau-Gaussian}]
The proof follows the steps in the proof of Theorem~\ref{thm:mr-plateau},
with the Gaussian rather than non-Gaussian alternatives.
The second term in \eqref{eq:Gdecomp} now has limit $s^{-1}\lp_N^2=s^{-1}L^{-(d-2)N}$,
which arises from the Gaussian integrals in
\eqref{eq:t_N_limit_l_scaling}--\eqref{eq:t_N_limit_h_scaling}.
Its combination with $u_{N,\ox}$ satisfies
\eqref{eq:uNox-PBC-tilde}--\eqref{eq:uNox-FBC-tilde}.
The third term in \eqref{eq:Gdecomp} is again an error term, and we obtain
the desired formulas for the two-point function.
\end{proof}

\section{RG map with observable fields}
\label{sec:3}

The proof of Theorem~\ref{thm:RG-obs} is achieved via the extension of the RG map used in \cite{MPS23} to include the observable fields.
In this section, we define the extended RG map $\Phi_+$ and introduce the structure
that the definition requires.  The analysis of the RG map
is given in Sections~\ref{sec:RGflow}--\ref{sec:4}.

\subsection{Integration to the coalescence scale}

We perform the Gaussian integral with covariance $C_{\le N}$
progressively.
That is, we set $Z_0=e^{-V_0}$ with $V_0$ given by \eqref{eq:V0def},
and for $j=0,\ldots, \Npp-1$ we write the scale-$(j+1)$ partition function $Z_{j+1}$
as
\begin{equation}
\label{eq:Z_j_definition}
	Z_{j+1}(\varphi) = \E_{C_{j+1}}Z_{j}(\varphi +\zeta )	,  \qquad  j \in \{0, \cdots,  \Npp - 1\} .
\end{equation}
The Gaussian integral on the right-hand side integrates over $\zeta$ with $\varphi$
held fixed, with the covariance $C_{j+1}$ given by \eqref{eq:Cjdef}.
To lighten the notation,  we write \eqref{eq:Z_j_definition} as
\begin{equation}
\label{eq:Z_j_definition+}
	Z_{+}(\varphi) = \E_{C_{+}}Z(\varphi +\zeta )  =
    \Eplus Z(\varphi +\zeta ).
\end{equation}
This is consistent with the notation in \cite{BBS-brief,MPS23}:
when two generic scales $j$ and $j+1$
are under discussion,
we drop the subscript $j$ and replace $j+1$ simply by $+$.

\begin{definition}
A \emph{polymer activity} at scale $j$ is a smooth function
$F : \cB \times (\R^n)^{\cB} \rightarrow \R$ such that $F(b, \varphi)$ only depends on $\varphi(b)$ (and not otherwise on the block $b\in \cB$).
The space of polymer activities is denoted $\cN_\varnothing$.
We also define $\cN_{*} = \sigma_* \cN_{\bulk}$ for $* \in \{\bulk, \o,\x,\ox \}$ and
\begin{align}
	\cN = \cN_{\bulk} + \cN_{\o} + \cN_{\x} + \cN_{\ox}.
\end{align}
\end{definition}

This differs from \cite[Definition~5.1]{MPS23} because we give up some symmetries,
but we will restore them later in Definition~\ref{def:SK}.
The space $\cN$
is a graded commutative algebra, and any $F \in \cN$ can be written uniquely as
\begin{align}
	F
		= F_{\bulk} + \sigma_\o F_{\o} + \sigma_\x F_{\x} + \sigma_\ox F_{\ox}
	.
	\label{eq:Fdecomp}
\end{align}
We define projections by $\pi_* F = \sigma_* F_*$
for each $* \in \{ \bulk, \o,\x,\ox \}$.
Since an element $F(b, \varphi)$ of $\cN_\varnothing$
depends on $\varphi \in (\R^n)^{\cB}$ only via $\varphi(b)$,
we sometimes write its argument as $\varphi_b$ instead of $\varphi$.  For a block $B$ at the next scale, we define $\cB (B)$ to be the set of scale-$j$ blocks in $B$, and
\begin{align}
	F(B,  \varphi) = \sum_{b \in \cB (B)} F(b, \varphi_b) \qquad (B \in \cB_+)
	.
\end{align}

Given $\zeta : \cB \to \R^n$ and $F : (\R^n)^{\cB} \to \R$, we define
$\theta F : \R^n \to \R$ by $(\theta F)(\varphi) = F(\varphi + \zeta)$.
We do not indicate $\zeta$ explicitly in the notation $\theta F$, but it will always
occur inside a convolution integral of the form
\begin{equation}
    (\E_{+}\theta F)(\varphi) = \E_+ F(\varphi + \zeta),
\end{equation}
where the integration on the right-hand side
is over $\zeta$ in $(\R^n)^{\cB}$
with $\varphi\in \R^n$ held constant on blocks in $\cB_+$.
In particular, we can rewrite \eqref{eq:Z_j_definition+} as
\begin{equation}
\label{eq:Ztheta}
	Z_{+}  =
    	\Eplus \theta Z .
\end{equation}
Also,
for $\varphi \in (\R^n)^{\cB}$ and with the integration
over $\zeta \in (\R^n)^{\Lambda}$ with covariance $C_{\ka,\le j}$,
\begin{equation}
    Z_{j} (\varphi)  = \E_{C_{\ka, \le j}} Z_0 (\varphi + \zeta).
    \label{eq:Zjtheta}
\end{equation}

Recall from Section~\ref{sec:model} that the \emph{coalescence scale}
$j_{\ox}$ is the smallest scale such that $\o$ and $\x$ lie
in the same block.
The following lemma shows that the observables play a trivial role
below the coalescence scale.
Since $Z_{j_\ox-1,\bulk}(\varphi)$ has been well understood already
in \cite{MPS23}, we will be left to concentrate on scales above the coalescence
scale, where the observables do play an important role.

\begin{lemma}
\label{lem:initialbulkflow}
For a field $\varphi$ which is constant on
blocks $B$ of scale $( j_\ox-1 ) \vee 0$,
\begin{align}
    Z_{( j_\ox-1 ) \vee 0}(B, \varphi)
    =
    e^{\sigma_\o \1_{\o \in B}   \varphi^{(1)}_\o
    + \sigma_\x \1_{\x \in B}  \varphi^{(1)}_\x  }
    Z_{( j_\ox-1 ) \vee 0 ,\bulk}(B,\varphi).
\end{align}
\end{lemma}

\begin{proof}
 When $j_\ox$ is equal to $0$ or $1$,
there is nothing to prove, so we only consider $j_{\ox} -1 \ge 1$.
Let $\varphi$ be constant on blocks $B\in\cB_{j_\ox -1}$.
By \eqref{eq:Zjtheta},
$Z_{j_\ox - 1} (\varphi) = \E_{< j_\ox} \prod_{B \in \cB_{j_{\ox}-1}} Z_0 (B,\varphi+\zeta )$.
It follows from the discussion in \cite[Section~2.1]{MPS23} that the expectation factors
over blocks, i.e.,
\begin{align}
	\E_{< j_\ox}
    \!\!\prod_{B \in \cB_{j_{\ox}-1}} Z_0 (B,\varphi+\zeta )
    = \prod_{B \in \cB_{j_{\ox}-1}} \!\! \E_{< j_\ox}
    Z_0(B,\varphi+\zeta)
    ,
\end{align}
where $\zeta = \sum_{j < j_\ox}\zeta_j$ with independent Gaussian
random variables $\zeta_j$
of covariance $C_{j}$, and
\begin{equation}
    Z_0(B,\eta) =
    \exp\Big[- \sum_{x\in B}
    \big(\frac 14 g|\eta_x|^4
    + \frac 12 \nu |\eta_x|^2 - \sigma_\o \eta_\o^{(1)}\1_{\o\in B}
    - \sigma_\x \eta_\x^{(1)}\1_{\x\in B}\big)
    \Big].
\end{equation}
Since $j_\ox$ is the first scale for which $\o$ and $\x$ are in the same block,
they are in distinct blocks $B_\o$ and $B_\x$ at scales strictly less than $j_\ox$.
For notational
simplicity, let us fix $B=B_\o$.
Then, by \eqref{eq:sigmaox},
\begin{align}
    Z_0(B_\o,\varphi+\zeta) & =
    e^{\sigma_\o \varphi_\o^{(1)}}
    \left( 1 + \sigma_\o \zeta_\o^{(1)}   \right)
    Z_{0,\bulk}(B_\o,\varphi+\zeta)
    .
\end{align}
Application of $\E_{< j_\ox}$
to the term with $1$ gives the desired result, so our task is
to prove that the remaining term has zero expectation, i.e., that
\begin{align}
\label{eq:zetabo}
    \E_{< j_\ox} \zeta_\o
    Z_{0,\bulk}(B_\o,\varphi+\zeta) & = 0.
\end{align}
Both $Z_{0,\bulk}(B_\o,\varphi+\zeta)$ and
the Gaussian measure with covariance $C_{<j_\ox}$
are invariant under the interchange of
$\zeta_y$ and $\zeta_{z}$ for points $y,z\in \cB(B_o)$.
Therefore we can replace $\zeta_{\o}$
in \eqref{eq:zetabo} by its average over the block $B_\o$.
But this average is almost surely zero
because each random variable
$\zeta_j$ almost surely sums to zero on every $j$-block
(see \cite[(2.14)]{MPS23}).
This completes the proof.
\end{proof}

\subsection{The renormalisation group map}
\label{eq:TheRGmap}

\subsubsection{Effective potentials}

Effective potentials are important examples of polymer activities,
defined for $b \in \cB = \cB_j$ by
\begin{align}
	V_{j,\bulk} (b, \varphi)
		&= L^{jd} \Big( \frac{1}{4} g_j |\varphi|^4 + \frac{1}{2} \nu_j |\varphi|^2 \Big) ,
		\label{eq:V_bulk_form} \\
	V_j (b, \varphi)
		&= V_{j,\bulk} (b, \varphi) - \big( \sigma_\o \1_{\o \in b} + \sigma_\x \1_{\x \in b} \big) \varphi^{(1)}
		,
	\label{eq:V_form} \\
	U_{j,\bulk} (b,\cdot)
        &= V_{j,\bulk} (b,\cdot)
        -
        \delta u_{j, \bulk}
        ,
    \\
    U_j (b,\cdot)
    	& =  U_{j,\bulk} (b,\cdot) - \sigma_\ox \delta u_{j,\ox} \1_{\o,\x \in b}
     \label{eq:U_form}
	.
\end{align}
The real-valued
parameters $(g_j,\nu_j,  \delta u_{j, \bulk},\delta u_{j,\ox})$ are called
\emph{coupling constants}.
We write $\cV_j$ and $\cU_j$ for the vector spaces of polymer activities of the form
\eqref{eq:V_form} and \eqref{eq:U_form}, respectively.

\subsubsection{Localisation}

Given $F \in \cN_\bulk$, the degree $k$ Taylor polynomial of $F$ is defined by
\begin{align}
	\Tay_k F (b, \varphi)
		= \sum_{l=0}^k \frac{1}{l !} D^l F (b, 0) (\varphi^{\otimes l}),
\end{align}
and
the \emph{localisation operator} $\Loc : \cN \to \cN$ is defined
by
\begin{align}
	\Loc F_{\bulk} =  \Tay_4 F_{\bulk}, \qquad
	\Loc  \sigma_* F_{*} = \sigma_{*}  \Tay_0 F_{*}
    \quad (* \in \{ \o,\x,\ox  \})
	.
	\label{eq:Loc_defi}
\end{align}
This is the same as the localisation operator used in \cite{BBS-brief,MPS23} for the bulk,
and it adapts the localisation
operator used in \cite{BBS-saw4,ST-phi4} to incorporate the observables.

The localisation operator serves to extract from a polymer activity its
relevant and marginal parts.
For the bulk, these are the terms in the potential $V_{j,\bulk}$ as well as constants which constitute the vacuum energy $\bar u_{N,\bulk}$.
This is why $\Tay_4$ is used for the bulk: it extracts monomials up to fourth degree.
For the observables, it is only necessary to extract the constants
(the justification for this occurs in Lemma~\ref{lem:1-Loc}).
These constitute the term $u_{N,\ox}$, which
as seen in \eqref{eq:Gdecomp} give rise to the Gaussian part of the two-point function.

\subsubsection{The perturbative RG map}

The \emph{perturbative map} $\Phi_{\pt}:\cV \to\cU_+$ is defined by
\begin{align}
	\Phi_{\pt} (V) = \E_+ \theta V (B) - \frac{1}{2} \Loc \cov_+ ( \theta V (B) ; \theta V (B) ),
	\label{eq:phi_pt_defi}
\end{align}
where $\cov_+(X,Y)=\E_+(XY)-(\E_+X)(\E_+Y)$.
The  formula \eqref{eq:phi_pt_defi} looks the same as in \cite[Definition~5.2.5]{BBS-brief},
but $\Phi_{\pt}$ now acts on a larger space including observables.
The action of $\Phi_{\pt}$ on $\cV_\varnothing$ has already been detailed
in \cite[Proposition~5.3.1]{BBS-brief}.
The next lemma extends the action to include the observables.

\begin{lemma}
\label{lemma:Uplus_pt}
For $V (b, \varphi) = V_{\bulk} - \big( \sigma_\o \1_{\o \in b} + \sigma_\x \1_{\x \in b} \big) \varphi^{(1)}$,
we have
\begin{align}
	\Phi_{\pt} (V)
		= \Phi_{\pt} (V_\bulk) -
    \big( \sigma_\o \1_{\o \in B} + \sigma_\x \1_{\x \in B} \big) \varphi^{(1)}
    - \sigma_{\ox}  \delta u_{\pt,\ox} \1_{\o,\x \in B},
\end{align}
with
\begin{align}
	\delta u_{\pt,\ox}
		=  C_{+} (\x)
		    \qquad (j_+ \ge j_\ox) .
	\label{eq:deltq_pt}
\end{align}
In particular, the coefficient of $\varphi^{(1)}$ does not evolve.
\end{lemma}

\begin{proof}
We only need to consider $(1-\pi_\bulk ) \Phi_{\pt}$, and for this we compute
\eqref{eq:phi_pt_defi} directly.
For $z =o,x$, we have simply $\E_+\theta V_z(B) = V_z(B)$,
since $V_z (B) = - \sigma_\z \one_{z \in B} \varphi^{(1)}$ by \eqref{eq:Fdecomp} and \eqref{eq:V_form}.
The covariance term is
\begin{align}
	\pi_{\z} \cov_+ ( \theta V (B)  ; \theta V (B) )
		= - 2 \sigma_\z  \cov_+ ( \theta \varphi^{(1)} \1_{\z \in B} ; \theta V_{\bulk} (B) )
		.
		\label{eq:Vx}
\end{align}
Therefore, $\Tay_0 \pi_\z \cov_+ ( \theta V (B)  ; \theta V (B) ) = 0$,
since the integration variable has zero average on $B$ (similarly to \eqref{eq:zetabo}).
For the coefficient of $\sigma_\ox$, for $j_+ \ge j_\ox$ and $\o,\x\in B$,
we have
\begin{align}
	\frac{1}{2} \pi_{\ox} \cov_+ (\theta V(B) ; \theta V (B))
		= \sigma_\ox  \E_+ [ \zeta_{b_\o}^{(1)} \zeta_{b_\x}^{(1)} ]
		= \sigma_\ox C_{+} (\x).
\end{align}
This completes the proof.
\end{proof}

\subsubsection{The full RG map}

In \cite[Definition~5.2.8]{BBS-brief},
the bulk renormalisation group map $\Phi_{+, \bulk}$ is defined
as a map
\begin{equation}
    \Phi_{+,\bulk} (V_{\bulk}, K_\bulk) =  (U_{+,\bulk},  K_{+, \bulk} )
    = (\delta u_{+, \bulk},V_{+,\bulk},K_{+,\bulk}),
\end{equation}
with explicit formulas for $(U_{+,\bulk},K_{+, \bulk})$.
In particular,
\begin{align}
    U_{+,\bulk} = \Phi_{\pt,\bulk}(V_\bulk-  \Loc \,(e^{V_\bulk} K_\bulk)).
\end{align}
There is a suitable domain for the bulk effective potential $V_\bulk$
and the bulk non-perturbative coordinate $K_\bulk$, which we will recall and extend
below.
According to Lemma~\ref{lem:initialbulkflow}, it is only the bulk flow that
is needed below the coalescence scale.  We can therefore use the results of \cite{MPS23}
for those scales, based on the renormalisation group map
\begin{align}
\label{eq:Phibulk+def}
    \Phi_+(V_\bulk,K_\bulk) & =
    \Phi_{+,\bulk} (V_{\bulk}, K_\bulk) =  (U_{+,\bulk},K_{+\bulk})
    \quad (j_+<j_\ox) .
\end{align}

Above the coalescence scale, we must include
the observable fields.  For this, we define the \emph{renormalisation group map} (RG map)
\begin{align}
\label{eq:Phi+def}
    \Phi_+  (V,K)
    & =   (U_+,K_+) = (\delta u_+,V_+,K_+)
    			\quad ( j_+ \ge j_\ox ) ,
\end{align}
with the explicit formulas (for $j_+ \ge j_\ox$ and $B \in \cB_+$)
\begin{align}
\label{eq:U+def}
    U_+ & = \Phi_{\pt}(V-  \Loc \,(e^V K)),
    \\
\label{eq:K+ext_field}
	K_{+}(B) &=
e^{-\delta \bar u_+(B)}
\E_{+}\theta
    \bigg(\prod_{b \in \cB(B)} (e^{-V(b)} + K(b)) \bigg) - e^{-V_{+}(B)}  .
\end{align}
We compute $U_+$ explicitly in Lemma~\ref{lemma:Uplus}.  It includes
the polymer activity
$\delta \bar{u}_+ (B)$ defined by
\begin{align}
    \delta \bar{u}_+ (B) = \delta \bar{u}_{+,\bulk} (B) + \sigma_\ox \delta \bar{u}_{+,\ox} (B)
    =
	\delta u_{+,\bulk} |B| + \sigma_\ox  \delta u_{+,\ox} \one_{\o,\x \in B}.
\end{align}
We write the components of $\Phi_+$ as
\begin{equation}
\label{eq:Phi-components}
    \Phi_+
    =
	(\Phi_+^u,\Phi_+^V,\Phi_+^K)
    .
\end{equation}

The formulas \eqref{eq:U+def}--\eqref{eq:K+ext_field}
appear \emph{exactly} the same as those used for the bulk
in \cite[Definition~5.5]{MPS23}, but they strictly extend those bulk formulas
because now $(V,K)$ is an element of a larger space containing components
for the observables.  In Section~\ref{sec:norms}, we equip the larger space with a carefully defined norm which bundles together both the bulk and observable components.
Norm estimates then proceed in much the same way with observables as without,
with the role of the observables primarily noticed in the estimates of
Appendix~\ref{sec:pfSProps}; see also Lemma~\ref{lemma:bulk_to_obs_bound}.

We discuss the domain of $\Phi_+$ below; it requires in particular that
$(V,K)\in \cV \times \cF$ is
such that the expectation in \eqref{eq:K+ext_field} is convergent.
Assuming the convergence of the expectations,
we know from \cite[(5.2.4)]{BBS-brief} that
\begin{align}
	\Eplus \theta  \prod_{b \in \cB(B)} ( e^{-V_{\bulk}(b)} + K_\bulk(b) )
	 	= e^{\delta \bar{u}_{+, \bulk}(B)} ( e^{-V_{+,\bulk}(B)} + K_{+,\bulk}(B) )
 ,
\label{eq:Z+bulk}
\end{align}
and it similarly follows
from \eqref{eq:K+ext_field} and simple algebra that
\begin{align}
    \E_{+}\theta
     \prod_{b \in \cB(B)} (e^{-V(b)} + K(b))
     &=
     e^{\delta \bar u_+(B)}( e^{-V_{+}(B)}+ K_{+}(B))
     \qquad (j_+ \ge j_\ox).
\end{align}
If we are able to iterate the expectation over all scales,
starting from $(V_0,K_0=0)$, then the above gives
\begin{equation}
\label{eq:ZN}
    Z_{\Npp}  = e^{\bar u_N}( e^{-V_{\Npp}(\Lambda_{\Npp} )}+ K_{\Npp} (\Lambda_{\Npp} ))
\end{equation}
with $\bar{u}_{\Npp} = \sum_{j=1 }^{\Npp} \delta \bar{u}_j$,
as claimed in \eqref{eq:ZNy}.
We will justify the iteration in Section~\ref{sec:RGflowexists}.

The subspace $\cS$ of $\cN$ given by the
following definition incorporates some natural symmetries into our formalism.
In particular, each of $\cV$ and $\cU$ is a subspace of $\cS$.

\begin{definition}
\label{def:SK}
The subspace $\cS$ of $\cN$ consists of those $F\in \cN$ for which:
\begin{enumerate}
\item
$F_\bulk (b, \varphi)$ does not depend on $b$, and $F_\bulk (b, M \varphi) = F_\bulk (b, \varphi)$ for any orthogonal transformation $M \in O(n)$.

\item $F(b, -\varphi) = ( F_{\bulk} - F_{\o} -  F_{\x} + F_{\ox} ) (b,\varphi)$.

\item $F_{\ox} (b) =0$ unless $\o,\x \in b$.  For $* \in \{\o,\x \}$, $F_* (b) = 0$ unless $* \in b$.
\end{enumerate}
\end{definition}

The effective potential $U_+$ is given explicitly in the next lemma.
In \eqref{eq:u+def},
$b_{\ox}$ is the unique $j$-block containing both $\o$ and $\x$
(if it exists).
Note that the coefficient of the $\varphi^{(1)}$ term remains constant and does
not flow---there is no non-perturbative contribution.
This is a simplification compared to \cite{BBS-saw4,ST-phi4,BLS20,LSW17}
where the analogous coefficient does flow.

\begin{lemma}
\label{lemma:Uplus}
If $K \in \cS$ then
\begin{align}
	U_{+} (B, \varphi)
		= U_{+,\bulk} (B, \varphi) -
    \big( \sigma_\o \1_{\o \in B} + \sigma_\x \1_{\x \in B} \big) \varphi^{(1)}
    - \sigma_\ox \delta u_{+,\ox}  \1_{\o,\x \in B}
    ,
\end{align}
where
\begin{align}
\label{eq:u+def}
\delta u_{+,\ox}
		&=
		\delta u_{\pt,\ox} + K_{\ox} (b_\ox ,0)
=
		C_{+} (\x) + K_{\ox} (b_\ox ,0).
\end{align}
\end{lemma}

\begin{proof}
By \eqref{eq:U+def} and Lemma~\ref{lemma:Uplus_pt}, we only need to compute $\Loc (e^{V} K)$.
First,
for $\z \in \{\o,\x\}$,
\begin{align}
	{\pi_\z} e^{V} K = \sigma_\z e^{V_{\bulk}} ( K_\z + V_\z K_{\bulk} )	
	.
\end{align}
Since $V,K \in \cS$, each of
$K_\z$ and $V_z$ is zero when $\varphi=0$,
and hence $\Loc {\pi_\z} e^{V} K=0$.
Also,
\begin{align}
	\pi_{\o \x} e^{V} K
		= \sigma_\ox e^{V_{\bulk}} ( V_{\o} V_{\x} K_{\bulk} + K_{\ox} + V_{\o} K_{\x} + V_{\x} K_{\o} )
\end{align}
and thus
\begin{align}
	\Loc \pi_{\o \x} e^{V} K = \sigma_\ox  K_{\ox} (b, 0)
	.
\end{align}
Since $K \in \cS$, by condition~(iii) in Definition~\ref{def:SK}
the right-hand side is zero unless $\o,\x \in b$.
This leads to the desired result, with the second equality in \eqref{eq:u+def}
due to \eqref{eq:deltq_pt}.
\end{proof}

Lemma~\ref{lemma:Uplus} shows that if $K \in \cS$ then $U_+\in \cS$.
The next lemma shows that also $K_+\in \cS$.

\begin{lemma}
\label{lem:KinS}
If $K \in \cS$, and if the integral in \eqref{eq:K+ext_field} converges,
then $K_+ \in \cS$.
\end{lemma}

\begin{proof}
We verify that $K_+$ obeys the requirements of Definition~\ref{def:SK}.
For condition~(i), we observe that $K_{+,\bulk}$ is the same function of $\varphi$
for each block $B$.  Also, the required $O(n)$-invariance
follows from \eqref{eq:K+ext_field} and the $O(n)$-invariance of $V_\bulk$, $K_\bulk$ and the covariance $C_+$.
For condition~(ii),
by Lemma~\ref{lemma:Uplus} and \eqref{eq:K+ext_field}, it suffices to show that
$\E_{+}\theta F \in \cS$, where
$F= \prod_{b \in \cB(B)} (e^{-V(b)} + K(b))$.  But it follows from the symmetry
of the Gaussian integral, together with the fact that $V,K\in \cS$
and hence $F \in \cS$, that
\begin{align}
    (\E_{+}\theta F)(-\varphi) & = \E_{+}F(-\varphi+\zeta)
    = \E_+F(-\varphi-\zeta)
    \nnb & =
    \E_+ \big(F_\bulk(\varphi+\zeta) -F_\o(\varphi+\zeta) -F_\x(\varphi+\zeta) + F_\ox(\varphi+\zeta) \big).
\end{align}
Finally, condition~(iii) can be seen from the explicit formula
\eqref{eq:K+ext_field}.
This completes the proof.
\end{proof}

\subsection{Norms and domains}
\label{sec:norms}

We now extend the norms and domains from \cite{MPS23}
to include the observables above the coalescence scale.

\subsubsection{Bulk norms and domains}

For each $\tilde{a} \in [0,1)$ and $ j \ge 1$,
we define the mass domain
\begin{align}
\label{eq:massdomain}
	I_{j} (\tilde{\ka} )
		= \begin{cases}
			(\frac{1}{2} \tilde{\ka} , 2 \tilde{\ka} ) & (\tilde{\ka}  \in (0,1) ) \\
			(-\frac{1}{2} L^{-2(j-1)} , \frac{1}{2} L^{-2(j-1)} ) & (\tilde{\ka} = 0).
		\end{cases}
\end{align}
We also define the mass scale $j_{\tilde{\ka}}$ to be the greatest integer $j$ such that $L^{2j} \tilde{\ka} \le 1$, and we use a sequence
\begin{align}
	\tilde{\vartheta}_j = 2^{-(j- j_{\tilde{\ka}})_+}
\end{align}
to capture the decay of the covariance after the mass scale.

Let
\begin{align}
        \scale_j = L^{-(d-4) j},
	\label{eq:rhog}
\end{align}
and
let $\beta_j = B (1 + \tilde{\ka} L^{2j})^{-2} \scale_j$
be the sequence from \cite[(5.10)]{MPS23}.
Given $\tilde{\ka}$,
let $\tilde{g}_j$ be the sequence defined recursively by
\begin{align}
	\tilde{g}_{j+1} = \tilde{g}_j - \beta_j \tilde{g}_j^2,
    \qquad \tilde{g}_0 = g .
    \label{eq:tldgj}
\end{align}
For sufficiently small $g$, the sequence $\tilde{g}_j$ satisfies
(cf.\ \cite[Proposition~6.1.3]{BBS-brief})
\begin{align}
	\tilde{g}_{j+1} \le \tilde{g}_j \le 2 \tilde{g}_{j+1}
	\label{eq:tgc}
\end{align}
for all $j$.
We use the same domain for $V$ as was used for the bulk polynomial
$V_\bulk$ in \cite{MPS23}, i.e.,
\begin{align}
\label{eq:cD-def}
	\cD_j
		= \Big\{ (g_j , \nu_j) \in \R^2 : 2k_0 \tilde{g}_j < g_j < \frac{1}{2k_0} \tilde{g}_j ,  \;\; |\nu_j | < \frac{1}{2k_0} \tilde{g}_j \scale_j L^{2j} \Big\}
	,
\end{align}
where $k_0>0$ is a constant that is chosen sufficiently small depending on $n$ in \cite{MPS23}; in particular $k_0 \le \frac{1}{2}$.
This is as in \cite[Section~5.2]{MPS23}.
We say that $V_j \in \cD_j$ whenever the coupling constants defining
$V_j$ lie in $\cD_j$.   The coupling constants for the observable terms in
$V_j$ are always equal to $1$, so they do not enter into the definition of the domain.

The size of a polynomial $V_j$ is measured with
the $\norm{\cdot}_{T_{\varphi} (\kh_\bulk)}$ semi-norm on $\cN_\bulk$, which is
defined (exactly as in \cite[Section~5.2]{MPS23}) as follows.
For $F : \R^n \to \R$, with the norm of an element $\varphi$ of $\R^n$
given by $|\varphi|/\kh_\bulk$ and with the absolute value norm on $\R$,
the semi-norm of $F$ is defined by
\begin{align}
	\norm{F}_{T_{\varphi} (\kh_\bulk)}
    = \sum_{p=0}^{\infty} \frac{1}{p!} \norm{F^{(p)}(\varphi)}
    ,
	\label{eq:T_z_seminorm_definition}
\end{align}
where $F^{(p)}$ denotes the $p^{\rm th}$ Fr\'echet derivative.
Properties of the semi-norm are presented in \cite[Chapter~7]{BBS-brief}.
There are two parameters $\kh_\bulk$ of interest.
As in \cite[Section~5.2]{MPS23}, we use the two choices $\ell_{\bulk,j}$ and $h_{\bulk,j}$ defined by
\begin{align}
\label{eq:hbulk-def}
	\kh_{\bulk, j}
		= \begin{cases}
		\ell_0 L^{-\frac{d-2}{2} j} & (\kh = \ell) \\
		k_0 \tilde{g}_j^{-1/4} L^{-\frac{d}{4} j} & (\kh = h)	,
		\end{cases}
\end{align}
where $\ell_0 = L^{1+ d/2}$ and $k_0$ is as in \eqref{eq:cD-def}.
For example,
\begin{align}
\label{eq:phi4norm}
    L^{dj}\norm{\tilde g_j |\varphi|^4}_{T_0 (\kh_{\varnothing,j})}
    &=
    \begin{cases}
    \ell_0^4 \tilde g_j \scale_j & (\kh_\bulk = \ell_\bulk)
    \\
    k_0^4 & (\kh_\bulk = h_\bulk).
    \end{cases}
\end{align}

We use a norm for $K_\bulk$ that combines the $T_0(\ell)$ norm (to control $K$ near
zero) and the $T_\varphi (h)$ norm (to control $K$ for large fields).
For the latter, we define the \emph{large-field regulator}
\begin{align}
	G_j (b, \varphi) = \exp (- \kappa |\varphi / h_{\bulk,j} |^4)
    \qquad (b \in \cB_j,  \; \varphi \in \R^n),
\label{eq:Gj}
\end{align}
where $\kappa > 0$ is chosen in \eqref{eq:kappa-def} to be small
depending only on the number $n$ of field components.
In particular, $\kappa < 1$.
The field-regulated semi-norm is defined by
\begin{align}
\label{eq:bulk-norm}
	\norm{K_\bulk (b)}_{T^G_{\bulk} (h)} = \sup_{\varphi \in \R^n} G^{-1} (b, \varphi) \norm{K_\bulk (b)}_{T_\varphi (h)}.
\end{align}
Then we define a norm on $\cN_\varnothing \cap \cS$ by
\begin{align}
	\norm{K_\bulk }_{\Wkappa_{\bulk, j}}
		&=
\norm{K_\bulk (b)}_{T_0 (\ell)} + \tilde{g}^{9/4}_j  \scale_{j}^{{\kb}_1} \norm{K_\bulk (b)}_{T^G_{\bulk} (h)}  ,
		\label{eq:Wbulk}
\end{align}
where $\kb_1 (d) = \kaa_1(d) - \kp_1 (d) >0$ with
\begin{align}
\label{eq:kbd1}
	\kaa_1 (d) =&
	\begin{cases}
	3 &  (d=5) \\
	\frac{2d- 6}{d-4} (1- \epsilon_1(d))  & (d \ge 6),
	\end{cases}
	\qquad \quad 	
    \kp_1 (d) =
	\begin{cases}
	\frac{3}{4} &  (5 \le d < 12)  \\
	\frac{d}{2(d-4)}  (1- \epsilon_1(d))  & (d \ge 12) .
	\end{cases}
\end{align}
The parameter $\epsilon_1(d)> 0$ is chosen to ensure that $\kaa_1(d) > 2$ and  $\kp_1(d) > \frac{1}{2}$,
which holds if and only if $\epsilon_1 (d) \in (0, \frac{1}{d-3})$.
The exponents $\kaa_1,\kb_1,\kp_1$ only occur as powers of $\scale_j=L^{-(d-4)j}$,
so they are immaterial for $d=4$ and we therefore do not specify their values for $d=4$.
For $d >4$, the power $\kb_1$ in \eqref{eq:Wbulk} reflects decay of $K_\bulk$ that we do
not believe to be sharp, though it is sufficient for our needs.
Exactly as in \cite[(5.38)]{MPS23}, and with the same $L$-dependent
constant $C_{\rg}$ specified in \cite[Remark~9.3]{MPS23},
we define the bulk domain
\begin{align}
	\cK_{\bulk,j} &= \{ K \in \cN_\varnothing \cap \cS : \norm{K}_{\cW^{\kappa}_{\bulk,j}} \leq C_{\rg} \tilde{\vartheta}_j^3 \tilde{g}_j^3 \scale_j^{\kaa_1}  \}
	\label{eq:bulk_domain}
	.	
\end{align}

\subsubsection{Observable norms and domains}

For $j \ge j_{\ox}-1$,
we use additional parameters for the observable fields.
For the large-field parameter, we define
\begin{align}
	\label{eq:h-sig-alt}
	h_{\sigma,j}
		= \tilde{g}_j^{1/4} L^{\frac{d}{4} ( j_{\ox} - 1) } 2^{ j - j_{\ox} + 1}
		,
\end{align}
which essentially is a constant multiple of \cite[(1.78)]{BS-rg-IE} for $d=4$
and provides an extension to higher dimensions.
For the observable counterpart of the parameter $\ell_\bulk$, we adapt
\cite[(1.78)]{BS-rg-IE} (which is for $d=4$), and define
\begin{align}
	\label{eq:ell-sig-2}
	\ell_{\sigma,j} =
    \begin{cases}
		\tilde{g}_j L^{\frac{6-d}{2} (j_\ox-1)} 2^{ (j - j_\ox +1 )}
    & (d=4,5)
    \\
		\tilde{g}_j L^{ - \frac{d - 6}{2} j} L^{ - \tau (j - j_\ox +1 )}
    & (d\ge 6)
		,
    \end{cases}
\end{align}
where we fix a choice of $\tau \in (0,  1)$ for $d \ge 6$.
In order to obtain a unified formula for $\ell_{\sigma,j}$, for $d=4,5$ we solve
$L^{-\frac{d-6}{2}} L^{-\tau} = 2$ to obtain $\tau = -\frac{d-6}{2} - \log_L 2 \in (0,1)$ for $d=4,5$.  This permits us to write
\begin{align}
	\label{eq:ell-sig-1}
	\ell_{\sigma,j} =
		\tilde{g}_j L^{ - \frac{d - 6}{2} j} L^{ - \tau (j - j_\ox +1 )}
    \qquad (d \ge 4)
		,
\end{align}
with $\tau\in (0,1)$ for $d \ge 6$ and with
$\tau = \frac{6-d}{2} - \log_L 2 \in (0,1)$ for $d=4,5$.
The fact that $\tau$ depends on $L$ for $d=4,5$ will not be harmful:
for large $L$,
$\tau$ is close to $1$ for $d=4$ and is close to $\frac 12$ for $d=5$.
By definition,
\begin{align}
	\label{eq:exbd7}
\ell_{\sigma, j} \le h_{\sigma, j}
	.
\end{align}
Indeed, for $d=4,5$ this is immediate from \eqref{eq:h-sig-alt}--\eqref{eq:ell-sig-2} because $\tilde g_j$ is small
and $\frac{6-d}{2} \le \frac d4$ for all $d\ge 4$, and for
$d \ge 6$ it is evident from the fact that $\ell_{\sigma,j}$ is decaying
with the scale whereas $h_{\sigma,j}$ is growing.

We use
$\kh_\sigma$ to denote either of $\ell_\sigma$ or $h_\sigma$, and we use
$\kh$ to denote the pair $(\kh_{\bulk}, \kh_{\sigma})$.
We define norms on the observable fields by
\begin{align}
	\norm{\sigma_{\0}}_{T_{\varphi} (\kh)}
    = \norm{\sigma_{\x}}_{T_{\varphi} (\kh)} = \kh_{\sigma} , \qquad
	\norm{\sigma_{\ox}}_{T_{\varphi} (\kh)} = \kh_{\sigma}^2.
\end{align}
Then we define semi-norms on $\cN$ by
\begin{align}
	\norm{F}_{T_{\varphi} (\kh)}
		&=
    \sum_{* \in \{ \bulk,\o,\x,\ox \}}  \norm{F_*}_{T_{\varphi} (\kh)}
    \norm{\sigma_*}_{T_{\varphi} (\kh)}
    \nnb & =
    \norm{F_\varnothing}_{T_{\varphi} (\kh)}
    +
    \norm{F_\o}_{T_{\varphi} (\kh)}\kh_{\sigma}
    +
    \norm{F_\x}_{T_{\varphi} (\kh)}\kh_{\sigma}
    +
    \norm{F_\ox}_{T_{\varphi} (\kh)}\kh_{\sigma}^2
		.
		\label{eq:T_norm}
\end{align}
Since $\ell_{\bulk,j} \le h_{\bulk,j}$ (by definition)
and $\ell_{\sigma,j} \le h_{\sigma,j}$ (by \eqref{eq:exbd7}), it follows from the
definitions \eqref{eq:T_z_seminorm_definition} and \eqref{eq:T_norm} that for all $\varphi$ and for all $F \in \mathcal N$,
\begin{equation}
\label{eq:Tnorm-ell-h}
    \|F\|_{T_\varphi(\ell)} \le \|F\|_{T_\varphi(h)}
    .
\end{equation}

The important product property
\begin{align}
	\norm{F_1 F_2}_{T_{\varphi} (\kh)} \le \norm{F_1}_{T_{\varphi} (\kh)}  \norm{F_2}_{T_{\varphi} (\kh)}
	\label{eq:subx1}
\end{align}
holds for the norm \eqref{eq:T_norm}.
This product property is an instance of the much more general product property \cite[Proposition~3.7]{BS-rg-norm}
(as described in \cite[Section~1.1.6]{BS-rg-IE}).
But it can also be seen directly from the product property for the
norm \eqref{eq:T_z_seminorm_definition}
(see \cite[Lemma~7.1.3]{BBS-brief}) by applying the latter term by term in
the analogue of \eqref{eq:T_norm} for $F=F_1F_2$.

In particular,
\begin{align}
\label{eq:Vznorm-ell}
	\norm{\sigma_\z  \varphi_\z^{(1)}}_{T_0 (\ell_{j})} &=
		\ell_{\bulk,j}\ell_{\sigma,j} =
		\ell_0\tilde{g}_j \scale_j L^{-\tau (j - j_\ox +1 )}
		,
\\
\label{eq:Vznorm-h}
	\norm{\sigma_\z  \varphi_\z^{(1)}}_{T_0 (h_{j})}&  =
		h_{\bulk,j}h_{\sigma,j}
		= k_0 (2L^{-d/4})^{ j - j_{\ox} + 1}
		.
\end{align}
Note that in both cases, the norms of $\sigma_\z \varphi_z^{(1)}$
are relatively small compared to the norm of the quartic term
(the main term in $V_\bulk$) in \eqref{eq:phi4norm}.
The size of $\kh_\sigma$ has been limited intentionally for this to happen.
Similarly, for the perturbative value $\delta u_{\pt,\ox}$ given in \eqref{eq:deltq_pt},
since $L^{j_{\ox}} \asymp |x|$,
\begin{align}
    \norm{\sigma_\ox C_{\ka, j+1}(\x)}_{T_0 (\ell_{j})}
    &
    \le
    O(L^{-(d-2)j}) \tilde{g}_j^2 L^{ - (d - 6)j} L^{ - 2\tau (j - j_\ox +1 )}
    \asymp
    \tilde{g}_j^2  \frac{1}{|x|^{2(d-4)}} \Big( \frac{2}{L^{2(d-4+\tau)}}\Big)^{j-j_\ox}
\end{align}
(with $L$ dependent constants for $\asymp$),
which is much smaller than \eqref{eq:phi4norm}.

Now we extend the bulk norm \eqref{eq:bulk-norm} by using the $T_\varphi(\kh)$
semi-norm from \eqref{eq:T_norm} to include observables, i.e.,
\begin{align}
\label{eq:norm-with-half}
	\norm{K (b)}_{T^G (h)} = \sup_{\varphi \in \R^n} G^{-\frac{1}{2}} (b, \varphi) \norm{K (b)}_{T_\varphi (h)}
	.
\end{align}
The power $-1$ on the regulator for the bulk which appears in \eqref{eq:bulk-norm} has been
reduced to the power $-\frac{1}{2}$ in \eqref{eq:norm-with-half}, in order to
provide room to multiply
a bulk polymer activity by
an observable polynomial in $\varphi$ while keeping the norm finite.
We define $\kb_2 = \kaa_2 - \kp_2$ with
\begin{align}
\label{eq:kbd2} 	
	{\kaa}_2 (d) =
		\frac{\tau + d-4}{d-4} (1- \epsilon(d))  \quad (d \ge 5),
    \qquad
	{\kp}_2(d) =
	\begin{cases}
		\frac{3}{4} & (d=5) \\
		\frac{d}{4(d-4)}  (1- \epsilon(d))  & (d \ge 6)
	\end{cases}
\end{align}
for some fixed $\epsilon (d) \in (0, \frac{\tau}{d-4 + \tau})$.
Again we do not specify values of $\kaa_2,\kb_2,\kp_2$
for $d=4$ since they are exponents of $\scale_j=L^{-(d-4)j} = 1$.
An important role is played by the scale immediately preceding the
coalescence scale $j_\ox$, and in order to accommodate the degenerate case
 $\x = \o$ (for which $j_\ox=0$),  we define
\begin{align}
	\joxm = \begin{cases}
		j_{\ox} - 1 & (\x \neq \o)  \\
		0 & (\x = \o) .
	\end{cases}	\label{eq:joxm}
\end{align}
Then we define
\begin{align}
\label{eq:rhotil-def}
	\tilde{\scale}_j ( {\kaa})
		= \scale_{\joxm}^{{\kaa}_1} \scale_{j -\joxm}^{ {\kaa}_2},
		\qquad
	\tilde{\scale}_j ( {\kb})
		= \scale_{\joxm}^{ {\kb}_1} \scale_{j-\joxm}^{ {\kb}_2} ,
\end{align}
which are interpreted as equal to $1$ for $d=4$.
Then we define a norm on $\cN$ by
\begin{align}
\label{eq:Wkappa-def}
	\norm{K}_{\Wkappa_j}
		&=
\max_{b \in \cB} \Big(
\norm{K (b)}_{T_0 (\ell)}
+ \tilde{g}^{9/4}_j  \tilde{\scale}_j ({\kb})  \norm{K (b) }_{T^G (h)} \Big),
\end{align}
and, given $\tilde{C}_{\rg}>0$,
the domain
\begin{align}
	\cK_j &= \{ K \in \cS  : \norm{K}_{\Wkappa_j} \leq \tilde{C}_{\rg}  \tilde{\vartheta}_j^3  \tilde{g}_j^3 \tilde{\scale}_j ({\kaa}) , \; K_\bulk \in \cK_{\bulk,j}\}.
		\label{eq:tilde_cK_def}
\end{align}
The domain of the RG map is then defined to be
\begin{align}
\label{eq:domRG-def}
	\domRG_j = {\cD}_j \times  {\cK}_j .
\end{align}
We make the choice
\begin{align}
\label{eq:CRGtilde}
    \tilde C_{\rg}= \max\{ 2 C_{\pt} ,  C_{\rg}' \},
    \qquad
    C_{\rg}' =
	(4\kappa^{-1})^2
    C_{\rg},
\end{align}
where $C_{\rg}$ is from \eqref{eq:bulk_domain},
and $C_{\pt}$ is the specific $L$-dependent constant determined in Proposition~\ref{prop:Phi+0}.
This specific choice plays a role in
Lemma~\ref{lemma:bulk_to_obs_bound} and
Proposition~\ref{prop:stable_manifold}.
Since $K$ now also contains the observables, it is natural that its norm
increases compared to $K_\bulk$, and this is the reason that we need the larger
constant $\tilde C_{\rg}$.

The definition of the norm in \eqref{eq:Wkappa-def} is elaborate, but its sophistication
encodes many estimates simultaneously.  As an example, if $K \in \cK$ then
at scale $j\ge j_\ox$ for $d=4$ (for which $\tilde\scale =1$), and using $\tilde \vartheta \le 1$,
\begin{align}
    |K_{j,\ox} (b,0)| \le \norm{K_{j,\ox} (b)}_{T_0 (\ell) }
    & \le \ell_{\sigma,j}^{-2} \norm{K_j}_{\Wkappa}
    \le
    \ell_{\sigma,j}^{-2} \tilde{C}_{\rg}  \tilde{g}_j^3
    =
    \tilde{C}_{\rg}  \tilde{g}_j L^{-2j}
    L^{2\tau(j-j_\ox+1)}.
\end{align}
Since $\tau \in (0,1)$, the right-hand side is summable over $j \ge j_\ox$.
This fact
is an essential
ingredient in the control of the Gaussian
contribution $u_{N,\ox}$ to the two-point function.
This will become apparent in Section~\ref{sec:pf_thm:RG-obs}, see in particular
\eqref{eq:KjoxAbsconv}.
In a similar way, as discussed further at \eqref{eq:KNoxy}, for $d=4$
and for $K_N \in \cK_N$,
\begin{align}
\label{eq:KNybd-preliminary}
	|K_{N,\ox} (\Lambda_N,  y)|
    & \le h_{\sigma,N}^{-2} \tilde{g}_N^{-9/4}  \norm{ K_{N} }_{\Wkappa_N}
		 G_N^{\frac{1}{2}}   (\Lambda_N, y)
    \nnb & \le
  \tilde  C_{\rg} \tilde{g}_N^{1/4} L^{-2(j_\ox-1)} 4^{-(N-j_\ox +1)}
     G_N^{\frac{1}{2}}   (\Lambda_N, y)
	\nnb &	\le \const \,  \tilde{g}_N^{1/4} |\x |^{-2} 4^{-(N-j_{\ox})}
		 e^{- \kappa \tilde g_N |\Lambda_N| |y|^4}
    .
\end{align}
Here the bound $G_N^{\frac{1}{2}}   (\Lambda_N, y) \le e^{- \kappa \tilde g_N |\Lambda_N| |y|^4}$
follows from $h_{N,\bulk}^{-4} \le L^{dN} \tilde{g}_N$ since $k_0 \le 1/2$.

The specific choices we make for $\kaa,\kb$ in $\tilde \scale$
are non-canonical devices that work for $d>4$.

\subsubsection{Parameter restrictions}
\label{sec:prrt}

For later reference,
we summarise a list of restrictions that the parameters $\tau, \epsilon_1$ and $\epsilon$ entail.
Recall that $\tau\in (0,1)$ for $d \ge 6$ and
$\tau = \frac{6-d}{2} - \log_L 2 \in (0,1)$ for $d=4,5$.
The previously stated conditions $\epsilon_1 \in (0, \frac{1}{d-3})$
and  $\epsilon \in (0, \frac{\tau}{d-4 + \tau})$ imply the following statements.
For $d \ge 5$, the exponents $\kaa_1,\kb_1,\kp_1$ are defined at \eqref{eq:kbd1}, and
$\kaa_2,\kb_2,\kp_2$ are defined at \eqref{eq:kbd2}.
For $d=4$, they play no role and are therefore not defined.

The first three restrictions are purely on the exponents, so we only consider $d \ge 5$.

\begin{enumerate}
\item
Since $\epsilon_1 < \frac{1}{d - 3}$, we have
\begin{align}
	 \kaa_1 > 2 ,\qquad \kp_1 > \frac 12
	\label{eq:exbd5}
	.
\end{align}
By \eqref{eq:kbd1} and \eqref{eq:kbd2},
\begin{equation}
\label{eq:kb12}
    0 \le  \kb_2 \le \kb_1 < \frac 52.
\end{equation}

\item We have
\begin{align}
	(d-4) \kp_1 + \frac{d}{2} \ge
    d- \frac{7}{4}
		.
	\label{eq:exbd6}
\end{align}
For $d =5$,  this follows
from $(d-4) \kp_1 + \frac{d}{2} = \frac{5}{4} d -3$.
For $d \ge 6$,
we have $(d-4) \kp_1 + \frac{d}{2} = d - \frac{d}{2} \epsilon_1 (d)$,
and the desired bound is equivalent to $\epsilon_1 < \frac{7}{2d}$.
Since $\frac{1}{d-3} < \frac{7}{2d}$ for $d > \frac{21}{5}$,
we have the desired bound.

\item For both $i=1,2$, we have
\begin{align}
	3 \ge \kaa_i
    > \max\{ 1, \kb_i \}
	\label{eq:exbd1}
	.
\end{align}
Since $3 \ge \kaa_i$ by definition, we only need to check
that $\kaa_i > \max\{1,\kb_i\}$.
The bound $\kaa_i > \kb_i$ holds because $\kaa_i - \kb_i = \kp_i > 0$.
Also, by (i), $\kaa_1 > 2 >1$.
Finally,
by definition,
$\kaa_2 (d) = (1 + \frac{\tau}{d-4}) (1-\epsilon) > 1$ whenever $\epsilon < \frac{\tau}{d-4 + \tau}$.

\setcounter{savedenum}{\value{enumi}}
\end{enumerate}

The next three restrictions involve $\tau$,
and are for all $d \ge 4$.
The exponents $\kaa_i$ and $\kb_i$ always occur in a product with $d-4$,
so their values are not needed for $d=4$.

\begin{enumerate}
\resumeenumerate

\item There is a $d$-dependent choice of $t = t(d) >0$ such that
\begin{align}
	\label{eq:exbd2}
	\min \Big\{ \tau + d - 4  ,  \frac{d}{4} + (d-4) \kb_2 \Big\}
    \ge  (d-4) \kaa_2 + t .
\end{align}
 For $d=4$, the left-hand side of \eqref{eq:exbd2}
is $\min\{\tau,1\}$ and the right-hand side is $t$.
For $d=5$,
$(d-4)\kb_2 = \frac94$ and $(d-4) \kaa_2 = (\tau+1)(1-\eps)$.
Finally, for $d \ge 6$, $(d-4)\kb_2 = (\tau + \frac{3d}{4} -4) (1-\eps)$ and $(d-4) \kaa_2 = (\tau+d-4)(1-\eps)$.
In all cases, \eqref{eq:exbd2} is satisfied.

\item
We have $(d-4) \kaa_2 - (2\tau + d-6) > 0$, since
\begin{align}
	(d-4) \kaa_2 - (2\tau + d-6) =
		\begin{cases}
			2 - 2\tau & (d=4) \\
			(2- \tau) - \epsilon (\tau + d -4)	& (d>4)	
		\end{cases}
		\label{eq:exbd3}
\end{align}
is positive whenever $\tau < 1$ and $\epsilon < \frac{\tau}{\tau + d-4} < \frac{2-\tau}{\tau + d -4}$.

\item The inequality
\begin{align}
	0<  d-4 + \tau \le 2d - 6
	\label{eq:exbd4}
\end{align}
holds with room to spare when $d\ge 4$ since $\tau <1$.
\end{enumerate}

\section{RG flow with observables}
\label{sec:RGflow}

We now establish conditions which allow the RG map to be iterated over scales
and thereby prove Theorem~\ref{thm:RG-obs}.
The main ingredient in the proof of Theorem~\ref{thm:RG-obs} is Theorem~\ref{thm:Phi^K_q0} for the RG map, which
extends a related theorem without observables in \cite{BBS-brief,MPS23}.
Theorem~\ref{thm:Phi^K_q0} is proved in Section~\ref{sec:conditional-proof} subject to Propositions~\ref{prop:Phi+0}--\ref{prop:crucial-short-3},
whose proofs are deferred to Section~\ref{sec:4}.
The proof of Theorem~\ref{thm:RG-obs} is given in Section~\ref{sec:pf_thm:RG-obs}
subject to Theorem~\ref{thm:Phi^K_q0}.

\subsection{Existence of the RG flow}
\label{sec:RGflowexists}

The repeated iteration of the RG map is called an RG flow.  We formalise this notion in the following definition.
The mass parameter $\ka$ mentioned in the definition is the mass parameter defining the covariance used in the RG map,
and the mass domain $\II_{\Npp} (\tilde{\ka})$ is defined in \eqref{eq:massdomain}.
Also, recall $\joxm$ from \eqref{eq:joxm}.

\begin{definition}
\label{def:RGflow}
Let $\tilde{\ka} \in [0,1)$ and $\ka \in \II_{\Npp} (\tilde{\ka})$.
Given initial coupling constants $g_0>0$, $\nu_0\in \R$ defining $V_0$
and $K_{0,\bulk} =0$,
we say that
$(V_{j,\bulk}, K_{j,\bulk})_{j < j_{\ox}}$ and
$(V_j,K_j)_{j_\ox \le j \le k}$ is an \emph{RG flow to scale} $k\in\N$ \emph{with initial condition} $(g_0,\nu_0)$
and mass parameter $\ka$,
if the sequence is determined by the recursion
\begin{alignat}{3}
	(V_{j+1,\bulk},K_{j+1,\bulk})
		& =  \big(\Phi_{j+1,\bulk}^V (V_{j,\bulk},K_{j,\bulk}) , \Phi_{j+1,\bulk}^K(V_{j,\bulk},K_{j,\bulk}) \big)
    	\qquad & ( & j < \joxm ), \\
    (V_{j+1},K_{j+1})
    	& = \big(\Phi_{j+1}^V(V_j,K_j) , \Phi_{j+1}^K(V_j,K_j ) \big)
    	\qquad & ( & \joxm \le j \le k  - 1 ),
\end{alignat}
with
\begin{alignat}{3}
	( V_{j,\bulk}, K_{j,\bulk} )
		& \in \cD_j \times \cK_{j,\bulk}
	\qquad & ( & j < \joxm) , \\
    (V_j,K_j)
    	& \in \cD_j \times \cK_j
	\qquad & ( & \joxm \le j \le k-1) .
\label{eq:RGmap_jox}
\end{alignat}
When an RG flow exists to scale $k$ for all $k \in \N$ then we refer to it
as a \emph{global RG flow}.
A \emph{bulk} RG flow is an RG flow for the bulk coordinates $(V_{j,\bulk},K_{j,\bulk})$
only, i.e., an RG flow with $\sigma_\o$ and $\sigma_\x$ set equal to zero.
For the transition at the coalescence scale
in \eqref{eq:RGmap_jox}, consistent
with Lemma~\ref{lem:initialbulkflow}, we use
\begin{align}
    V_{\joxm}( B,  \varphi) &= V_{\joxm,\bulk}(B, \varphi)
    - \sigma_\o \one_{\o \in B} \varphi^{(1)}  - \sigma_\x \one_{\x \in B} \varphi^{(1)}
    \\
\label{eq:Kjox-1}
    K_{\joxm}(B,\varphi) &=
    e^{\sigma_\o \1_{\o \in B}\varphi^{(1)}  + \sigma_\x \1_{\x \in B}\varphi^{(1)}}
    K_{\joxm,\bulk}(B,\varphi),
\end{align}
for a block $B$ at scale $j_{\ox ,-}$.
\end{definition}

When an RG flow exists to scale $\Npp$, it follows that
\begin{equation}
    Z_j(\Lambda_{\Npp}) = e^{u_{j,\bulk}|\Lambda_{\Npp}| + \sigma_\ox u_{j,\ox}}
    \prod_{B \in \cB_j} \big( e^{-V_j(B)}+K_j(B) \big)
    \qquad (j \le \Npp),
\end{equation}
where
\begin{equation}
    u_{j,*} = \sum_{i=1}^j \delta u_{i,*}
    \qquad
    (* \in \{\bulk, \ox\}).
\end{equation}
In particular, since at scale $\Npp$ there is a single block,
with $\bar u_{\Npp} (\Lambda_{\Npp}) = u_{\Npp,\bulk}
|\Lambda_{\Npp}|+\sigma_\ox u_{\Npp,\ox}$,
as in \eqref{eq:ZN} we have
\begin{equation}
\label{eq:RGflowN}
    Z_{\Npp} (\Lambda_{\Npp}) = e^{\bar u_{\Npp} (\Lambda_{\Npp})} \big( e^{-V_{\Npp}(\Lambda_{\Npp})} + K_{\Npp} (\Lambda_{\Npp}) \big)
    .
\end{equation}

The following proposition unites special cases of the statements of \cite[Proposition~6.2 and 6.4]{MPS23}.
It provides the existence of bulk RG flows.

\begin{proposition} \label{prop:nu_c}
Let $d \geq 4$,
let $L$ be sufficiently large (depending on $d,n$), and let $g >0$ be sufficiently small
(depending on $L$).
Then there exists $C_{\rg}$ (in the definition \eqref{eq:bulk_domain}
of $\cK_{\bulk,j}$) such that the following hold:

\begin{enumerate}
\item For $\ka \in [0,1)$ and $\tilde\ka=\ka$, there exists
    a continuous function $\nu_c : [0,   \infty ) \rightarrow \R$
    such that
$(V_{j, \bulk}, K_{j,\bulk})$ is a global bulk RG flow with mass $\ka$ and initial condition $(g, \nu_c (\ka))$.

\item For $N \in \N$, $\ka \in ( - \frac{1}{2} L^{-2(\Npp-1)} , 0 ]$,
and $\tilde{\ka} = 0$, there exists
a continuous function
$\nu_{0,\Npp} : ( -\frac{1}{2} L^{-2(\Npp -1)} , 0 ] \rightarrow \R$ such that
$(V_{j,\bulk}, K_{j,\bulk})$ is a bulk RG flow to scale $\Npp$ with mass $\ka$ and initial condition $(g, \nu_{0, \Npp} (\ka))$.
Moreover, $\nu_{0, \Npp} (0) = \nu_c (0)$.
\end{enumerate}
\end{proposition}

The two critical values $\nu_c(a)$ and $\nu_{0,\Npp}(a)$
have appeared already in \eqref{eq:nu_1_N_def}, where they were joined to
define (part of) the continuous function
\begin{equation}
\label{eq:nu_1_N_def-bis}
	\nu_{1,\Npp}(\ka) =
	\begin{cases}
		\nu_c(\ka) & (\ka \in [0,1) )  \\
		\nu_{0,\Npp}(\ka) & ( \ka \in (-\frac12 L^{-2(\Npp - 1)} ,  0) ).
	\end{cases}
\end{equation}
It is also proved in \cite{MPS23} that
 $\nu_{1,\Npp}$ is bounded by $O(g)$ uniformly in $\ka$.
For our purpose,
we will only need $\ka$ in the mass interval $\II_{\rm crit}$ defined
in \eqref{eq:I_crit_asymp}.

As a consequence of Proposition~\ref{prop:nu_c} and Lemma~\ref{lem:initialbulkflow},
we know that for $\ka \in \II_{\rm crit}$
and $\nu=\nu_{1,\Npp}(a)$, we have
\begin{align}
\label{eq:bulkflow}
    Z_{\joxm}(\varphi)
    & =
    e^{\sigma_\o \varphi^{(1)}_\o  + \sigma_\x \varphi^{(1)}_\x }
    e^{u_{\joxm}|\Lambda_N|}
    \prod_{B \in \cB_{\joxm}}
    (e^{-V_{\joxm},\bulk} +K_{\joxm,\bulk})(B),
\end{align}
with
\begin{equation}
\label{eq:bulkflowdomain}
    (V_{\joxm,\bulk},K_{\joxm,\bulk} ) \in \domRG_{\joxm}.
\end{equation}
The sequence $\tilde g$ in the next lemma is given by \eqref{eq:tldgj};
it is an essential ingredient in the definition of the domain $\domRG$.

\begin{lemma} \label{lemma:bulk_to_obs_bound}
 If $\tilde{g}_{\joxm}$ is sufficiently small, and if $K_{\joxm,\bulk} \in \cK_{\joxm,\bulk}$, then  $K_{\joxm} \in \cK_{\joxm}$.
\end{lemma}

\begin{proof}
Suppose first that $j_{\ox}$ is not $0$ (the case of primary interest
is large $j_\ox$).
In this case,
a block $B$ at scale $\joxm = j_\ox -1$ can contain at most one of the points $\o,\x$.
If it contains neither, there is nothing to prove, so let us assume for notational
simplicity that $\o \in B$.
In this case, by \eqref{eq:Kjox-1},
\begin{align}
    K_{j_\ox -1}(B,\varphi) =
    e^{\sigma_\o  \varphi^{(1)}}
    K_{j_\ox -1,\bulk}(B,  \varphi)
    =
    (1+ \sigma_\o  \varphi^{(1)}  )
    K_{j_\ox -1,\bulk}(B,  \varphi).
\end{align}
We take the $T_\varphi(\kh)$ norm at scale $j_\ox -1$ and use the product property
\eqref{eq:subx1}.  With Lemma~\ref{lemma:poly_bound} to bound the $T_\varphi$ norm by the $T_0$ norm,
and with \eqref{eq:Vznorm-ell}--\eqref{eq:Vznorm-h}
to evaluate the $T_0$ norm, this gives
\begin{align}
    \norm{K_{j_\ox -1}(B)}_{T_\varphi(\kh)}
    \le
    \big( 1+ \kh_\bulk \kh_\sigma (1+|\varphi|/\kh_\bulk) \big)
    \norm{K_{j_\ox -1,\bulk}(B)}_{T_\varphi(\kh)} .
\end{align}
Since $1 + x \le 2+x^4 \le \lambda^{-1} e^{\lambda x^4 }$ for $\lambda \in (0,  \frac{1}{2} ]$,
we can elevate the polynomial factor into the exponent,  for both $\kh_\bulk \in \{ \ell_\bulk, h_\bulk  \}$, via
\begin{align}
	1+ \frac{|\varphi|}{\kh_\bulk}
	\le
	\frac{1}{\lambda}
	\exp \Big( \lambda \Big|\frac{\varphi}{ h_{\bulk}}\Big|^4\Big)
		=  \frac{1}{\lambda}  G^{-\frac{\lambda}{\kappa}} (B, \varphi) .
\end{align}
Since  $\kh_\bulk\kh_\sigma \le  \frac{1}{2}$
by \eqref{eq:Vznorm-ell}--\eqref{eq:Vznorm-h},
\begin{align}
    1+ \kh_\bulk \kh_\sigma (1+|\varphi|/\kh_\bulk) )
    \le
    1 + \frac{1}{2} \frac{1}{\lambda} G^{-\frac{\lambda}{\kappa}} (B, \varphi)
    \le
    \frac{1}{\lambda} G^{-\frac{\lambda}{\kappa}} (B, \varphi).
\end{align}
If $j_\ox=0$
then we instead have two such factors.  Therefore, in any case,
\begin{align}
    \norm{K_{\joxm}(B)}_{T_\varphi(\kh)}
    \le
    \frac{1}{\lambda^2} G^{-2\frac{\lambda}{\kappa}} (B, \varphi)
     \norm{K_{\joxm,\bulk}(B)}_{T_\varphi(\kh)} .
\end{align}
We choose $\lambda = \kappa/4$ so that the bad sign in the
exponent of the regulator on the right-hand side becomes $-1/2$.
This is compensated by the change in regulator
from \eqref{eq:bulk-norm} to \eqref{eq:norm-with-half}.
The factor $\lambda^{-2} = (4/\kappa)^2$
is responsible for the appearance of the same factor
in the definition of $C_{\rg}'$ in \eqref{eq:CRGtilde}.
This completes the proof.
\end{proof}

The next theorem is what permits
the RG flow with observables to be extended beyond the coalescence scale.  It
verifies that the $\norm{\cdot}_{\Wkappa_+}$-norm of $K_+$ stays controlled whenever $(V,K) \in \domRG$ and the parameters $L, \kappa$ and $\tilde{g}$ satisfy the following hypotheses.
\begin{quote}
\begin{itemize}
\item[\customlabel{quote:assumPhi}{$\assumPhi$}]

Let $\tilde{\ka} \in [0,1)$,
let $L$ be sufficiently large, and let $\kappa$
be sufficiently small
(the latter depending only on $n$).
At scale $j$, let $\tilde{g} = \tilde{g}_j$ be sufficiently small depending on $L$
(and not depending on $j$),
$\ka \in \II_+ (\tilde{\ka})$,
$(V, K ) \in \domRG = \domRG_j$.

\end{itemize}
\end{quote}

\begin{theorem}
\label{thm:Phi^K_q0}
Let $d \ge 4$.  Let $j \in \{ \joxm  ,\ldots,  \Npp - 1\}$.  Assume
\ref{quote:assumPhi} holds at scale $j$.
There exists a positive constant $C_{\pt}$ such that for $\tilde{C}_{\rg} \ge 2 C_{\pt}$,
 $\Phi_{+}^K$ is a well-defined continuous map $\domRG \rightarrow \Wkappa$ (i.e., the integral \eqref{eq:K+ext_field} converges) that is also continuous in $\ka$,
and for all $(V,K) \in \domRG$ we have
\begin{align}
	\Phi^K_+(V,K) \in \cK_{+}   .  \label{eq:Phi^K_q0}
\end{align}
\end{theorem}

The following proposition shows that the initial condition $\nu_{1,N}(\ka)$ defining the
bulk RG flow also defines an RG flow when the observables are included.
Since the interval $\II_{\rm crit}$ depends on $N$, so does the choice of $\ka$ in the proposition.

\begin{proposition} \label{prop:stable_manifold}
Let $d \geq 4$,
let $L$ be sufficiently large (depending on $d,n$), and let $g >0$ be sufficiently small
(depending on $L$).
Let $C_{\rg}$ and $C_{\pt}$ be as in Proposition~\ref{prop:nu_c} and Theorem~\ref{thm:Phi^K_q0}, respectively and let
$\tilde{C}_{\rg} = \max\{ C_{\rg}'  ,2 C_{\pt} \}$
be chosen as in \eqref{eq:CRGtilde} (for the definition \eqref{eq:tilde_cK_def} of $\cK_j$).
Let $N \in \N$, let $\ka \in \II_{\rm crit}$, and set
$\tilde\ka=\max\{\ka,0\}$.  Then
$(V_{j}, K_{j})$ is an RG flow to scale $\Npp$ with mass $\ka$ and initial condition $(g, \nu_{1,\Npp} (\ka))$.
\end{proposition}

\begin{proof}
By the choice of the parameters,  we can utilise the previous results in this section.
With $\nu_{1,N}(\ka)$ given by \eqref{eq:nu_1_N_def-bis}, we
use Proposition~\ref{prop:nu_c} to produce a bulk flow to the coalescence
scale, and then we continue the flow beyond the coalescence scale using
\eqref{eq:bulkflow}.
The flow of $V$ is simply given by the bulk flow, as the observable terms
in $V$ do not evolve.
Proposition~\ref{prop:nu_c} tells us that the bulk flow remains in the RG domain.
By \eqref{eq:bulkflowdomain} and Lemma~\ref{lemma:bulk_to_obs_bound},  the coordinates with the observable terms also remain in the RG domain at scale $j = \joxm$.
Then, by Theorem~\ref{thm:Phi^K_q0},
we can iterate the non-perturbative RG map for $j \ge \joxm$
and $K$ will also remain in its domain.
\end{proof}

\subsection{Proof of Theorem~\ref{thm:Phi^K_q0} subject to
Propositions~\ref{prop:Phi+0}--\ref{prop:crucial-short-3}}
\label{sec:conditional-proof}

Since $K=0$ in the next proposition, the constant $C_{\pt}$ does not
depend on the constant $\tilde C_{\rg}$ in the definition of the domain of $K$.

\begin{proposition} \label{prop:Phi+0}
Assume \ref{quote:assumPhi} and $j_+ \ge j_\ox$.
There is an $L$-dependent constant $C_{\pt}$ such that
\begin{align}
	\norm{\Phi^K_+(V,0)}_{\Wkappa_{+}}\le C_{\pt}
	\tilde{\vartheta}^3_+	
	 \tilde{g}_{+}^3 \tilde{\scale}_+ (\kaa)
	.
\end{align}
\end{proposition}

The next proposition is a crucial contraction estimate.

\begin{proposition} \label{prop:crucial-short-3}
Assume \ref{quote:assumPhi} and $j_+ \ge j_\ox$.
There are $L$-independent constants $C=C(d)>0$ and $t=t(d)>0$ such that,
for $\norm{\dot{K}}_{\Wkappa} < \infty$,
\begin{align}
	\norm{D_K \Phi_+^K (V,  K; \dot K) }_{\Wkappa_+}
    & \le
    	C  \norm{ \dot{K}}_{\Wkappa}
        L^{- (d-4) \kaa_2 - t}.
    \label{eq:crucial-short}
\end{align}
\end{proposition}

Propositions~\ref{prop:Phi+0} and \ref{prop:crucial-short-3} are proved
in Section~\ref{sec:4}.
Given Propositions~\ref{prop:Phi+0}--\ref{prop:crucial-short-3}, the proof of
Theorem~\ref{thm:Phi^K_q0} is straightforward.

\begin{proof}[Proof of Theorem~\ref{thm:Phi^K_q0}]
By the triangle inequality,
\begin{equation}
\label{eq:DKPhih}
    \norm{\Phi_+^K (V,K)}_{\Wkappa_+}
    \le
    \norm{\Phi_+^K (V,0)}_{\Wkappa_+}
    +
    \norm{ \Phi_+^K (V,K) - \Phi_+^K (V,0) }_{\Wkappa_+}
    .
\end{equation}
As in \eqref{eq:CRGtilde}, we take $\tilde{C}_{\rg} \ge 2 C_{\pt}$,
with $C_{\pt}$ the constant of Proposition~\ref{prop:Phi+0}.
Then,  by Proposition~\ref{prop:Phi+0},
\begin{align}
	\label{eq:Phi_K_bd_at_zero_pf}
	\norm{ \Phi_+^K (V,0)}_{\Wkappa_+}
		& \le
			\frac12 \tilde{C}_{\rg}
			\tilde{\vartheta}^3_+				
			\tilde{g}_+^3 \tilde{\scale}_+ ({\kaa}) .
\end{align}
For the second term in \eqref{eq:DKPhih},  by
writing the difference as the integral of a derivative, and by
Proposition~\ref{prop:crucial-short-3},
\begin{align}
	\norm{\Phi_+^K (V,K) - \Phi_+^K (V,0) }_{\Wkappa_+}
		& \leq \norm{K}_{\Wkappa} \sup_{t\in [0,1]} \norm{D_K \Phi_+^K (V, tK)}_{\Wkappa \rightarrow  \Wkappa_+}
	\nnb
		& \leq \tilde{C}_{\rg}
		\tilde{\vartheta}^3_+		
		\tilde{g}^3 \tilde{\scale} ({\kaa})
		C	L^{-  (d-4) \kaa_2 - t}.
\end{align}
Since $\tilde{g} \le 2 \tilde{g}_+$ and
$\tilde{\scale}_+ ({\kaa}) = \tilde{\scale} ({\kaa})L^{- (d-4) \kaa_2}$
 (by \eqref{eq:rhotil-def}), this gives
\begin{align}
	\norm{\Phi_+^K (V,K) - \Phi_+^K (V,0) }_{\Wkappa_+} & \leq
    (8CL^{-t})\tilde{C}_{\rg}
	\tilde{\vartheta}^3_+
     \tilde{g}_+^3 \tilde{\scale}_+ ({\kaa}).
\end{align}
The proof of \eqref{eq:Phi^K_q0} is completed by taking $L$ large enough that $8CL^{-t} \le \frac 12$.

Finally,  we verify the continuity of $\Phi^K_+$ in $V,K$ and $\ka$.  Continuity in $V,K$ is evident from \eqref{eq:U+def}--\eqref{eq:K+ext_field}.
For continuity in $\ka$, it is shown in \cite[Lemma~8.1]{MPS23} that
 for any polymer activity $F$ the expectation $\E_+ \theta F^{B}$ is
differentiable and therefore
continuous in $\ka$ (where $\E_+$ depends on $\ka$ via the covariance $C_+$).  Thus we see from the definitions \eqref{eq:U+def}--\eqref{eq:K+ext_field} that $\Phi_+$ is also continuous in $\ka$.
\end{proof}

\subsection{Proof of Theorem~\ref{thm:RG-obs} subject to
Theorem~\ref{thm:Phi^K_q0}}
\label{sec:pf_thm:RG-obs}

Fix $N \in \N$ and $\ka \in \II_{\rm crit}$ (recall the definition
\eqref{eq:I_crit_asymp}).
Let $\tilde\ka=\max\{ \ka, 0\}$.
Then Proposition~\ref{prop:stable_manifold} gives an RG flow to scale $\Npp$
with mass $\ka$ and initial condition $(g,\nu_{1,N}(\ka))$.

\begin{proof}[Proof of \eqref{eq:ZNy}]
Since Proposition~\ref{prop:stable_manifold} gives an RG flow to scale $\Npp$,
\eqref{eq:ZNy} follows from \eqref{eq:RGflowN}.
\end{proof}

\begin{proof}[Proof of \eqref{eq:qN-thm}]
The claim \eqref{eq:qN-thm} is that
there is a function $\vv : \Lambda_\infty \to \R$,
with $\vv (\x) = O(g)$ (uniformly in $x$) and $\lim_{|x| \rightarrow \infty} |\x|^{d-2} \vv (\x) = 0$, such that
\begin{align}
    & u_{N,\ox} = C_{\ka,\le N} (\x) + \vv (\x ) +  \frac{1}{|x|^{d-2}} \cO_{N,\x}  .
    \label{eq:uNoxAndCv}
\end{align}
To prove this, we start by recalling that, as indicated below \eqref{eq:ZN},
$\bar{u}_{\Npp} = \sum_{j=1}^{\Npp} \delta \bar{u}_j$.
The formula for $\delta u_{+,\ox}$ in \eqref{eq:u+def} implies that
\begin{align}
\label{eq:uNoxform}
    u_{N,\ox}   &=    \sum_{j= j_\ox}^N \big(  C_{\ka,j}(\x)  + K_{j,\ox}(b_{j,\ox},0 ; \ka) \big) ,
\end{align}
where we now make the dependence of $K_j$ on $\ka$ explicit.
Since $C_{\ka, j} (\x)=0$ for $j<j_\ox$, the first term on the right-hand side
of \eqref{eq:uNoxform} sums
to $C_{\ka, \le N}(\x)$.
To focus on what is important, we abbreviate the notation to
$A_j(\ka)=K_{j,\ox}(b_{j,\ox},0 ; \ka)$.

We define
\begin{align}
	\vv (\x) &=  \sum_{j = j_\ox }^{\infty} A_j(0 )
\label{eq:vvdefi}
\end{align}
The sum \eqref{eq:vvdefi} will be shown below to converge.  Then
\begin{align}
\label{eq:vv-decomp}
     \sum_{j=j_\ox}^N A_j( \ka)
     &=
     \vv(\x) - \sum_{j=N+1}^\infty A_j( 0)
     +
     \sum_{j=j_\ox}^N ( A_j( \ka) - A_j(0)).
\end{align}
It suffices now to prove that the sum \eqref{eq:vvdefi} converges,
that $\vv (\x) = O(g)$ and
 $\lim_{|\x| \rightarrow \infty} |\x|^{d-2} \vv (\x )= 0$,
and that the last two terms in \eqref{eq:vv-decomp} are
$|x|^{-(d-2)}\cO_{N,\x}$.

We begin with the convergence and properties of $\vv$.
By the definition of the semi-norm in \eqref{eq:T_norm},
\begin{align}
	\norm{K_{j, \ox} (b_{j,\ox} ; \ka)}_{T_0 (\ell)} &\leq
			\ell_{\sigma,j}^{-2} \norm{K_{j} (b_{j,\ox})}_{T_0 (\ell)}  ,
	\nnb
		& \le \tilde{g}_j^{-2} L^{(d-6) j} L^{2 \tau (j - j_\ox +1 )}  \times
      \tilde{C}_{\rg}
      \tilde{g}_j^3 L^{-(d-4) {\kaa}_1 \joxm }
      L^{-(d-4) {\kaa}_2  (j - \joxm) }
		\nnb
		&
		\le O_L(1)
    \tilde{g}_j L^{- y_d (j_\ox ) } L^{ -z_d  (j - j_\ox )}
\end{align}
where in the last line we introduced
the exponents $y_d = (d-4) \kaa_1 - (d-6)$ and $z_d = (d-4) \kaa_2 - (2\tau + d-6)$,
and used the fact that $\tilde{C}_{\rg}$ is an $L$-dependent constant.
Since $z_d$ is positive due to \eqref{eq:exbd3},
for $j' \ge j_\ox$ we have
\begin{align}
	\sum_{j= j' }^\infty | A_j (\ka) |
	= \sum_{j= j' }^\infty | K_{j, \ox} (b_{j,\ox} ,  0  ; \ka ) |
		\le O_L (1) \tilde{g}_{j'}  L^{- y_d j_{\ox}  }   L^{ -z_d  (j' - j_\ox )}
    .
    \label{eq:KjoxAbsconv}
\end{align}
By \eqref{eq:KjoxAbsconv},
the sum \eqref{eq:vvdefi}
defining $\vv (\x)$ is absolutely convergent and $\vv(\x) = O(g)$.
Also, it follows from \cite[Lemma~6.5]{MPS23}
and the fact that $j \wedge j_{\ka} \sim j$
for $j \le N$ (since $\ka \in \II_{\rm crit}$),
that $\tilde g_j \asymp j^{-1}$ for $d=4$ and $\tilde g_j \asymp g$ for $d>4$.
For $d=4$,  this implies that $\tilde{g}_{j_\ox } \le O_L (\log |x|)^{-1}$
and $y_d = 2$,
which implies that
\begin{align}
	|\vv (\x)| \le O_L(1)  \tilde{g}_{j_\ox}  L^{- 2 j_{\ox}}
		\le | \x|^{-2} O_L \big( ( \log | \x|  )^{-1} \big)
		.
\end{align}
For $d >4$,
we instead have $\tilde g_{j_{\ox}} = O(g)$, so
\begin{align}
\label{eq:K0x-d5}
	|\vv (\x )| \le
	O_L(1)
    \tilde{g}_{j_\ox} L^{- y_d  j_{\ox} } \le O_L (g) L^{- y_d  j_{\ox} }	
	.
\end{align}
Since $y_d = (d-2) + (d-4) (\ka_1 -2) >  d-2$ due to \eqref{eq:exbd5},
$\lim_{|\x|\rightarrow \infty} |\x|^{d-2} \vv (\x ) = 0$.
Thus we have proved that $\vv$ converges and satisfies all the required properties.

For the term $\sum_{j=N+1}^\infty A_j( 0)$ in \eqref{eq:vv-decomp}, we
see from \eqref{eq:KjoxAbsconv} that its absolute value is bounded
above by an $L$-dependent multiple of
$\tilde{g}_{N}  L^{- y_d j_{\ox}  }   L^{ -z_d  (N - j_{\ox} )}$.
As in the previous paragraph, this is $|x|^{-(d-2)}\cO_{N,x}$, as required.
Finally, it remains to prove that
\begin{align}
	|x|^{d-2} \sum_{j=j_\ox }^{N} \Big[ A_{j}( \ka) -   A_{j}( 0) \Big]
    =\cO_{N,x} .
\end{align}
The  bounds in the previous paragraph
prove that the left-hand side goes to zero as $|x|\to\infty$
even without any cancellation from the subtraction.  On the other hand,
as $N \to \infty$ with bounded $x$, the left-hand side goes to zero
by dominated convergence since $A_j$ is continuous in the mass
and since $\ka \to 0$ as $N \to\infty$ because
$\ka  \in \II_{\crit}$.
This completes the proof of \eqref{eq:qN-thm}.
\end{proof}

\begin{proof}[Proof of \eqref{eq:KNoxbd}]
By definition of the norms,
\begin{align}
	|K_{N,\ox} (\Lambda_N,  y)|
		& \le h_{\sigma,N}^{-2} \norm{ K_{N,\ox} }_{T^{G_N} (h_N)}
	 G_N^{\frac{1}{2}}   (\Lambda_N, y) \nnb
		& \le h_{\sigma,N}^{-2} \tilde{g}_N^{-9/4}  \tilde{\scale}_N ( \kb )^{-1}  \norm{ K_{N} }_{\Wkappa_N}
		 G_N^{\frac{1}{2}}   (\Lambda_N, y)
	.
\end{align}
By Proposition~\ref{prop:stable_manifold}, we have $K_N \in \cK_N$
and hence
$
	\norm{ K_{N} }_{\Wkappa_N}
		\le \tilde{C}_{\rg}   \tilde{g}_N^3 \tilde\scale_{N} ({\kaa})
$,
so
\begin{align}
	h_{\sigma,N}^{-2} \tilde{g}_N^{-9/4}  \tilde{\scale}_N ( {\kb} )^{-1}  \norm{ K_{N} }_{\Wkappa_N}
		& \le h_{\sigma,N}^{-2}  \tilde{C}_{\rg} \tilde{g}_N^{3/4} \scale_{\joxm}^{\kp_1} \scale_{N- \joxm}^{\kp_2}	\nnb
		& \le
		O_L(1) \tilde{g}_N^{1/4}
    L^{ - \left( (d-4) \kp_1 + \frac{d}{2} \right) j_\ox }
    \big( 4 L^{(d-4) \kp_2} \big)^{-(N-j_\ox)}.
\end{align}
For $d=4$, we obtain
\begin{align}
\label{eq:KNoxy}
	|K_{N,\ox} (\Lambda_N,  y)|
		\le
		O_L(1)
    \tilde{g}_N^{1/4}
|\x |^{-2} 4^{-(N-j_{\ox})}
		 G_N^{\frac{1}{2}}   (\Lambda_N, y)
	,
\end{align}
while for $d \ge 5$, \eqref{eq:exbd6} gives
\begin{align}
\label{eq:KNd6}
	|K_{N,\ox} (\Lambda_N,  y)|
		\le
		O_L(1)
\tilde{g}_N^{1/4} |\x |^{- (d- \frac{7}{4}) } \big( 4 L^{\frac{d}{4}  (1-\epsilon_1 (d)) } \big)^{-(N-j_{\ox})}
		 G_N^{\frac{1}{2}}   (\Lambda_N, y)
	.
\end{align}
Since $k_0 \le 1/2$,
by the definition \eqref{eq:Gj}
we have $G_N^{\frac{1}{2}} (\Lambda_N, y)  \le  e^{- \kappa |y / \hh_N|^4}$,
so the bounds \eqref{eq:KNoxy}--\eqref{eq:KNd6}
are both smaller than or equal to the claimed error.
\end{proof}

\section{A non-perturbative RG step}
\label{sec:4}

It remains to prove Propositions~\ref{prop:Phi+0}--\ref{prop:crucial-short-3}.
We prove Proposition~\ref{prop:Phi+0} in
Section~\ref{sec:RGextendednorm}, with the help of estimates from
Appendix~\ref{sec:pfSProps}.
To handle derivatives of the RG map with respect to $K$, we use an
extended norm which incorporates bounds on $K$-derivatives, and we therefore
formulate our proof of Proposition~\ref{prop:Phi+0} in terms of the extended norm.
Proposition~\ref{prop:Phi+0} is a special case of the more
general Proposition~\ref{prop:S0_obs}.
The crucial contraction Proposition~\ref{prop:crucial-short-3}
in proved in Section~\ref{sec:contractions}.

We follow the same approach used to prove the restricted versions of
Propositions~\ref{prop:Phi+0}--\ref{prop:crucial-short-3} for the bulk,
in \cite[Section~9]{MPS23}.  The new feature here is that we also have the
observables, and a new norm that incorporates estimates on the observable components
of $\cN$.  However the norm on the larger space $\cN$ is very efficient, in the sense that
it is possible to prove analogous statements to those used for the bulk
proofs (such as stability and related estimates) with the new norm.
The analogous statements are so analogous that the proofs for the bulk can
then be followed in the new norm with only minor adjustments.

\subsection{Extended norm}
\label{sec:ext_norm}

We extend the semi-norms to include $K$-derivatives, as in \cite[Section~8.1]{MPS23}
(see also \cite[Section~9.2]{BBS-brief}).
It is simpler here since we do not include derivatives with respect to
$V$ or $\ka$ as was done in \cite{MPS23}.

Given $\lambda_K > 0$,
let $\cZ$ be the normed space $\cZ =\R^n \times \cY(\lambda_K)$, where
the norm of $\varphi\in\R^n$ is $|\varphi|/\hh$,
$\cY (\lambda_K)$ is the space $\cK$ with norm
$\norm{K}_{\cW^\kappa}/\lambda_K$, and
$\|(\varphi, y)\|_{\cZ} = \max\{ \frac{|\varphi|}{\mathfrak{h}},  \norm{y}_{\cY} \}$.
We write $y$ in place of $K$ when it is regarded
as an element of $\cY(\lambda_K)$.
For a smooth function $F:\cZ \to \R$,
and for $b \in \cB$ and $\varphi \in \R^n$, we define extended norms as
\begin{align}
\label{eq:ext-norm}
	\norm{F (b, \varphi)}_{T_{\varphi, y} (\kh, \lambda_K)}  & = \sum_{p=0}^{\infty} \frac{1}{p !}  \norm{D_y^p F (b, \varphi) }_{T_{\varphi} (\kh)},
	\\
	\norm{F (b)}_{T_{y}^G (h, \lambda_K)}	&= \sup_{\varphi \in \R^n} G(b, \varphi)^{-1} \norm{F (b, \varphi)}_{T_{\varphi, y} (h, \lambda_K)}, \\
	\norm{F}_{\Wkappa_y (\lambda_K)}	
&= \max_{b \in \cB} \; \Big( \norm{F (b)}_{T_{0,y} (\ell, \lambda_K)} +  \tilde{g}^{9/4} \tilde{\scale} (\tilde{\kb}) \norm{F (b)}_{T_{y}^G (h, \lambda_K)} \Big).
\end{align}
It then follows (a special case
of \cite[(9.2.8)]{BBS-brief}) from the definition  that
\begin{align}
	\norm{D_K^q F (b)}_{T_\varphi (\kh)} & \le \frac{q !}{\lambda_K^q} \norm{F (b)}_{T_{\varphi,y} (\kh,\lambda_K)}
	.
	\label{eq:DKF}
\end{align}
We also include the case $\lambda_K=0$ by interpreting this
as involving only the term $p=0$ in \eqref{eq:ext-norm} (with $K$ held fixed);
in this case the extended norms do not involve the $K$-derivative
and are the same as the norms defined in Section~\ref{sec:3}.

General facts about the extended norm can be found in
\cite[Section~7]{BBS-brief}.
In particular, as in \cite[Lemma~7.1.3]{BBS-brief},
the product property holds:
\begin{align}
	\norm{F_1 F_2}_{T_{\varphi, y} (\kh, \lambda_K)} \le \norm{F_1}_{T_{\varphi, y} (\kh, \lambda_K)}  \norm{F_2}_{T_{\varphi, y} (\kh, \lambda_K)}.
	\label{eq:subx}
\end{align}

\subsection{Extended norm of the RG map: proof of Proposition~\ref{prop:Phi+0}}
\label{sec:RGextendednorm}

For $b \in \cB$, we define
\begin{align}
\label{eq:hatV_hatK_definition}
	\hat{V} = V - Q
    ,
    \qquad Q (b) = \Loc (e^V K) (b),
	\qquad
	\hat{K} = K - (e^{-\hat{V}} - e^{-V}).
\end{align}
For $B \in \cB_+$ and for a set $X \subset \cB(B)$,
we write $|X|$ for the number of blocks in $X$ and define
\begin{align}
	\hat{V} (B \backslash X) = \sum_{b\in \cB (B\backslash X)} \hat{V} (b), \qquad \hat{K}^X = \prod_{b\in X} \hat{K} (b).
\end{align}
With the above definitions,
it follows from the definition of $K_+$ in \eqref{eq:K+ext_field} that
\begin{align}
	& K_+ = \Phi_+^K (V,K) = S_0 + S_1,
	\label{eq:S0S1_decomp}
\end{align}
where
\begin{align}
    S_0 &= e^{- \delta \bar u_{+}  (B)}  \Eplus \Big( e^{- \theta \hat{V}(B)} - e^{-U_{+} (B)} \Big) 	,
    \label{eq:S0_defi}	\\
    S_1 &= e^{- \delta \bar u_{+}  (B)} \sum_{X\subset \cB (B), |X| \geq 1} \Eplus \theta \Big( e^{-\hat{V} (B \backslash X)} \hat{K}^X \Big).
	\label{eq:S1_defi}
\end{align}
The next two propositions provide bounds on $S_0$ and $S_1$.
Since $S_1=0$ when $K=0$,
Proposition~\ref{prop:S0_obs}
immediately gives Proposition~\ref{prop:Phi+0}.
Since $K$ plays no role in the bound \eqref{eq:S0_K0},
the constant $C_{\pt}$ does not depend on $\tilde C_{\rg}$ (it does depend on $L$).

\begin{proposition} \label{prop:S0_obs}
Assume \ref{quote:assumPhi}, $j_+ \ge j_\ox$, and suppose that
$\lambda_K \leq  \tilde{g} \scale$.
Then
\begin{align}
	\norm{S_0}_{\Wkappa_{y,+} (\lambda_K)}
		\le C_L \tilde{\vartheta}^3_+ \tilde{g}_{+}^3 \tilde{\scale}_+ (\kaa)
	.
	\label{eq:S0_bound}
\end{align}
In particular, there is a constant $C_{\pt}$, which is independent of $\tilde C_{\rg}$,
such that
\begin{align}
	\norm{S_0(V,0)}_{\Wkappa_{+}}
		\le C_{\pt} \tilde{\vartheta}^3_+ \tilde{g}_{+}^3 \tilde{\scale}_+ (\kaa)
	.
	\label{eq:S0_K0}
\end{align}
\end{proposition}

\begin{proposition} \label{prop:S1_obs}
Assume \ref{quote:assumPhi}, $j_+ \ge j_\ox$, and suppose that
$\lambda_K \leq \tilde{g}^{9/4} \tilde{\scale} ({\kb} )$.  Then
\begin{align}
	& \norm{S_1}_{\Wkappa_{y,+} (\lambda_K)}
   		\leq C_L
    	(\tilde{\vartheta}^3_+  \tilde{g}_+^3 \tilde{\scale}_+ ({\kaa}) + \lambda_K ) .
\end{align}
\end{proposition}

Propositions~\ref{prop:S0_obs}--\ref{prop:S1_obs} incorporate the
observables into the bulk
statements in \cite[Propositions~9.2, 9.4]{MPS23}.
The introduction of our norm on $\cN$ streamlines the adaptation
of the proofs of \cite[Propositions~9.2, 9.4]{MPS23} to include the
observables.  This requires
revisiting some estimates from \cite{MPS23} but does not introduce new ideas.
We therefore defer the proofs of Propositions~\ref{prop:S0_obs}--\ref{prop:S1_obs}
to Appendix~\ref{sec:pfSProps}.

To prove Proposition~\ref{prop:crucial-short-3}, we will also need control of
the second $K$-derivative of the non-perturbative RG map.
This control is included in the next proposition.
Note that Proposition~\ref{prop:Phi^K_estimate-extended} does not prove
Proposition~\ref{prop:crucial-short-3}: the case $q=1$
of \eqref{eq:Phi^K_derivatives_estimates-mass_derivative} does not
show that $M_1$ decays in $L$.
In addition, in the proof of Proposition~\ref{prop:Phi^K_estimate-extended},
\eqref{eq:S01-pf} (with $\lambda_K =0$) gives the bound
\begin{align}
	\norm{\Phi_+^K }_{\Wkappa_{+}}
		& \le 3 C_L
		\tilde{\vartheta}^3_+
		\tilde{g}_+^{9/4} \tilde{\scale}_+ ({\kb} \big)
		,
\end{align}
which is not as strong as Theorem~\ref{thm:Phi^K_q0} because the constant on
the right-hand side has not been shown to be less that $\tilde C_{\rg}$.
In order to obtain a bound $\tilde C_{\rg} \tilde{\vartheta}^3_+ \tilde{g}_+^{9/4} \tilde{\scale}_+ ({\kb} \big)$
as required by Theorem~\ref{thm:Phi^K_q0},
we will need to apply the crucial contraction
Proposition~\ref{prop:crucial-short-3}.
We have dropped $\tilde{\vartheta}^3_+\le 1$ from \eqref{eq:Phi^K_derivatives_estimates-mass_derivative}
since it is not needed when we apply
\eqref{eq:Phi^K_derivatives_estimates-mass_derivative}.

\begin{proposition}
\label{prop:Phi^K_estimate-extended}
Let $d \ge 4$.
At scale  $j \ge j_\ox - 1$,
assume \ref{quote:assumPhi}.
There exists a positive constant $\tilde{C}_{\rg}$ such that
$\Phi_{+}^K$ is well-defined as a map $\domRG \rightarrow \Wkappa_{+}$
(in particular,
the integral \eqref{eq:K+ext_field} converges), and there exist $M_{q} > 0$ for $q \in \N$ such that
\begin{align}
	\norm{D^q_K \Phi_{+}^K}_{\domRG  \rightarrow \Wkappa_{+}}
		\leq M_{q} 	
		\tilde{g}_{+}^{ - \frac{9}{4}(q-1)}
    	 \tilde{\scale}_+ ({\kb})^{- (q-1)} \qquad (q\geq 1).
	\label{eq:Phi^K_derivatives_estimates-mass_derivative}
\end{align}
\end{proposition}

\begin{proof}
We apply Propositions~\ref{prop:S0_obs}--\ref{prop:S1_obs}
with $\lambda_K =  \tilde{g}^{9/4} \tilde{\scale} (\kb)$
(and use $\kb_i \le \kaa_i$ for $i=1,2$) to obtain
\begin{align}
	\norm{\Phi_+^K }_{\Wkappa_{y,+} (\lambda_K)}
		&\le \|S_0\|_{\Wkappa_{y,+} (\lambda_K)}  + \|S_1 \|_{\Wkappa_{y,+} (\lambda_K)}
			\nnb
		& \le C_L\tilde{\vartheta}^3_+
		\big(\tilde{g}_+^3
        \tilde\scale_+
            (\kaa)
+ \tilde{g}_+^3 \tilde{\scale}_+ (\kaa) + \tilde{g}^{9/4} \tilde{\scale} ({\kb})  )
		\nnb
		& \le 3C_L\tilde{\vartheta}^3_+
		\tilde{g}_+^{9/4} \tilde{\scale}_+ ({\kb} \big)
		.
\label{eq:S01-pf}
\end{align}
Let $q \ge 1$.
By \eqref{eq:DKF},
$\norm{D_K^q \Phi_+^K}_{\Wkappa} \leq \frac{q !}{\lambda_K^q} \norm{\Phi_+^K}_{\Wkappa_{y,+} (\lambda_K)}$,
and the proof is complete.
\end{proof}

\subsection{The crucial contraction: proof of Proposition~\ref{prop:crucial-short-3}}
\label{sec:contractions}

This section is dedicated to the proof of
Proposition~\ref{prop:crucial-short-3}.
We only need to consider scales $j \ge j_\ox - 1$,
and
the statement of Proposition~\ref{prop:crucial-short-3} does not
need the extended norm.
It suffices to prove the following proposition.

\begin{proposition} \label{prop:crucial}
Assume \ref{quote:assumPhi} and let $j_+ \ge j_\ox$.
There is an $L$-independent constant $C=C(d)>0$ such that,
for $\norm{\dot{K}}_{\Wkappa} < \infty$,
\begin{align}
	\norm{D_K \Phi_+^K (V,  K; \dot K) }_{\Wkappa_+}
    & \le
    	C  \norm{ \dot{K}}_{\Wkappa}
    	\max \big\{ L^{ - ( \tau + d-4 )  } ,  L^{-\frac{d}{4} - (d-4)\kb_2 } \big\}
    	.
    \label{eq:crucial}
\end{align}
\end{proposition}

\begin{proof}[Proof of Proposition~\ref{prop:crucial-short-3}]
According to \eqref{eq:exbd2}, there is a positive $t=t(d)$
such that the maximum on the right-hand
side of \eqref{eq:crucial} is bounded above by $L^{-(d-4) \kaa_2 - t}$.
\end{proof}

\subsubsection{Proof of Proposition~\ref{prop:crucial} given Proposition~\ref{prop:cL_contract}}

We now reduce the proof of Proposition~\ref{prop:crucial}
to Proposition~\ref{prop:cL_contract}.
For brevity, we denote the maximum appearing on the
right-hand side of \eqref{eq:crucial} as
\begin{align}
	\eta = \eta(d,L)
		=
			\max \big\{ L^{ - ( \tau + d-4 )  } ,  L^{-\frac{d}{4} - (d-4)\kb_2 } \big\}
			.
\end{align}
To begin, we consider the derivative in \eqref{eq:crucial} at $K=0$.
Since $K_+=S_0+S_1$, we need the derivatives of $S_0$ and $S_1$ at $K=0$.
For $S_0$, we know from Proposition~\ref{prop:S0_obs},
together with \eqref{eq:DKF} applied with $\lambda_K = \tilde{g} \scale$, that
\begin{align}
\label{eq:DKS0}
	\norm{D_K S_0 (V, 0)}_{\Wkappa \rightarrow \Wkappa_+}
		\le \frac{1}{\lambda_K} \norm{S_0 (V, 0)}_{\Wkappa_+ (\lambda_K)}
		\le C_L  \tilde{g}_+^2 \scale_+^2
		\le \eta
	,
\end{align}
where the final bound holds for $\tilde{g}$ sufficiently small.
Our main effort is to bound $D_K  S_1$.
We extract the linear (in $K$) part of $S_1$ by writing
\begin{align}
	S_1 =   e^{\delta \bar u_{\pt} (B)} \Eplus \theta \cL K + \cE ,
	\label{eq:cE_defi}
\end{align}
with
\begin{align}
	\cL K (B) &= \sum_{b\in \blocks (B)} e^{-V (B)} (1- \Loc) (e^{V(b)} K(b))
	\label{eq:LK0}
	,
\end{align}
and with $\cE$ implicitly defined by \eqref{eq:cE_defi}.
Since the linear part (in $K$) has been extracted from $\cE$, the following lemma is not at all surprising.

\begin{lemma} \label{lemma:crucial_contraction_pre_obs}
Assume \ref{quote:assumPhi}.
The Fr{\'e}chet derivative of $\cE (V, \cdot):\Wkappa\to\Wkappa_+$ vanishes at $K =0$.
\end{lemma}

\begin{proof}
The proof with observables follows the same steps as the proof of
the $r=0$ case of
\cite[Lemma~9.8]{MPS23},
as we indicate in Appendix~\ref{sec:pfSProps}.
\end{proof}

By \eqref{eq:cE_defi} and Lemma~\ref{lemma:crucial_contraction_pre_obs},
\begin{align}
\label{eq:DKS1bd}
	D_K S_1 (V,0) (\dot{K})
    = e^{\delta \bar{u}_{\pt} (B)} \Eplus \theta \cL \dot{K}
	.
\end{align}
We will prove the following proposition in Section~\ref{sec:decomp_cL}.

\begin{proposition}
\label{prop:cL_contract}
Assume \ref{quote:assumPhi} and $j_+ \ge j_\ox$.
There is a $C>0$ (independent of $L$) such that
\begin{align}
	\norm{\Eplus \theta \cL \dot{K} (B) }_{\Wkappa_+}
    \le
	C \eta \norm{ \dot{K} }_{\Wkappa}.
\end{align}
\end{proposition}

\begin{corollary}	\label{cor:DKS1}
Assume \ref{quote:assumPhi} and $j_+ \ge j_\ox$.
Then
\begin{align}
	\norm{ D_K S_1 (V,0) }_{\Wkappa \rightarrow \Wkappa_+}
		\le
	(C+2)
    \eta .
\end{align}
\end{corollary}

\begin{proof}
The exponential prefactor in \eqref{eq:DKS1bd} can be bounded using Lemma~\ref{lemma:stability_estimate_obs},
and $\Eplus \theta \cL \dot{K}$ is bounded using
Proposition~\ref{prop:cL_contract}.
\end{proof}

\begin{proof}[Proof of Proposition~\ref{prop:crucial}]
We combine the bounds on $D_K S_0 (V,0)$ and $D_K S_1 (V,0)$ from \eqref{eq:DKS0}
and Corollary~\ref{cor:DKS1},
to obtain
\begin{align}
	\norm{D_K \Phi_+^K (V, 0)}_{\Wkappa \rightarrow \Wkappa_+}
		\le
	(C+1)
    \eta
		.
\end{align}
It suffices to bound by $\eta$ the norm of the difference between
the derivative at $K$ and at $K=0$.  For this, we use
\begin{align}
	\label{eq:MVT_pf_contrac}
	&\norm{D_K ( \Phi_+^K (V, K) - \Phi^K_+ (V,0) ) }_{\Wkappa \rightarrow \Wkappa_+}
		\le
		\sup_{t\in [0,1]}
		\norm{D^2_K \Phi^K_+ (V,  t K) }_{ (\Wkappa)^2 \rightarrow \Wkappa_+} \norm{K}_{\Wkappa}.
\end{align}
To bound $D_K^2 \Phi^K_+ = D_K^2 S_0 + D_K^2 S_1$, we use Propositions~\ref{prop:S0_obs}--\ref{prop:S1_obs} with $\lambda_K = \tilde{g}^{9/4} \min\{ \tilde{\scale} ({\kb}) , \, \scale \}$ to see that
\begin{align}
	\norm{\Phi_+^K}_{\Wkappa_+ (\lambda_K)}
		\le O_L (1) ( \tilde{g}_+^3 \tilde{\scale}_+ (\kaa) + \lambda_K )
		\le O_L (1) \lambda_K,
\end{align}
where the final inequality follows for small $\tilde{g}$ because
$\kaa_i \ge \max\{ \kb_i ,  1 \}$ for each of $i=1,2$,
by \eqref{eq:exbd1}.
Thus by \eqref{eq:DKF},
\begin{align}
	\norm{D^2_K \Phi^K_+ (V,  t K) }_{ (\Wkappa)^2 \rightarrow \Wkappa_+} \norm{K}_{\Wkappa}
		& \le \frac{2}{\lambda_K^2} \norm{\Phi^K_+ (V,  t K) }_{ \Wkappa_+ (\lambda_K)} \norm{K}_{\Wkappa} \nnb
		& 		
		\le C_L \tilde{g}_+^{-\frac{9}{4}} \min\{ \tilde{\scale} ({\kb}) , \, \scale \}^{-1} \tilde{g}^3 \tilde{\scale} ( {\kaa})
		\le \eta
    ,
\end{align}
where the final inequality again follows for $\tilde{g}$ small depending on $L$,
since $\kaa_i \ge \max\{ \kb_i ,  1 \}$.
This concludes the proof.
\end{proof}

\subsubsection{Proof of Proposition~\ref{prop:cL_contract}}
\label{sec:decomp_cL}

It remains to prove Proposition~\ref{prop:cL_contract}.
To do so, we follow the procedure used to prove the bulk contraction
in \cite[Section~9.4.2]{MPS23}.
This requires the extension of the bulk result \cite[(10.5.1)]{BBS-brief}
to the full space $\cN$ including observables.
Unlike the situation for the bounds on $S_0$ and $S_1$ in
Propositions~\ref{prop:S0_obs}--\ref{prop:S1_obs},
the observables now play an important role in the proof.  We
prove Proposition~\ref{prop:cL_contract} here,
with reference to Appendix~\ref{sec:pfSProps} for some details.

For $\kh >0$, we define a linear polynomial in $|\varphi|$  by
\begin{equation}
\label{eq:Phdef}
    P_{\kh}(\varphi) = 1 + \frac{|\varphi|}{\kh}.
\end{equation}
We need the following lemma
from \cite{MPS23}, which
enables the regulator to be propagated in estimates from one scale to the next.

\begin{lemma} \label{lemma:EG}
{\!\! \cite[Lemma~~8.5]{MPS23}.}
For any $t \geq 0$, $p \ge 0$, $\varphi \in \R^n$, $b \in \cB$, and $B \in \cB_+$,
\begin{align}
	& \Eplus \Bigg[ G^t (b, \varphi + \zeta_b) P_{h_+}^p (\varphi + \zeta_b)  \prod_{b' \in \cB ,  \, b' \neq b} G^{2t} (b, \varphi + \zeta_{b'}) \Bigg] \leq O_{p,t} (1) G_+^t (B,\varphi)
	.
\end{align}
\end{lemma}

The $T_0(\ell)$-seminorm is one of two terms in the definition of the
$\Wkappa$-norm in \eqref{eq:Wkappa-def}.  The next lemma provides a way to bound
the $T_\varphi(\ell)$-semi-norm in terms of the $\Wkappa$-norm.
It is an extension of \cite[(10.4.5)]{BBS-brief} to include the observables.

\begin{lemma} \label{lemma:FFF-1}
For $F \in \cN$ and $\varphi \in \R^n$,
\begin{align}
\label{eq:FFF-1}
    \|  {F} \|_{T_{\varphi}(\ell)}
    \le P^{10}_{\ell_\bulk} ( \varphi) \norm{ {F}}_{\Wkappa}.
\end{align}
\end{lemma}

\begin{proof}
By definition of the norms, it suffices to prove \eqref{eq:FFF-1} for $F\in \cN_\bulk$,
which we now assume.
According to the definitions in \eqref{eq:hbulk-def},
\begin{equation}
\label{h-ell-ratio-W}
    \Big( \frac{\ell_\bulk}{h_\bulk} \Big)^{10}
    =
    \frac{\ell_0^{10}}{k_0^{10}} (\tilde g \scale)^{5/2}
    =
    \frac{\ell_0^{10}}{k_0^{10}} \tilde g^{9/4}\scale^{\kb_1} (\tilde g^{1/4}\scale^{\frac 52 - \kb_1})
    \le \frac{\ell_0^{10}}{k_0^{10}} \tilde g^{9/4} \tilde\scale (\kb)
    (\tilde g^{1/4}\scale^{\frac 52 - \kb_1}),
\end{equation}
where the inequality follows from $\kb_2 \le \kb_1$
(recall \eqref{eq:kb12}) when $d >4$ and from $\scale = \tilde{\scale} = 1$ when $d=4$.
Since $\kb_1 < \frac 52$ (again by \eqref{eq:kb12}),
the last factor in parenthesis is as small as desired
for small $\tilde g$,
and in particular the right-hand side of \eqref{h-ell-ratio-W} is bounded above
simply by $\tilde g^{9/4} \tilde\scale(\kb)$.
Then \eqref{eq:FFF-1} follows exactly as in the proof of \cite[Lemma~8.2.3]{BBS-brief}.
\end{proof}

A key step is provided by the
next lemma, which is related to \cite[Proposition~4.9]{BS-rg-IE}.
The lemma shows how, under a change of scale, the removal of relevant and
marginal terms (in the RG sense) gives rise to an inverse power of $L$, which
itself gives rise to the contraction of the RG map.
The proof proceeds by applying the bulk result componentwise in $\cN$.
For its statement, we recall that the subspace $\cS$ is defined
in Definition~\ref{def:SK}, and we define
\begin{align}
    \gamma(\ell,b) & = L^{-3(d-2)} + L^{-( \tau+d-4)}\1_{\{\o,\x\}\cap b \neq \emptyset},
    \\
    \gamma(h,b) & = L^{-3d/2} + L^{-d/4} \1_{\{\o,\x\}\cap b \neq \emptyset}.
\end{align}
In preparation for the proof, we also
observe that the definitions \eqref{eq:hbulk-def}, \eqref{eq:h-sig-alt},
\eqref{eq:ell-sig-1},
and the bound \eqref{eq:tgc} imply that
\begin{align}
\label{eq:h-ratios}
    \frac{\kh_{\varnothing,+}}{\kh_{\varnothing}}
    \le
    \begin{cases}
    L^{-\frac{d-2}{2}} & (\kh=\ell)
    \\
    2
    L^{-d/4} & (\kh=h),
    \end{cases}
    \qquad
    \frac{\kh_{\sigma,+}}{\kh_{\sigma}}
    \le
    \begin{cases}
    L^{-\tau-\frac{d-6}{2}} & (\kh=\ell)
    \\
    2 & (\kh=h),
    \end{cases}
\end{align}
\begin{align}
\label{eq:h-ratios-2}
    \frac{\kh_{\varnothing,+}}{\kh_{\varnothing}}
    \frac{\kh_{\sigma,+}}{\kh_{\sigma}}
    \le
    \begin{cases}
    L^{-(\tau+d-4)} & (\kh=\ell)
    \\
    4L^{-d/4} & (\kh=h).
    \end{cases}
\end{align}

\begin{lemma}
\label{lem:1-Loc}
There is an $L$-independent constant $c$ such that
for $E, F \in \cS$,
for $b\in\cB$, and for $\kh\in \{\ell,h\}$,
\begin{align}
    \| E (1-\Loc)F (b)\|_{T_\varphi(\kh_+)}
    \le
    c
    \gamma(\kh,b) P_{\kh_{\bulk},+}^6(\varphi)
    \norm{E(b)}_{T_{\varphi} (\kh)}
    \sup_{0\le t\le 1} \| F (b) \|_{T_{t\varphi}(\kh)}
    .
\end{align}
\end{lemma}
\begin{proof}
By definition of the norm in \eqref{eq:T_norm},
\begin{align}
\label{eq:cc-pf1}
    \| E (1-\Loc)F(b)\|_{T_\varphi(\kh_+)}
    &  =
    \sum_{* \in \{ \bulk,\o,\x,\ox \}} \|\sigma_{* }\|_{T_\varphi(\kh_+)}
    \| ( E (1-\Loc )F )_* (b) \|_{T_\varphi(\kh_+)}.
\end{align}
In the following, we sometimes omit the dependence of $E$ and $F$ on $b$
to reduce the notation.

Let $p_\varnothing = 6$ and $p_\o=p_\x = p_\ox =1$.
By \cite[Lemma~7.5.3]{BBS-brief} and by the definition of $\Loc$ in \eqref{eq:Loc_defi},
\begin{align}
    \|(1-\Loc_*)F_{*}(b)\|_{T_\varphi(\kh_+)}
    \le
    \1_{* \in b} \, 2 \Big( \frac{\kh_{\varnothing,+}}{\kh_{\varnothing}} \Big)^{p_*}
    P_{\kh_{\bulk},+}^{p_*}(\varphi)
    \sup_{0\le t\le 1} \|F_* \|_{T_{t\varphi}(\kh)}
    .
\end{align}
Here $p_\varnothing = 6$ occurs rather than $4+1=5$ because
we can replace $\Tay_4$ by $\Tay_5$ in \eqref{eq:Loc_defi}
due to our symmetry assumption on $F \in \cS$.
In the above right-hand side,
$\o \in b$ and $\x \in b$ are interpreted as usual,
$\bulk \in b$ is always true, and $\ox \in b$ means $\o \in b$ and $\x \in b$.

By expressing each $(E (1-\Loc )F)_{*}$ as a linear combination of
$E_{a_1} (1-\Loc )F_{a_2}$ such that $\sigma_* = \sigma_{a_1}  \sigma_{a_2}$,
and using the monotonicity
$\| E_{a_1} \|_{T_{\varphi}(\kh_+)} \le \| E_{a_1} \|_{T_{\varphi}(\kh)}$
(since $\kh_+ \le \kh$),
we obtain
\begin{align}
\label{eq:cc-pf2}
	& \| (E(1-\Loc)F )_{*}(b)\|_{T_\varphi(\kh_+)}
	\nnb & \qquad
	\leq
    \1_{*  \in b}
    2 \Big( \frac{\kh_{\varnothing,+}}{\kh_{\varnothing}} \Big)^{p_*}
    	P_{\kh_{\bulk},+}^6 (\varphi)
    \sum_{a_1, a_2: \sigma_{a_1} \sigma_{a_2} = \sigma_*}
    \|  E_{a_1} \|_{T_{\varphi}(\kh)}
    \sup_{0\le t\le 1} \|   F_{a_2} \|_{T_{t\varphi}(\kh)}
    	.
\end{align}
We use \eqref{eq:cc-pf2} in conjunction with \eqref{eq:cc-pf1}, and we also convert
the scale of the norm of $\sigma_*$ by using
\begin{equation}
    \|\sigma_{*}\|_{T_\varphi(\kh_+)}
    =
    \Big(\frac{\kh_{\sigma,+}}{\kh_\sigma} \Big)^{|*|} \|\sigma_{*}\|_{T_\varphi(\kh)}
    \le  2 \Big(\frac{\kh_{\sigma,+}}{\kh_\sigma}
    \Big)^{|*|\wedge 1 }
    \|\sigma_{*}\|_{T_\varphi(\kh)},
\end{equation}
where $|\bulk|=0$, $|\o|=|\x|=1$, $|\ox|=2$, and we took the worst alternative
in the ratio \eqref{eq:h-ratios} when $|*|=2$.
Then we use $\norm{\sigma_{a_1} E_{a_1}}_{T_{\varphi} (\kh)}
\le \norm{E}_{T_{\varphi} (\kh)}$, and similarly for $F$.
The result is
\begin{align}
    &\| E (1-\Loc)F(b)\|_{T_\varphi(\kh_+)}
    \nnb &\hspace{10mm}  \le
    c
    P_{\kh_{\bulk},+}^{6}(\varphi)
	\norm{E}_{T_{\varphi} (\kh)} \sup_{0\le t\le 1} \| F \|_{T_{t\varphi}(\kh)}
    \sum_{*\in \{\bulk,\o,\x,\ox\}} \1_{*\in b}
    \Big(   \frac{\kh_{\bulk,+}}{\kh_{\bulk}} \Big)^{p_*}
    \Big(   \frac{\kh_{\sigma,+}}{\kh_\sigma}
    \Big)^{|*|\wedge 1}
    .
\end{align}
Finally, we use \eqref{eq:h-ratios}--\eqref{eq:h-ratios-2} to compute the terms in the above sum.
For $\kh=\ell$, we find
$L^{-3(d-2)}$ for $*=\bulk$, and $L^{-(\tau+d-4)}$ for $*=\o,\x,\ox$.
This produces
$\gamma(\ell,b)$.  For $\kh=h$, we find respectively
$L^{-3d/2}$ and $L^{-d/4}$ which produces $\gamma(h,b)$.
This completes the proof.
\end{proof}

For the proof of Proposition~\ref{prop:cL_contract},
as in \cite[Section~9.4.2]{MPS23} we decompose the operator $\cL$ from \eqref{eq:LK0} as
\begin{align}
	\label{eq:L_decomp}
	\cL \dot{K}  = \sum_{k =0}^{2} \frac{1}{k !} \cL_1 (V^k \dot{K} ) + \cL_2 \dot{K}
	,	
\end{align}
with
\begin{align}
	\cL_1 (V^q \dot{K}) (B)
    &=
    \sum_{b\in \blocks (B)} \Big( e^{-V  (B)} (1- \Loc) (V^q \dot{K}) (b) \Big),
	\\
	\cL_2  \dot{K} (B) &=
	\sum_{b\in \blocks (B)} \Big( e^{-V  (B)} (1- \Loc)
    \Big[\big( e^{V   }  -1 - V  - \frac 12 V^2   \big) \dot{K}  \Big](b) \Big).
\end{align}
Let $F_V  = (e^{V } - 1- V  - V^2 / 2)$.
Since the Taylor expansion of $F_V$ in $\varphi$ vanishes up to and including degree five in the bulk component and degree zero in the observable components,
$1 - \Loc$ acts as the identity
on $F_V \dot{K}$ and $e^{-V } F_V $, and therefore
\begin{align}
	\cL_2 \dot{K} (B)
	&= \sum_{b\in\blocks(B)} e^{-V  (B )} F_V  (b)  \dot{K}(b)
    = \sum_{b\in\blocks(B)} e^{-V  (B \backslash b)} \dot{K}(b) (1- \Loc) \big( e^{-V  (b)} F_V  (b)  \big)
	.
\end{align}
With the definitions
\begin{align}
\begin{array}{ll}
		E_1 (b) = e^{-V (b) } ,  &\quad E_2 (b) = \dot{K} (b), \\
		F_1 (b) = \sum_{k=0}^2 \frac{1}{k!} ( V^k \dot{K} ) (b)
		,&\quad
		F_2(b) =  e^{-V}(e^{V} - 1 - V - V^2 / 2 ) (b),
\end{array}
\end{align}
the above gives
\begin{align} \label{eq:Tbdef}
  \cL \dot{K}
  &
  =
  \sum_{b\in \cB(B)} e^{-V(B \backslash b)} \sum_{i=1}^{2} E_i (b) (1- \Loc) F_i (b).
\end{align}

\begin{proof}[Proof of Proposition~\ref{prop:cL_contract}]
We use $C$ for a generic $L$-independent constant whose value may change from
line to line.
With $\eta = \max\{L^{-(\tau+d-4)},L^{-\frac d4 -(d-2)\kb_2}\}$, our goal is to prove that
\begin{align}
	\norm{\Eplus \theta \cL \dot{K} (B) }_{\Wkappa_+}
    \le
	C \eta \norm{ \dot{K} }_{\Wkappa}.
\end{align}
For this, it suffices to prove that
\begin{align}
\label{eq:cc-goal-ell}
  \|\E_+\theta \cL \dot{K} \|_{T_{0}(\ell_+)}
  & \le C L^{-( \tau+d-4)}   \norm{ \dot{K} (b)}_{\Wkappa},
\\
\label{eq:cc-goal-h}
  \|\E_+\theta \cL \dot{K} \|_{T_{\varphi}(h_+)}
  & \le C L^{-d/4}   \norm{ \dot{K} (b)}_{T^G(h)}  G_+^{\frac{1}{2}} (B,\varphi),
\end{align}
and then $\kb_2$
in $\eta$ arises from adjusting the factor $\tilde\scale(\kb)$ in the norm
for the change in scale.
By the triangle inequality,
by use of \cite[Proposition~7.3.1]{BBS-brief} to bring the norm inside
the expectation, and by the product property of the norm,
\begin{align}
\label{eq:ETell}
	\|\E_+\theta \cL \dot{K} \|_{T_{\varphi}(\kh_+)}
	&
	\leq \sum_{b\in \cB(B)}
	\E_+  \bigg[ \Big( \prod_{b'\neq b}\| e^{-V(b')}\|_{T_{\varphi+\zeta_{b'}}(\kh_+ )} \Big)
	\sum_{i=1}^2
	\| E_i (b) (1-\Loc) F_i (b)\|_{T_{\varphi+\zeta_b}(\kh_+)}  \bigg]
	.
\end{align}

Let $\hat \varphi = \varphi+\zeta_b$.
By Lemma~\ref{lem:1-Loc},
\begin{align}
	\label{eq:EFib}
    \|E_i (1-\Loc) F_i(b)\|_{T_{\hat \varphi}(\kh_+)}
    \le
    c
    \gamma(\kh,b) P_{\kh_{\bulk,+}}^6(\hat \varphi)
	\| E_i \|_{T_{\hat\varphi}(\kh)}	
	\sup_{0\le t\le 1} \| F_i \|_{T_{t\hat \varphi}(\kh)}
    .
\end{align}
For $i=1$,
it follows immediately from \eqref{eq:4.6-1} that
\begin{align}
\label{eq:app-need-1}
	\|E_1(b)\|_{T_{\hat \varphi}(\kh)}
    =
    \|e^{-V(b)}\|_{T_{\hat \varphi}(\kh)}
		\leq C G^2 (b,  \hat{\varphi})
		\qquad (\kh \in \{ h, \ell \})
		.
\end{align}
Also, with Lemma~\ref{lemma:FFF-1} (for small $g$) when $\kh=\ell$, and with
Lemma~\ref{lemma:V_bounds} for both cases,
\begin{align}
\label{eq:app-need-5}
    \|F_1(b)\|_{T_{t\hat \varphi}(\kh)}
    &\le
    \big( 1 + \| V(b)\|_{T_{t\hat \varphi}(\kh)} \big)^2
    \| \dot{K} (b)\|_{T_{t\hat \varphi}(\kh)}
    \nnb & \le
    \begin{cases}
        P^{18}_{\ell_\bulk} ( \hat \varphi) \norm{ \dot{K}(b)}_{\Wkappa} & (\kh=\ell)
        \\
        C P^{8}_{h_\bulk} ( t \hat \varphi) G^{\frac{1}{2}}  (b,t\hat \varphi) \|\dot{K}(b)\|_{T^G(h)} & (\kh=h).
    \end{cases}
\end{align}
The regulator controls powers of $P_{h_\bulk}$ for the case $\kh=h$,
and $P_{\ell_\bulk} \le P_{\ell_{\bulk,+}}$, so
\begin{equation}
    P_{\kh_{\bulk,+}}^6(\hat \varphi)\| E_1 \|_{T_{\hat\varphi}(\kh)}	
	\sup_{0\le t\le 1} \| F_1 \|_{T_{t\varphi}(\kh)}
    \le  C\times
    \begin{cases}
        P^{24}_{\ell_{\bulk,+}} ( \hat \varphi) \norm{ \dot{K}(b)}_{\Wkappa} & (\kh=\ell)
        \\
        P_{\kh_{\bulk,+}}^6(\hat \varphi)
        G  (b,  \hat{\varphi})  \|\dot{K}(b)\|_{T^G(h)} & (\kh=h).
    \end{cases}
\end{equation}
For $i=2$, it follows immediately from \eqref{eq:4.6-3} that
\begin{align}
\label{eq:app-need-2}
    \|F_2(b)\|_{T_{t\hat \varphi}(\kh)}
    \le O(1),
\end{align}
while by Lemma~\ref{lemma:FFF-1} for $\kh = \ell$ and by the definition of the norm for $\kh = h$,
\begin{align}
    \| E_2(b)\|_{T_{\hat \varphi}(\kh)}
    =
    \| \dot{K} (b)\|_{T_{\hat \varphi}(\kh)}
    \le
    \begin{cases}
        P^{10}_{\ell_\bulk} ( \hat \varphi) \norm{ \dot{K} (b)}_{\Wkappa} & (\kh=\ell)
        \\
        G^{\frac{1}{2}} (b,\hat \varphi) \| \dot{K} (b)\|_{T^G(h)} & (\kh=h).
    \end{cases}
\end{align}
Thus we have shown that
\begin{align}
\label{eq:ELocF}
    \sum_{i=1}^2
	\| E_i (b) (1-\Loc) F_i (b)\|_{T_{\hat\varphi}(\kh_+)}
    \le
    C\times
    \begin{cases}
        P^{24}_{\ell_{\bulk,+}} ( \hat \varphi) \norm{ \dot{K}(b)}_{\Wkappa} & (\kh=\ell)
        \\
        P^6_{h_{\bulk,+}} (\hat{\varphi}) G^{\frac{1}{2}}  (b,  \hat{\varphi}) \|\dot{K}(b)\|_{T^G(h)}  & (\kh=h).
    \end{cases}
\end{align}
For the product on the right-hand side of \eqref{eq:ETell}, we apply \eqref{eq:eV4.6-old} to see that
\begin{align}
\label{eq:eVreg-old}
	\norm{ e^{-V (b')} }_{T_{\hat\varphi} (\mathfrak{h}_+)}
    &\le	
    	C^{L^{-d}}G^2(b',\varphi)
    	\times \begin{cases}
			1 & (b' \neq b_\o, \,b_\x) \\
			C  & (b' = b_\o, \,b_\x)
			.
    	\end{cases}
\end{align}
When we bound products over $e^{-V(b')}$,
the presence of $C$ in the second case of \eqref{eq:eVreg-old}
does not create any $L$-dependent constant since it occurs for at most two blocks.

We can now prove our two goals \eqref{eq:cc-goal-ell}--\eqref{eq:cc-goal-h}
by combining \eqref{eq:ETell}, \eqref{eq:ELocF}, and \eqref{eq:eVreg-old}.
For \eqref{eq:cc-goal-ell}, we set $\kh=\ell$ and $\varphi=0$.
Since
$\Eplus [ P_{\ell_{\bulk,+}}^q(\zeta_b )] \le C_q$ for any $q\ge 0$
(see \cite[Lemma~10.3.1]{BBS-brief}),
we find that
\begin{align}
  \|\E_+\theta \cL \dot{K} \|_{T_{0}(\ell_+)}
  &
  \leq C \norm{ \dot{K} (b)}_{\Wkappa} \sum_{b\in \cB(B)} \gamma(\ell,b)
  \E_+ ( P_{\ell_{\bulk,+}}^{24}(\zeta_b) )
  \nnb
  &
  \leq C \norm{ \dot{K} (b)}_{\Wkappa} \sum_{b\in \cB(B)} \gamma(\ell,b)
  \nnb
  &\le C \norm{ \dot{K} (b)}_{\Wkappa}\max\{L^{d-3d(d-2)} , L^{-( \tau+d-4)} \}
  \nnb
  & = C L^{-( \tau+d-4)}   \norm{ \dot{K} (b)}_{\Wkappa},
\label{eq:cc-ellbd}
\end{align}
where the equality is due to \eqref{eq:exbd4}.

For \eqref{eq:cc-goal-h},
by using Lemma~\ref{lemma:EG} to bound the expectation, we have
\begin{align}
  \|\E_+\theta \cL \dot{K} \|_{T_{\varphi}(h_+)}
  &
  \leq C \norm{ \dot{K} (b)}_{T^G(h)}  \sum_{b\in \cB(B)} \gamma(h,b)
  \E_+\Big( P_{h_{\bulk,+}}^{6}(\varphi+\zeta_b) G^{\frac{1}{2}} (b,\varphi+\zeta_b)
  \prod_{b'\neq b}   G^2(b',\varphi+ \zeta_{b'})
  \Big)
  \nnb & \le
  C \norm{ \dot{K} (b)}_{T^G(h)} G^{\frac{1}{2}}_+(B,\varphi) \sum_{b\in \cB(B)} \gamma(h,b)
  \nnb &
  = C \norm{ \dot{K} (b)}_{T^G(h)} G^{\frac{1}{2}}_+(B,\varphi) \max\{L^{d-3d/2} , L^{-d/4} \}
  \nnb &
  = C L^{-d/4}   \norm{ \dot{K} (b)}_{T^G(h)} G^{\frac{1}{2}}_+ (B,\varphi).
\end{align}
We have proved \eqref{eq:cc-goal-ell}--\eqref{eq:cc-goal-h}, so the proof is complete.
\end{proof}

\appendix

\section{Stability and related estimates, with observables}
\label{sec:pfSProps}

Propositions~\ref{prop:S0_obs}--\ref{prop:S1_obs} and
Lemma~\ref{lemma:crucial_contraction_pre_obs}
incorporate the
observables into the bulk
statements in \cite[Propositions~9.2, 9.4]{MPS23} and \cite[Lemma~9.8]{MPS23}.
In this appendix, we indicate the relatively minor adaptations of
the analysis of \cite{MPS23} that their proofs entail.
The appendix is not self-contained and is intended to be read in conjunction
with \cite{MPS23}.

Recall from \eqref{eq:Phdef} that
    $P_{\kh} (\varphi) = 1 + |\varphi| / \kh$.
The following general lemma is useful for bounding the norm of polynomials in
the field.

\begin{lemma} \label{lemma:poly_bound}
For $p >0$ and $i_1, \ldots, i_p \in \{ 1, \ldots, n \}$, we have $\norm{\varphi^{(i_1)} \cdots \varphi^{(i_p)}}_{T_0 (\kh)} \leq \kh^p$.
If $F \in \cN_\bulk$ is a polynomial of degree $k$ in $\varphi$, then
\begin{align}
\label{eq:polybd}
	\norm{F(b, \varphi)}_{T_{\varphi,y} (\kh,\lambda_K)}
    \leq \norm{F(b, \varphi)}_{T_{0,y} (\kh,\lambda_K)}
    P_{\kh}^k (\varphi).
\end{align}
\end{lemma}
\begin{proof}
The bound $\norm{\varphi^{(i_1)} \cdots \varphi^{(i_p)}}_{T_0 (\kh)} \leq \kh^p$ follows by definition, and \eqref{eq:polybd} is \cite[Exercise~7.5.2]{BBS-brief}.
\end{proof}

The following six lemmas provide the ingredients used in the proofs of
Propositions~\ref{prop:S0_obs}--\ref{prop:S1_obs} and
Lemma~\ref{lemma:crucial_contraction_pre_obs}.

\begin{lemma} \label{lemma:V_bounds}
Let $V \in \cD$, $j_+ \ge j_\ox$, $b \in \cB$, and $\z \in \{ \o,\x \}$.
There is an $L$-independent constant $C$ and an $L$-dependent constant $C_L$ such that,
\begin{align}
	\norm{\sigma_\z V_\z (b)}_{T_\varphi (\kh)} & \le
	\begin{cases}	
		C_L  \tilde{g} \scale P_{\ell_\bulk} (\varphi) &   (\mathfrak{h} = \ell) \\
		C P_{h_\bulk} (\varphi) &   (\mathfrak{h} = h,  \; j = j_\ox - 1) \\
		C L^{-\frac{d}{4}} P_{h_\bulk} (\varphi) &   (\mathfrak{h} = h,  \; j \ge j_\ox)
    ,
    \\
	\end{cases}
	\label{eq:Vz_bound}	 \\
	\norm{V (b)}_{T_{\varphi} (\mathfrak{h})} & \le \begin{array}{ll}
	\begin{cases}	
	C_L \tilde{g} \scale P^4_{\ell_\bulk} (\varphi) &   (\mathfrak{h} = \ell ) \\
	C P^4_{h_\bulk} (\varphi) &   (\mathfrak{h} = h).
	\end{cases}
	\end{array}
	\label{eq:V_bound0}
\end{align}
\end{lemma}

\begin{proof}
By Lemma~\ref{lemma:poly_bound}, it suffices to consider the case $\varphi=0$,
and for $\varphi=0$ \eqref{eq:Vz_bound} follows from
\eqref{eq:Vznorm-ell}--\eqref{eq:Vznorm-h}.
Then \eqref{eq:V_bound0} follows from the bounds on $V_\bulk$ in
\cite[(8.62)]{MPS23}.
\end{proof}

For clarity when working with the extended norm, given a block $b$
we write $K^\#$ for the function $K^\#:(K, \varphi) \mapsto K(b, \varphi)$.
This occurs, in particular, in the next proof.
By the definition of the extended norm,
\begin{equation}
\label{eq:Ksharp}
    \norm{K^{\#}}_{T_{0,y} (\kh, \lambda_K)} = \norm{K}_{T_0 (\kh)} + \lambda_K.
\end{equation}
As in \eqref{eq:hatV_hatK_definition}, $Q$ is the polynomial defined by
$Q (b) = \Loc (e^V K) (b)$.

\begin{lemma}
\label{lemma:Q_bound}
For $(V,K) \in \domRG$ and $j_+ \ge j_\ox$,
\begin{align}
\label{eq:Qphi}
	\norm{Q(b)}_{T_{\varphi,y} (\kh, \lambda_K)}
   &  \le e^{\norm{V(b)}_{T_0 (\kh)} } ( \norm{K}_{T_0 (\kh)} + \lambda_K ) P_{\kh_\bulk}^4 (\varphi)	, \\
\label{eq:Qphi_obs}
	\norm{(Q - Q_\bulk) (b)}_{T_{\varphi,y} (\kh, \lambda_K)}
    &  \le e^{\norm{V(b)}_{T_0 (\kh)} }  ( \norm{K}_{T_0 (\kh)} + \lambda_K ) P_{\kh_\bulk} (\varphi)	.
\end{align}
If $\kh_\bulk \ge \ell_\bulk$, $\kh_\sigma \ge \ell_\sigma$, and $\lambda_K \le \tilde{g} \scale$,
then for sufficiently small $\tilde{g}$,
\begin{align}
\label{eq:Qh}
	\norm{Q(b)}_{T_{0, y } (\mathfrak{h}, \lambda_K ) }
    \leq O_L  (1) \Big( \frac{\kh_\bulk}{h_\bulk} \Big)^4
    .
\end{align}
\end{lemma}

\begin{proof}
Since $Q$ is a polynomial of degree 4,
Lemma~\ref{lemma:poly_bound} implies that
\begin{align}
	\norm{Q }_{T_{\varphi,y} (\kh, \lambda_K)} \leq \norm{Q }_{T_{0,y} (\kh, \lambda_K)} P_{\kh_\bulk}^4 (\varphi).
\end{align}
It is proved in \cite[Lemma~7.5.1]{BBS-brief} that $\norm{\Tay_k F}_{T_{0,y} (\kh, \lambda_K)} \leq \norm{F}_{T_{0,y} (\kh, \lambda_K)}$ for any $F : \R^n \rightarrow \R$ and any $k\ge 0$.  Therefore,
\begin{align}
	\norm{Q }_{T_{0,y} (\kh,  \lambda_K)}
		& \le \sum_{* \in \{ \bulk, \o,\x, \ox  \} } \big\|  \pi_* e^V K^{\#}  \big\|_{ T_{0,y} (\kh, \lambda_K)}
		= \big\|  e^V K^{\#}  \big\|_{ T_{0,y} (\kh, \lambda_K)}
		\nnb
		& \le
e^{\norm{V(b)}_{T_0 (\kh)} }
        \big( \norm{K}_{T_0 (\kh)} + \lambda_K \big),
\end{align}
where the final inequality uses the product property
\eqref{eq:subx} as well as \eqref{eq:Ksharp}.
This gives \eqref{eq:Qphi}.
The inequality \eqref{eq:Qphi_obs} follows similarly once we observe
that $Q-Q_\bulk$ is a polynomial of degree 1.

Next, since $Q$ is a polynomial of degree $4$,
if we assume $\kh_\bulk \ge \ell_\bulk$ and $\kh_\sigma \ge \ell_\sigma$ then
\begin{align}
	\norm{Q }_{T_{0, y} (\kh, \lambda_K)}
\leq
\Big( \frac{\kh_\bulk}{\ell_\bulk} \Big)^4 \norm{Q }_{T_{0, y} (\ell, \lambda_K)}
\leq
\Big( \frac{\kh_\bulk}{h_\bulk} \Big)^4
\Big( \frac{h_\bulk}{\ell_\bulk} \Big)^4
e^{\norm{V(b)}_{T_0 (\ell)} }  \norm{K^{\#}}_{T_{0, y} (\ell, \lambda_K)}
	.
\end{align}
By Lemma~\ref{lemma:V_bounds}, $\norm{V(b)}_{T_0 (\ell)}$ is bounded by a constant if $\tilde g$ is small depending on $L$, and, by \eqref{eq:hbulk-def},
\begin{equation}
\label{h-ell-ratio}
    \Big( \frac{h_\bulk}{\ell_\bulk} \Big)^4
    =
    \frac{k_0^4}{\ell_0^4} \frac{1}{\tilde g \scale}.
\end{equation}
Also,  for $d >4$,
it follows from \eqref{eq:exbd1} that $\kaa_i >1$ for $i=1,2$, so
$\norm{K}_{T_0 (\ell)} \le \tilde{C}_{\rg} \tilde{g}^3 \tilde{\scale} (\kaa) \le \tilde{g} \scale$.
The same bound is true for $d=4$ since then $\scale = \tilde{\scale} = 1$.
With our assumption on $\lambda_K$,
this proves \eqref{eq:Qh}.
\end{proof}

The next lemma is useful for bounding exponentials of polymer activities
with observables.

\begin{lemma}
\label{lemma:expF}
Suppose  $\norm{F}_{T_{\varphi,y} (\kh, \lambda_K)} < \infty$. Then
\begin{align}
	\norm{e^F}_{T_{\varphi,y} (\kh,\lambda_K)}
		\le \big\| e^{F_{\bulk}} \big\|_{T_{\varphi,y} (\kh, \lambda_K)} \big( 1 + \norm{F-F_{\bulk}}_{T_{\varphi,y} (\kh, \lambda_K)} \big)^2
	.
\end{align}
\end{lemma}

\begin{proof}
Since $(F-F_{\bulk})^3 = 0$, we have
	$e^F = e^{F_{\bulk}} \sum_{k=0}^2 \frac{1}{k !} (F- F_\bulk)^k$,
and the desired bound follows from the product property \eqref{eq:subx}.
\end{proof}

Recall from \eqref{eq:hatV_hatK_definition} that $\hat V = V-Q$
and $\hat{K} = K - (e^{-\hat{V}} - e^{-V})$.
Lemma~\ref{lemma:hat_K_bound}(ii) provides the key stability estimates which give rise to the $\cW^\kappa$-norm's regulator in Lemma~\ref{lem:kappa}.

\begin{lemma}
\label{lemma:hat_K_bound}
Assume \ref{quote:assumPhi}, $j_+ \ge j_\ox$, and $\lambda_K \leq \tilde{g} \scale$.
There is an $L$-independent constant $C$
and an $L$-dependent constant $C_L$ such that the following hold:
\begin{enumerate}
\item
For $B \in \cB_+$,
\begin{align}
	\norm{\hat V (B)}_{T_{0,y} (\mathfrak{h_+}, \lambda_K )}
    & \le \begin{array}{ll}
	\begin{cases}	
	C_L  \tilde{g} \scale  &  (\mathfrak{h} = \ell ) \\	
	C & (\mathfrak{h} = h).
	\end{cases}
	\end{array}
	\label{eq:Vhat_bound0}
\end{align}

\item
Let $\mathfrak{h} \in \{ \ell, h \}$.
There is a constant $\cst >0$
(depending only on $n$)
such that for all $\varphi\in \R^n$ and $b\in \cB$,
and for any $t\geq 0$ and any $s \in [0,1]$,
\begin{align}
	\label{eq:V-Q_bound1-h}
    \norm{ e^{-t(V - s Q) (b)} }_{T_{\varphi, y} (\mathfrak{h}, \lambda_K)}
    	& \le
		C^{t}
		e^{- 5 t\cst |\varphi / h_\bulk |^4}
    \times
    \begin{cases}
    1 & (\{\o,\x\} \cap b =\emptyset)
    \\
    C & (\{\o,\x\} \cap b \neq \emptyset),
    \end{cases}
\end{align}
\begin{align}
\label{eq:V-Q_bound1-hplus}
	\norm{ e^{-t(V - s Q) (b)} }_{T_{\varphi, y} (\mathfrak{h}_+, \lambda_K)}
    &\le	
    	(C^{t} e^{- 5 t\cst |\varphi / h_{\bulk,+} |^4})^{L^{-d}}
    \times
    \begin{cases}
    1 & (\{\o,\x\} \cap b =\emptyset)
    \\
    C & (\{\o,\x\} \cap b \neq \emptyset).
    \end{cases}
\end{align}

\item  For all $\varphi\in \R^n$ and $b\in \cB$,
\begin{align}
	& \norm{\hat{K} (b) }_{T_{0,y} (\ell, \lambda_K)} \le
	C ( \norm{K(b)}_{\Wkappa} + \lambda_K ) ,
	\label{eq:hat_K_bound2}
	\\
	& \norm{\hat{K} (b)}_{T_{\varphi,y} (h, \lambda_K)} \le
	C
	(\norm{K(b)}_{\Wkappa} + \lambda_K ) (\tilde{g}^{9/4} \tilde{\scale} (\tilde{\kb} ) )^{-1}
    e^{- 3 \cst |\varphi / h_{\bulk}|^4}
    .
	\label{eq:hat_K_bound1}
\end{align}
\end{enumerate}
\end{lemma}

\begin{proof}
(i)
This follows from Lemmas~\ref{lemma:V_bounds}--\ref{lemma:Q_bound}.

\smallskip\noindent
(ii)
By monotonicity of the norm (recall
\eqref{eq:ext-norm}
and \eqref{eq:Tnorm-ell-h})
we only have to consider the case $\kh = h$.
Both \eqref{eq:V-Q_bound1-h} and
\eqref{eq:V-Q_bound1-hplus} are proved in \cite[Lemma~8.11(ii)]{MPS23},
with $V - sQ$ replaced by $V_\bulk - s Q_\bulk$, with $C^t=2^{t/4}$,
and with the exponent $5\cst$ instead given by $c'$ (we rename the constant $4\cst$ from \cite[Lemma~8.11(ii)]{MPS23}
as $c'$, it depends only on $n$).
The same therefore hold when $\{ \o,\x \} \cap b = \emptyset$,
so we set our focus on the case $\{ \o,\x \} \cap b \neq \emptyset$.
For the observables, we use
Lemma~\ref{lemma:expF},  \eqref{eq:Vz_bound}, and \eqref{eq:Qphi_obs}
to see that,
whenever $\{ \o ,  \x \} \cap b \neq \emptyset$,
there exists $C >0$ such that
\begin{align}
	\norm{ e^{-t ( (V-V_\bulk) - s (Q - Q_\bulk) ) (b) } }_{T_{\varphi,y} (h, \lambda_K)}
		& \le
		\big( 1 + t \norm{ ((V-V_\bulk) - s (Q - Q_\bulk) ) (b) }_{T_{\varphi,y} (h, \lambda_K)} \big)^2 \nnb
		&  \le (1+  C t P_{h_{\bulk} } ( \varphi ))^2   \nnb
		&  \le \Big( 1 +  \frac{C}{c' / 4} \Big)^2   e^{ \frac{1}{2} t c' P_{h_\bulk} (\varphi) }
		\label{eq:V-Q_bound1-h_proof1}
    ,
\end{align}
where the final inequality holds by first replacing $C$ by $\frac{1}{4} c'$ using a multiplicative factor $( 1 +  \frac{C}{c' / 4} )^2$ and then bounding $1 + \frac{1}{4} t c' P_{h_\bulk}$ by $e^{\frac{1}{4} t c' P_{h_\bulk}}$.
Hence,
\begin{align}
    \norm{ e^{-t(V - s Q) (b)} }_{T_{\varphi, y} (\kh, \lambda_K)}
    	& \le
		2^{t/4}
		e^{- t c' |\varphi / h_\bulk |^4}
    	C e^{\frac{1}{2} t c' P_{h_\bulk} (\varphi)}
		\le C^{1+t} e^{-\frac{1}{2} t c' |\varphi / h_\bulk |^4}
	.
\end{align}
This proves \eqref{eq:V-Q_bound1-h} with $\cst = \frac{1}{10}  c'$.

For \eqref{eq:V-Q_bound1-hplus}, we proceed similarly.
If $b$ contains neither $\o$ nor $\x$, then the observables play no role,
and \eqref{eq:V-Q_bound1-hplus} is an immediate consequence of
\cite[Lemma~8.11(ii)]{MPS23}.
Thus we can assume that $b$ contains at least one of $\o,\x$.
In this case, we recall from Lemma~\ref{lemma:V_bounds} that $\norm{\sigma_\z V_\z (b)}_{T_{\varphi} (h_+)} \le C L^{-\frac{d}{4}} P_{h_{\bulk,+}} (\varphi)$
(since $j_+ \ge j_\ox$).
Also, since $h_{\bulk,+} \le h_\bulk$, if we take into account the
change of scale in the norm of the observables using \eqref{eq:h-ratios} then
we find that
$\norm{K}_{T_{0} (h_+)}\le 4 \norm{K}_{T_{0} (h)}  \le 1$.
With \eqref{eq:Qphi_obs}, these bounds give,
for $\{ \o,\x \} \cap b \neq \emptyset$,
\begin{align}
	\norm{(V - V_\bulk) (b)}_{T_{\varphi,y} (h_+, \lambda_K)}
		& \le
		C  L^{-\frac{d}{4}} P_{h_{\bulk,+}} (\varphi)
,
			\label{eq:V-Vbulk}
		 \\
	\norm{(Q- Q_\bulk) (b)}_{T_{\varphi,y} (h_+, \lambda_K)}
		& \le C \tilde{g} \scale P_{h_{\bulk,+}} (\varphi)
		\le  C L^{-\frac{d}{4}} P_{h_{\bulk,+}} (\varphi)
		,
			\label{eq:Q-Qbulk}
\end{align}
where the final inequality of \eqref{eq:Q-Qbulk} holds for sufficiently small $\tilde{g}$.
Then, using the same steps used to derive \eqref{eq:V-Q_bound1-h_proof1},
Lemma~\ref{lemma:expF} gives
\begin{align}
	\norm{ e^{-t ( (V-V_\bulk) - s (Q - Q_\bulk) ) (b) } }_{T_{\varphi,y} (h_+, \lambda_K)}
	& \leq C^{1+t L^{-d}} e^{\frac{1}{2} t L^{-d} c'  |\varphi / h_{\bulk,+} |}
	.
\end{align}
Thus by multiplying with the bound on $e^{-t (V_{\bulk} - s Q_{\bulk})}$,
we have
\begin{align}
	\norm{ e^{-t(V - s Q) (b)} }_{T_{\varphi, y} (\mathfrak{h}_+, \lambda_K)}
    &\le	
    	C (C^{t} e^{-  \frac{1}{2} t c' |\varphi / h_{\bulk,+} |^4})^{L^{-d}}
    	.
\end{align}
Due to our choice $\cst=\frac{1}{10}c'$, this gives
\eqref{eq:V-Q_bound1-hplus}.

\smallskip\noindent
(iii)
Since $K \in \cK$ (by assumption \ref{quote:assumPhi}), it suffices to bound
the difference $\hat{K} - K$, which is
\begin{align}
\label{eq:ehatV}
    \hat{K} - K =
	e^{-\hat{V}} - e^{-V} = Q \int_0^1 e^{-V + s Q} ds.
\end{align}
Then \eqref{eq:hat_K_bound2}
follows directly from the bounds on $Q$ and $e^{-V + sQ}$ in  Lemmas~\ref{lemma:V_bounds}--\ref{lemma:Q_bound}.
For \eqref{eq:hat_K_bound1},
we use Lemma~\ref{lemma:Q_bound} and \eqref{eq:V-Q_bound1-h} to bound the norm
of the right-hand side of \eqref{eq:ehatV} as
\begin{align}
	\big\| e^{-\hat{V}} - e^{-V} \big\|_{T_{\varphi, y} (h, \lambda_K)}
		& \le C \big( \norm{K}_{T_0 (h)} + \lambda_K \big) P_{h_{\bulk}}^4 (\varphi) e^{- 4 \cst |\varphi / h_{\bulk}|^4}.
\end{align}
We absorb the polynomial into the exponential by adjusting the latter's
coefficient from $4$ to $3$.
Also, by definition of the norm, $\norm{K (b)}_{T_0 (h)} \le \norm{K}_{\Wkappa} (\tilde{g}^{9/4} \tilde\scale (\kb) )^{-1}$.
This proves \eqref{eq:hat_K_bound1}.
\end{proof}

The next lemma collects estimates for easy reference.
For the regulator $G$ defined in \eqref{eq:Gj},
we fix $\kappa=\kappa(n)$ to be any value that obeys
\begin{equation}
\label{eq:kappa-def}
    0< \kappa \le
    \cst
	.
\end{equation}

\begin{lemma}
\label{lem:kappa}
Under the assumptions of Lemma~\ref{lemma:hat_K_bound},
for $\kh  \in \{\ell, h\}$ and $\kh_+  \in \{\ell_+, h_+\}$, we have
\begin{align}
	\norm{e^{-V (b)}} _{T_{\varphi} (\kh)}
		& \le C G^2 (b, \varphi) ,
	\label{eq:4.6-1}\\
	\norm{e^{-V (b)}} _{T_{\varphi} (\kh_+)}
		& \le C^{L^{-d}} G^2 (b, \varphi) \times \begin{cases}
		1 & (\{ \o,\x\} \cap b = \emptyset) \\
		C & (\{ \o,\x\} \cap b \neq \emptyset)
		,
	\end{cases}
	\label{eq:eV4.6-old}
\\
	\label{eq:4.6-3}
	\norm{(e^{-V} V^k) (b)}_{T_{\varphi} (\kh)} & \le O_k (1)	
	.
\end{align}
\end{lemma}

\begin{proof}
For \eqref{eq:4.6-1}, we take
$(t,s) = (1,0)$
in \eqref{eq:V-Q_bound1-h} and use the assumption $\kappa \le \cst$.
For \eqref{eq:eV4.6-old},
we take $(t,s) = (1,0)$ in \eqref{eq:V-Q_bound1-hplus} and
use $h_\bulk / h_{\bulk,+} \leq L^{\frac{d}{4}}$ to replace $L^{-d} |\varphi / h_{\bulk,+} |^4$ by $|\varphi /h_{\bulk}|^4$ in the bound.
Then the conclusion follows from $\kappa \le \cst$.
For \eqref{eq:4.6-3}, by \eqref{eq:Tnorm-ell-h} it suffices
to consider the case $\kh=h$.  For this, we use
the product property \eqref{eq:subx} of the norm, \eqref{eq:4.6-1},
and \eqref{eq:V_bound0}, to see that
\begin{align}
    \norm{(e^{-V} V^k) (b)}_{T_{\varphi} (h)}
    \le
    C G^2 (b, \varphi)  C^k P^{4k}_{h_\bulk} (\varphi)
    \le O_k (1) ,
\end{align}
where in the last inequality we used the regulator to control the polynomial.
\end{proof}

\begin{lemma}
\label{lemma:stability_estimate_obs}
Let $k_0$ be sufficiently small, let $\lambda_K \leq \tilde{g} \scale$, and assume \ref{quote:assumPhi}.  There is an $L$-independent constant $C$ such that for
all $B \in \cB_+$,
$\kh_+ \in \{ \ell_+, h_+ \}$,
and all $t \ge 0$,
\begin{align}
	\exp\big(  \norm{\delta \bar{u}_{+} (B)}_{T_{\varphi, y} (\mathfrak{h}_+, \lambda_K)}  \big)
		&  \leq C,
	\\
	\norm{ (\delta\bar u_{+} - \delta\bar u_{\pt} )(B) }_{T_\varphi (\kh_+) }
    	& \leq
    	O_L(\norm{K}_{\Wkappa}) ,
    \\
\label{eq:stability_estimate_obs}
	\big\| e^{-t U_+ (B)} \big\|_{T_{\varphi, y} (\mathfrak{h}_+, \lambda_K)}
    	& \leq C^{t} e^{ -2 t \cst |\varphi/   h_{\bulk,+} |^4 }
    	\times \begin{cases}
    		1 & ( \{\o,\x\} \cap b = \emptyset) \\
    		C & ( \{\o,\x\} \cap b \neq \emptyset)
    		.
    	\end{cases}
\end{align}
\end{lemma}

\begin{proof}
It is proved in \cite[Lemma~8.12]{MPS23} that, for the bulk,
\begin{align}
	\exp\big(  \norm{\delta \bar u_{+,\bulk} (B) }_{T_{\varphi, y} (\mathfrak{h}_+, \lambda_K)}  \big)
    & \leq 2^{1/2},
	\label{eq:vacuum_energy_estimate_MPS}
	\\
	\big| ( \delta \bar u_{+,\bulk} - \delta \bar u_{\pt, \bulk} ) (B) \big|
    & \leq O_L \big(
    \norm{K}_{\Wkappa} \big)
	.
	\label{eq:u_+_minus_u_pt_MPS}
\end{align}
It therefore remains to prove that
\begin{align}
	& \exp\big(  \norm{\sigma_\ox \delta {u}_{+,\ox} }_{T_{\varphi, y} (\mathfrak{h}_+, \lambda_K)}  \big) \leq C,
	\label{eq:deltaq_norm}
	\\
	& \norm{\sigma_\ox }_{T_\varphi (\kh_+) }  | \delta u_{+, \ox} - \delta u_{\pt,\ox} |
    \leq
    \norm{K}_{\Wkappa}
    ,
	\label{eq:deltaq_minus_qpt_norm}
\end{align}
as well as \eqref{eq:stability_estimate_obs}.

For \eqref{eq:deltaq_norm}--\eqref{eq:deltaq_minus_qpt_norm},
by Lemma~\ref{lemma:Uplus}
we know that
$\delta u_{\pt,\ox} = C_+ (\x)$ and $\delta u_{+,\ox} = \delta u_{\pt,\ox} + K_{\ox} (b_{\ox},0)$.
Thus it follows from the monotonicity of the $T_{\varphi, y} (\kh_+, \lambda_K)$-norm that, for small $\tilde{g}$,
\begin{align}
	\norm{ \sigma_\ox ( \delta u_{+,\ox} - \delta u_{\pt,\ox} ) }_{T_{\varphi, y} (\kh_+, \lambda_K)}  & \le  \norm{K^{\#}}_{\Wkappa_y (\lambda_K)}  \le 1
,
	\\
	\norm{ \sigma_\ox \delta u_{\pt,\ox}  }_{T_{\varphi, y} (\kh_+, \lambda_K)}  &  \le 2 L^{-(d-2)j} h_{\sigma,+}^2    \le 1
	,
\end{align}
where we have used the definition \eqref{eq:h-sig-alt} of $h_{\sigma,+}$
and $d \ge 4$ in the last line.
This allows us to conclude \eqref{eq:deltaq_norm}--\eqref{eq:deltaq_minus_qpt_norm}.

For \eqref{eq:stability_estimate_obs}, we already have a bound on $e^{-tV}$ from Lemma~\ref{lemma:hat_K_bound},
so it suffices to bound $e^{-t (U_+ - V)}$.
But again by Lemma~\ref{lemma:Uplus},
we have
\begin{align}
	( U_+ - V )(B) &= (U_{\bulk,+} - V_{\bulk}) (B)
    -
    \sigma_\ox\delta \bar{u}_{+,\ox} (B)
    ,
\end{align}
so together with \eqref{eq:deltaq_norm},
we now have
\begin{align}
	\norm{ e^{- t ( U_+ - V )(B)}}_{T_\varphi (\kh_+)}
		\le C^{t} \norm{ e^{ - t(U_{\bulk,+} - V_\bulk) (B) }}_{T_\varphi (\kh_+)}
	.
\end{align}
A sufficient bound on $e^{ - t(U_{\bulk,+} - V_\bulk) (B) }$ is given in
the first paragraph of the proof of \cite[Lemma~8.12]{MPS23}, so we obtain \eqref{eq:stability_estimate_obs}.
\end{proof}

\begin{proof}
[Proof of Propositions~\ref{prop:S0_obs}, \ref{prop:S1_obs}
and Lemma~\ref{lemma:crucial_contraction_pre_obs}.]

The ingredients provided by
Lemmas~\ref{lemma:V_bounds}--\ref{lemma:stability_estimate_obs}
permit the proofs of Propositions~\ref{prop:S0_obs}, \ref{prop:S1_obs} and
Lemma~\ref{lemma:crucial_contraction_pre_obs}
to be followed line by line by comparison with the proofs of
\cite[Proposition~9.2, 9.4, Lemma~9.8]{MPS23},
now with the inclusion of the observables.
What the proofs require are the properties of the norm (which persist for the norms we use here),
\cite[Lemmas~8.7,  8.9, 8.11, 8.12]{MPS23}, and the bound
\begin{align}
	\norm{Q(b)}_{T_{\varphi} (\kh)} \leq C  \norm{K(b)}_{T_0( \kh)}  P_{\kh_\bulk}^4 (\varphi)
	\label{eq:Qb_bound}
\end{align}
for both $\kh \in \{ \ell, h \}$
(see the line above \cite[(9.91)]{MPS23}).
More explicitly,
bounds on the effective potential, \cite[Lemma~8.9 and (8.62)]{MPS23} correspond to (and are extended by) Lemma~\ref{lemma:V_bounds}.
The analogue of \cite[Lemma~8.7]{MPS23}, a bound on $Q = \Loc (e^V K)$, is given by Lemma~\ref{lemma:Q_bound}.
The compilation of bounds on $\hat{V}$, $e^{-t (V - sQ)}$ and $\hat{K}$ in
\cite[Lemma~8.11]{MPS23},
corresponds to Lemma~\ref{lemma:hat_K_bound}.
The stability estimate \cite[Lemma~8.12]{MPS23} has its analogue in Lemma~\ref{lemma:stability_estimate_obs}.
Finally,
\eqref{eq:Qb_bound} is provided by Lemma~\ref{lemma:Q_bound} with $\norm{V(b)}_{T_0 (\kh)}$ bounded by \eqref{eq:V_bound0}.
A superficial difference is that the right-hand sides of
\cite[(8.64), (8.65)]{MPS23} do not have extra constant $C$
unlike the $\{\o,\x\} \cap b \neq \emptyset$ case of \eqref{eq:V-Q_bound1-h} and \eqref{eq:V-Q_bound1-hplus}.
These only affect \cite[(9.31), (9.37), (9.91)]{MPS23} by extra multiplicative constants,
which leave the proofs undamaged.
\end{proof}

\section{Covariance calculations: proof of Lemma~\ref{lem:uNox}}
\label{app:covariance}

In this appendix, we prove the statements about the hierarchical covariances
that lead to Lemma~\ref{lem:uNox},
to \eqref{eq:Cpa-plateau}, and to the two statements claimed above Corollary~\ref{cor:Gaussian}.  When $\x=\o$, inverse powers of $|\x|$ should
be interpreted as $1$; we opt not to clutter the notation to deal with
this special case.

We begin with the following lemma, which involves the metric on $\R$ defined for $r \le t$ by
\begin{equation}
    d_{\rm sq} (r,t) = \frac{1}{2}  \int_r^t |y|^{-1/2} dy
    =
    \begin{cases}
    \big| \sqrt{|r|} - \sqrt{|t|} \,\big| & (rt \ge 0)
    \\
    \sqrt{|r|} + \sqrt{|t|} & (rt < 0).
    \end{cases}
\end{equation}
Recall that $\gamma_j(\ka) = L^{2 (j-1)}(1+\ka L^{2(j-1)})^{-1}$
is defined in \eqref{eq:gamma_j}.

\begin{lemma}
\label{lem:Capp-prep}
For $d\ge 4$, for $j \le N$, and for $\ka_1, \ka_2 \ge -\frac{1}{2} L^{-2(N-1)}$,
\begin{align}
	| \gamma_j (\ka_1) - \gamma_j (\ka_2) | \le O( L^{3 (j-1)} ) d_{\rm sq} (\ka_1, \ka_2)
\end{align}
\end{lemma}

\begin{proof}
We differentiate the formula for $\gamma_j$ in \eqref{eq:gamma_j} to obtain
\begin{align}
	\left| \frac{d \gamma_j (\ka)}{d \ka} \right|
    =
    \frac{L^{4(j-1)}}{(1+\ka L^{2(j-1)})^2}
		=  \frac{( |\ka| L^{2(j-1)})^{1/2} }{(1 + \ka L^{2(j-1)})^2 }
    L^{3(j-1)} |\ka|^{-1/2}.
\end{align}
The fraction on the right-hand side is bounded above by
$\sup_{t \ge - \frac{1}{2}} |t|^{1/2} (1+ t)^{-2} < \infty$, and
the desired bound then follows by integration.
\end{proof}

\begin{lemma}
\label{lem:Capp}
For $d \ge 4$, for $\ka,\ka_1,\ka_2 \in [-\frac 12 L^{-2(N-1)}, A L^{-2N}]$,
and for $x \in \Lambda_N$, the following bounds hold
(constants may depend on $A$):
\begin{align}
\label{eq:Ca0}
    |C_{\ka_1,\le N}(x) - C_{\ka_2,\le N}(x)|
    & \le  O_L \left( \frac{d_{\rm sq}(\ka_1, \ka_2)}{|x|^{d-3}} \right)
\\
\label{eq:CNinf}
    |\HLap_{0,\infty}(x) - C_{0,\le N}(x)|
    & \le O (L^{-(d-2)N}),
\\
\label{eq:CNasy}
    C_{\ka,\le N}(x) + \gamma_N(\ka)L^{-dN} &\asymp \frac{1}{|x|^{d-2}}
    .
\end{align}
In addition, \eqref{eq:Cpa-plateau} holds, i.e., if
$\ka=tL^{-(2+\delta)N}$ with $t,\delta >0$,
then $\HLap^\per_{\ka,N}(\x) \asymp \frac{1}{|x|^{d-2}} + \frac{L^{\delta N}}{tL^{(d-2)N}}$.
\end{lemma}

\begin{proof}
We start by writing
\begin{align}
\label{eq:CQ}
    C_{\ka,\le N}
    & = \sum_{j=1}^N \gamma_j(\ka) P_j
    =
    \gamma_1(\ka) Q_0 + \sum_{j=1}^{N-1} (\gamma_{j+1}(\ka)-\gamma_j(\ka)) Q_j
    -\gamma_N(\ka) Q_N .
\end{align}
By hypothesis and Lemma~\ref{lem:Capp-prep}
 ($A$-dependence arises from the
denominator of $\gamma_j$),
\begin{equation}
	|\gamma_j(\ka)| \asymp L^{2(j-1)},
    \quad\;
    | \gamma_j (\ka_1) - \gamma_j (\ka_2) | \le O( L^{3 (j-1)} ) d_{\rm sq} (\ka_1, \ka_2),
    \quad\;
    0 < \gamma_{j+1}(\ka)-\gamma_j(\ka) \asymp L^{2j}.
\end{equation}
Since $C_j(\x)=0$ for $j < j_\ox$,
these bounds immediately lead to \eqref{eq:Ca0}--\eqref{eq:CNasy}:
\begin{align}
    |C_{\ka_1,\le N}(x) - C_{\ka_2,\le N}(x)|
    & \le O(d_{\rm sq} (\ka_1, \ka_2)) \sum_{j=j_\ox}^N  L^{3 (j-1)} L^{-d (j-1)}
    \le O_L \left( \frac{ d_{\rm sq} (\ka_1, \ka_2) }{|x|^{d-3}} \right),
    \label{eq:Ca0-2}
\\
    | \HLap_{0,\infty}(x) - C_{0,\le N}(x)|
    &
    \le O(1) \sum_{j=N+1}^\infty L^{2 (j-1)} L^{-d (j-1)} \le O(L^{-(d-2)N})
    ,
\end{align}
\begin{equation}
    C_{\ka,\le N}(x) + \gamma_N(\ka)L^{-dN}
    	= \sum_{j=j_\ox}^{N-1} (\gamma_{j+1} (\ka) - \gamma_j(\ka) ) L^{-d(j-1 )}(1-L^{-d})
     	\asymp \frac{1}{|x|^{d-2}}.
\end{equation}
Finally, we set $\ka = tL^{-(2+\delta)N}$ and use  \eqref{eq:CNasy} to see that
\begin{align}
    \HLap^\per_{\ka,N}(\x)
    &= \Big(C_{\ka,\le N}(\x)+\gamma_N(\ka) L^{-dN} \Big) + (\ka^{-1}-\gamma_N(\ka))L^{-dN}
    \nnb & =
    \Big(C_{\ka,\le N}(\x)+\gamma_N(\ka) L^{-dN} \Big) + \frac{L^{-dN}}{a(1+aL^{2(N-1)})}
    \asymp
    \frac{1}{|x|^{d-2}} + \frac{1}{aL^{dN}}.
\end{align}
This proves \eqref{eq:Cpa-plateau} and completes the proof.
\end{proof}

For the next lemma, we recall $\ka_N^*$ and $\tilde{\ka}_N^*$ from \eqref{eq:aN1}--\eqref{eq:aN2}.
We do not track $L$-dependence of the constants explicitly.
The estimates in Lemma~\ref{lem:HLapN-app} are not uniform in $s$.

\begin{lemma}
\label{lem:HLapN-app}
Let $\cO_{N,\x}$ is as in Definition~\ref{def:errorconv}.
For $\x\in\Lambda_N$, as $N \to\infty$,
there exist $|r_N|, |r'_N| \le O(|x|^{-(d-3)} L^{-N})$ such that
the following statements hold.
For $s\in \R$,
\begin{align}
\label{eq:HLaph-app}
    C_{a_N^\per(s), \le N} (\x)
    & = \HLap_{0,\infty}(\x) \big( 1+ \cO_{N,\x} \big) + r_N
    ,
    \\
\label{eq:CFBCh-app}
    C_{a_N^\free(s), \le N}(\x)
    & = \HLap_{0,\infty}(\x) \big( 1+ \cO_{N,\x} \big) + r'_N
    .
\end{align}
For $s_N >0$ such that $s_N \rightarrow s$ as $N\rightarrow \infty$,
\begin{align}
\label{eq:HLapl-app}
    C_{\tilde a_N^\per(s_N ), \le N}(\x)
    + \frac {1}{s_{N}} L^{-(d-2)N} &
    = \big( 1+ \cO_{N,\x} \big)
	\HLap_{sL^{-2N},N}^\per(\x)
    \asymp \frac{1}{|\x|^{d-2}} ,
\\
\label{eq:CFBC-app}
    C_{\tilde a_N^\free(s_N ), \le N}(\x) + \frac {1}{s_{N}} L^{-(d-2)N}
    &= \big( 1+ \cO_{N,\x} \big)
	\HLap_{(s-\qLap)L^{-2N},N}^\free(\x)
    \asymp \frac{1}{|\x|^{d-2}}  .
\end{align}
\end{lemma}

Lemma~\ref{lem:uNox} is a direct consequence of Lemma~\ref{lem:HLapN-app},
as follows.

\begin{proof}[Proof of Lemma~\ref{lem:uNox}]
Let $s\in\R$.
By \eqref{eq:HLaph-app} and \eqref{eq:CFBCh-app} with \eqref{eq:qN-thm},
\begin{align}
	u_{N;\ox}(a_N^\per(s)) & =
	\HLap_{0,\infty}(\x)  (1+ \cO_{N,\x} ) + \vv (\x )
            + O( |x|^{-(d-3)} L^{-N})
    ,
    \\
	u_{N;\ox}(a_N^\free(s)) & =
	\HLap_{0,\infty}(\x)  (1+ \cO_{N,\x} ) + \vv (\x )
        	+ O( |x|^{-(d-3)} L^{-N})
    ,
\end{align}
where we used $\HLap_{0,\infty}(\x) \asymp \frac{1}{|x|^{d-2}}$ to replace $|x|^{-(d-2)} \cO_{N,\x}$ by $\HLap_{0,\infty}(\x) \cO_{N,\x}$.
These give \eqref{eq:uNox-PBC} and \eqref{eq:uNox-FBC}.

Similarly, \eqref{eq:uNox-PBC-tilde} and \eqref{eq:uNox-FBC-tilde} follow immediately from \eqref{eq:qN-thm} and \eqref{eq:HLapl-app}--\eqref{eq:CFBC-app}.
\end{proof}

\begin{proof}[Proof of Lemma~\ref{lem:HLapN-app}]
We prove the four statements one by one.
Recall from \eqref{eq:aN1}--\eqref{eq:aN2} that
\begin{align}
\label{eq:aN1-app}
    a_N^\per(s) \sim s\hh_N^{-2}L^{-dN}
    \asymp \begin{cases}
    s N^{-1/2}L^{-2N} & (d=4)
    \\
    s L^{-dN/2} & (d>4)
    \end{cases}
    ,
    \qquad a_N^\free(s) \sim s\hh_N^{-2}L^{-dN} -\qLap L^{-2N},
\end{align}
\begin{align}
\label{eq:aN2-app}
    \tilde a_N^\per(s) \sim sL^{-2N}, \qquad \tilde a_N^\free(s) \sim (s-\qLap) L^{-2N}	.
\end{align}
In particular, this shows that $\ka_N^*$ and $\tilde{\ka}_N^*$ are
within the domain of Lemma~\ref{lem:Capp}, by taking $A$ large depending on $s$.
This is one source of the $s$-dependence of the error terms.

\medskip\noindent \emph{Proof of \eqref{eq:HLaph-app}.}
It suffices to prove that
\begin{equation}
    C_{a_N^\per(s),\le N}(\x) - \HLap_{0,\infty}(\x)= |x|^{-(d-2)}   \cO_{N,\x}
    + O(L^{-(d-2)N}).
\end{equation}
By \eqref{eq:CNinf}, and since $C_{0,\infty}(x) = \HLap_{0,\infty}(x) \asymp |x|^{-(d-2)}$ by \eqref{eq:HLap0-asy},
it therefore suffices to prove that
$C_{a_N^\per(s) ,\le N}(\x) - C_{0,\le N}(\x)= |x|^{-(d-2)} \cO_{N,\x}$.
But this follows from \eqref{eq:Ca0},
which gives
\begin{align}
	\label{eq:HLaph-app2}
	| C_{a_N^\per(s) , \le N}(\x) - C_{0,\le N}(\x) |
    & \le
    O( |a_N^\per (s)|^{1/2} )  \frac{1}{|x|^{d-3}}
	\nnb
    & \le
    O\Big( \frac{1}{ L^N |x|^{d-3}} \Big)
    \times
    \begin{cases}
    N^{-1/4} & (d=4)
    \\
    L^{-N(d-4)/4 } & (d>4) .
    \end{cases}
\end{align}

\medskip\noindent \emph{Proof of \eqref{eq:CFBCh-app}.}
As in the proof of \eqref{eq:HLaph-app}, it suffices to prove that
\begin{equation}
    C_{a_N^\free(s),\le N}( \x) - C_{0,\le N}(\x)= |x|^{-(d-2)} \cO_{N,\x} + O(L^{-(d-2) N}).
\end{equation}
We decompose the left-hand side as
\begin{equation}
    [C_{a_N^\free(s_{N}) ,\le N}(\x) - C_{-qL^{-2N},\le N}(\x)]
    + [C_{-qL^{-2N} ,\le N}(\x) - C_{0,\le N}(\x)].
\end{equation}
For the first difference in the decomposition,
by \eqref{eq:aN1-app} the two masses essentially differ by $\ka_N^\per(s)$,
so it is bounded by
\begin{equation}
    O(|\ka_N^\per(s)|^{1/2} ) \frac{1}{|x|^{d-3}}.
\end{equation}
This is the same as what we encountered in the proof of \eqref{eq:HLaph-app},
so we obtain the same bound for this term.
For the second difference in the decomposition, by \eqref{eq:Ca0} we
obtain the required bound:
\begin{align}
    |C_{-qL^{-2N} ,\le N}(\x) - C_{0,\le N}(\x) |
    \le
    O\Big( \frac{1}{L^N |x|^{d-3}} \Big) .
\end{align}

\medskip\noindent \emph{Proof of \eqref{eq:HLapl-app}.}
By the formula for $\HLap_{\ka,N}^\per$ given by \eqref{eq:ChatN} and \eqref{eq:Csum},
it suffices to prove that
\begin{align}
\label{eq:HLapl-pf}
    C_{\tilde a_N^\per(s_N ) ,\le N}(\x)
    + \frac{1}{s_N} L^{-(d-2)N} &
    = \big( 1+  \cO_{N,\x}  \big) \Big[C_{sL^{-2N},\le N}(\x) + \frac 1s L^{-(d-2)N} \Big]
    \asymp \frac{1}{|\x|^{d-2}}  .
\end{align}
The second statement (equivalence to  $|x|^{-(d-2)}$)
follows from \eqref{eq:CNasy} and
\begin{align}
    0 \le \left( s^{-1}  L^{2N} - \gamma_N(sL^{-2N})\right) L^{-dN} =
    \frac{1}{s(1+sL^{-2})}L^{-(d-2)N} \le c_s L^{-(d-2)N}.
\end{align}
For the equality in \eqref{eq:HLapl-pf},
similarly to the proof of \eqref{eq:CFBCh-app} we now obtain
\begin{align}
    |  C_{\tilde a_N^\per(s_N ) ,\le N}(\x)
    - C_{sL^{-2N},\le N}(\x)|
    \le
    \frac{ \cO_{N,\x} }{L^N |x|^{d-3}} \le \frac{1}{|x|^{d-2}} \cO_{N,\x} ,
\end{align}
as required.

\medskip\noindent \emph{Proof of \eqref{eq:CFBC-app}.}
By the formula for $\HLap_{\ka,N}^\free$ given by \eqref{eq:ChatN} and \eqref{eq:Csum},
it suffices to prove that
\begin{align}
\label{eq:CFBC-pf}
    C_{\tilde a_N^\free(s_N ) ,\le N}(\x) + \frac{1}{s_N} L^{-(d-2)N}
    &= (1+  \cO_{N,\x}   )
    \Big[ C_{(s-\qLap)L^{-2N} ,\le N}(\x) + \frac 1s L^{-(d-2)N} \Big]
    \asymp \frac{1}{|\x|^{d-2}} .
\end{align}
Again there are two statements.  The asymptotic equivalence follows as in the
proof of \eqref{eq:HLapl-app} from
(after some algebra)
\begin{equation}
    s^{-1}L^{-(d-2)N} -\gamma_N((s-\qLap)L^{-2N}) L^{-dN}
    =
    L^{-(d-2)N} \frac{1}{s} \frac{1}{1+sL^{-2}(1-\qLap L^{-2})^{-1}},
\end{equation}
which is positive and $O(L^{-(d-2)N})$.
The proof that $C_{\tilde a_N^\free(s_N ) ,\le N}(\x)
= (1+\cO_{N,\x}) C_{(s-\qLap)L^{-2N},\le N}(\x)$ is also similar to the proof of the
corresponding statement for PBC in the proof of \eqref{eq:HLapl-app}, now
using the fact that $\tilde a_N^\free(s_N) = (s-\qLap + \cO_N ) L^{-2N}$ by
\eqref{eq:aN2-app}.
\end{proof}

Finally, we prove the two statements claimed above Corollary~\ref{cor:Gaussian}.

\begin{lemma} \label{lemma:corGaussianLemma}
For $s >0$ and $\x \in \Lambda_{\infty}$,
\begin{align}
	\HLap_{0,\infty} (\x)
		=\lim_{N\rightarrow \infty} \HLap^\per_{sL^{-2N} , N} (\x)
		=\lim_{N\rightarrow \infty} \HLap^\free_{(s-q) L^{-2N} , N} (\x)	
    ,
		\label{eq:corGaussianLemma1}
\end{align}
and, for $N$ sufficiently large,
\begin{align}
	\HLap^\per_{sL^{-2N} , N} (\x) \asymp \HLap^\free_{(s-q) L^{-2N} , N} (\x)
    \asymp
    \frac{1}{|x|^{d-2}} .
	\label{eq:corGaussianLemma2}	
\end{align}
\end{lemma}

\begin{proof}
The reasoning is almost identical for $\HLap^\per$ and $\HLap^\free$,  so we only
present the proof for $\HLap^\per$.
For \eqref{eq:corGaussianLemma1}, we first apply Lemma~\ref{lem:Capp} to obtain
\begin{align}
	|\HLap^\per_{sL^{-2N} , N} (\x) - \HLap_{0,\infty} (\x)|
		&\le | C_{0, \le N} (\x) - \HLap_{0,\infty} (\x) | + | C_{s L^{-2N} , \le N} (\x)  -  C_{0, \le N} (\x)  | + (s L^{-2N})^{-1} L^{-dN} \nnb
		&\le O_{s,L}(1) \Big( L^{-(d-2) N} + \frac{L^{-N}}{|\x|^{d-3}}  \Big).
\end{align}
The right-hand side goes to zero as $N \to \infty$, so this proves
\eqref{eq:corGaussianLemma1}.
For \eqref{eq:corGaussianLemma2}, we use
\begin{align}
	\HLap^\per_{sL^{-2N} , N} (\x)  &= \big( C_{sL^{-2N} ,  \le N} (\x) + \gamma_N ( sL^{-2N} ) L^{-dN} \big) + ( s^{-1} L^{2N} - \gamma_N ( sL^{-2N} ) ) L^{-dN}  \nnb
	&= \big( C_{sL^{-2N} ,  \le N} (\x) + \gamma_N ( sL^{-2N} ) L^{-dN} \big) + \frac{L^{-(d-2)N}}{s(1 + s L^{-2})}
	\asymp
    \frac{1}{|x|^{d-2}} ,
\end{align}
where the final relation is given by \eqref{eq:CNasy}.
\end{proof}

\section*{Acknowledgements}
We thank Emmanuel Michta for useful discussions at the outset of this work.
The work of GS was supported in part by NSERC of Canada.

\end{document}